\documentclass[12pt]{article}
\usepackage[utf8]{inputenc}
\usepackage{geometry}
\usepackage{comment}
\usepackage{bbm}
\usepackage{float}
\usepackage{hyperref}
\usepackage{enumitem}
\usepackage{graphicx, xcolor, tcolorbox} 
\usepackage{caption, subcaption}
\usepackage[linesnumbered,ruled,vlined]{algorithm2e}
\usepackage{url}
\usepackage{natbib}
\RestyleAlgo{ruled}
\SetKwInput{KwBegin}{Begin}
\SetKwInput{KwReturn}{Return}
\SetKwComment{Comment}{/* }{ */}
\title{Compositional Covariate Importance Testing via Partial Conjunction of Bivariate Hypotheses}
\author{
Ritwik Bhaduri\and 
Siyuan Ma \and
Lucas Janson
}
\date{}
\usepackage{amsmath, amssymb, mathrsfs}

\makeatletter

\newcommand*{\centernot}{%
  \mathpalette\@centernot
}
\def\@centernot#1#2{%
  \mathrel{%
    \rlap{%
      \settowidth\dimen@{$\m@th#1{#2}$}%
      \kern.5\dimen@
      \settowidth\dimen@{$\m@th#1=$}%
      \kern-.5\dimen@
      $\m@th#1\not$%
    }%
    {#2}%
  }%
}
\makeatother
\usepackage{parskip}

\newcommand{\indep}{\perp\mkern-10mu\perp}
\newcommand{\nindep}{\centernot{\indep}}

\renewcommand*{\H}{\mathcal{H}}

\newcommand*{\R}{\mathbb{R}}
\renewcommand*{\P}{\mathbb{P}}
\newcommand*{\E}{\mathbb{E}}

\newcommand*{\I}{\mathcal{I}}
\renewcommand*{\S}{\mathcal{S}}
\newcommand*{\M}{\mathcal{M}}
\newcommand*{\1}{\mathbbm{1}}
\newcommand*{\D}{\mathcal{D}}

\newcommand{\PDense}{|\mathcal{D}|}
\newcommand{\Dense}{\mathcal{D}} 

\usepackage{xparse}
\NewDocumentCommand{\ceil}{s O{} m}{%
  \IfBooleanTF{#1} 
    {\left\lceil#3\right\rceil} 
    {#2\lceil#3#2\rceil} 
}
\usepackage{amsthm}

\newtheorem{Theorem}{Theorem}[section]
\newtheorem{corollary}{Corollary}[section]
\newtheorem{Definition}{Definition}[section]

\newtheorem{proposition}{Proposition}[section]
\newtheorem{Lemma}{Lemma}[section]

\newtheorem{remark}{Remark}
\newtheorem{example}{Example}
\newtheorem{assumption}{Assumption}

\usepackage{cleveref}
\Crefformat{Theorem}{#2Theorem #1#3}
\Crefformat{proposition}{#2Proposition #1#3}
\Crefformat{corollary}{#2Corollary #1#3}
\Crefformat{claim}{#2Claim #1#3}
\Crefformat{Lemma}{#2Lemma #1#3}
\Crefformat{Definition}{#2Definition #1#3}
\Crefformat{assumption}{#2Assumption #1#3}
\crefformat{assumption}{#2assumption #1#3}

\usepackage{listings}
\begin{document}
\maketitle

\begin{abstract}
     Compositional data (i.e., data comprising random variables that sum up to a constant) arises in many applications including microbiome studies, chemical ecology, political science, and experimental designs. 
     Yet when compositional data serve as covariates in a regression, the sum constraint renders every covariate automatically conditionally independent of the response given the other covariates, since each covariate is a deterministic function of the others. Since essentially all covariate importance tests and variable selection methods, including parametric ones, are at their core testing conditional independence, they are all completely powerless on regression problems with compositional covariates. In fact, compositionality causes ambiguity in the very notion of relevant covariates. To address this problem,  we identify a natural way to translate the typical notion of relevant covariates to the setting with compositional covariates and establish that it is intuitive, well-defined, and unique. 
     We then develop corresponding hypothesis tests and controlled variable selection procedures via a novel connection with \emph{bivariate} conditional independence testing and partial conjunction hypothesis testing. Finally, we provide theoretical guarantees of the validity of our methods, and through numerical experiments demonstrate that our methods are not only valid but also powerful across a range of data-generating scenarios.
\end{abstract}

\section{Introduction}\label{sec: introduction}
\subsection{Background and Motivation}\label{sec: background motivation}
Compositional random variables are random vectors whose elements sum up to a known non-random constant. 
Compositional data arises whenever a fixed whole is divided into parts, such as in chemistry (chemical composition of solutions~\citep{bruckner2017chemo, baum1988use}), material science~\citep{pesenson2015statistical}, marketing (e.g., brand or product shares~\citep{joueid2018marketing}), sociology (e.g., time-use surveys~\citep{dumuid2020compositional}), political science (e.g., voting proportions~\citep{rodrigues2009analysis}), food science (e.g., nutrient composition~\citep{greenfield2003food}), or geology (e.g., rock composition~\citep{buccianti2006compositional}). Compositionality also arises when data values are only meaningful relative to one another, making it natural to normalize them to sum to a constant, such as relative abundance data in microbiome research \citep{turnbaugh2007human,gloor2017microbiome}. Any categorical or factor random variable can also be treated as compositional via one-hot encoding, resulting in a random vector of all zeros except a single entry one indicating the factor level. Such categorical data arises in many settings, e.g., clinical trials (e.g., testing multiple drugs~\citep{festing2020completely}), social science/economics (e.g., evaluating multiple disjoint policies~\citep{stoker2009design}), genomics (e.g., gene knockout experiments~\citep{zhang1994positional}) and business/marketing (e.g., conjoint analysis~\citep{green1978conjoint}).
While their applications and usage vary, 
compositional data across fields all present a common problem when they are treated as covariates (i.e., explanatory variables, $X$) to explain or predict a response variable ($Y$).

Consider a specific example where the research question is to relate the gut microbiome (covariates measured, as is standard, via the compositional relative abundances of different gut microbes) with colorectal cancers, or CRC (response variable). {\it Fusobacterium nucleatum}, an opportunity pathogen, is considered a gut microbial risk factor for CRC, with mechanistic and epidemiological evidence supporting its role in the cancer's tumorigenesis and development~\citep{wang2023fusobacterium}.
However, the compositionality of microbiome abundances means that \emph{F. nucleatum} cannot be independent of the rest of the microbiome. Unconditionally, this means that any marginal dependence between CRC and \emph{F. nucleatum} guarantees marginal dependence between CRC and other microbes, even if they are non-cancer-related. Such issues with dependent (even non-compositional) covariates are the reason that regression inference is typically done in a multivariate way, where each covariate's dependence with the response is evaluated \emph{conditionally} on the other covariates. Yet due to compositionality, the abundance of \emph{F. nucleatum} is \emph{deterministic} given the abundance of the other microbes, and hence can have \emph{no} conditional dependence on the response, wrongfully negating its biological association with CRC. This example highlights the challenge of identifying, and even defining, important covariates under compositionality, so we now formalize a simplified version of it (where \emph{F. nucleatum} is the only important covariate) with mathematical notation and concrete distributions.

\begin{example}\label{ex:onebug}
Let $Y \in \{0,1\}$ represent the CRC disease status and $X:= (X_{\{1\}}, \dots, X_{\{p\}})$ represent the relative abundances of $p$ microbes in a gut microbiome sample. By definition, $\sum_{j=1}^p X_{\{j\}} = 1$. To simplify things further, suppose the $X_{\{j\}}$ are i.i.d. $\mathrm{Expo}(1)$ random variables normalized by their sum $\sum_{j=1}^p X_{\{j\}}$ (note the normalization makes them no longer independent of one another) and the true (unknown) causal model for $Y$ given $X$ is $Y\mid X \sim $ $\mathrm{Bernoulli}(e^{-X_{\{1\}}})$, where $X_{\{1\}}$ denotes the relative abundance of \emph{F. nucleatum}. Although compositionality ensures that a model that excludes $X_{\{1\}}$, namely, $Y\mid X \sim $ $\mathrm{Bernoulli}(e^{-(1-\sum_{j=2}^p X_{\{j\}})})$, explains the response as well as the causal model, selecting $\{X_{\{2\}}, \dots, X_{\{p\}}\}$ as the set of important covariates for $Y$ (as implied by that model) leads to a very wrong scientific conclusion. And indeed, while both the sets $\{X_{\{1\}}\}$ and $\{X_{\{2\}},\dots,X_{\{p\}}\}$ each can explain $Y$ equally well, in most scientific problems, a domain expert presented with these two options will choose $\{X_{\{1\}}\}$ as the far more biologically plausible one; this is the ubiquitous scientific principle of parsimony, or Occam's Razor. Thus, we argue that a good definition of an important covariate, when applied to this example, should consider only $\{X_{\{1\}}\}$ as important, yet this is not the case for existing methods as we explain in the rest of this subsection and also in Section~\ref{sec: related work}.    
\end{example}

As alluded to in the microbiome example above, existing definitions of covariate importance consider $X_{\{j\}}$ ``unimportant" (this is considered the null hypothesis $\H_j$) either when it satisfies \emph{unconditional} independence with $Y$ or \emph{conditional} independence with $Y$ given the remaining covariates $X_{\{j\}^\mathsf{c}}$.
In Example~\ref{ex:onebug}, the unconditional $\H_1$ will be false, correctly defining $X_{\{1\}}$ as an important covariate. But compositionality also makes every \emph{other} $X_{\{j\}}$ unconditionally dependent on $X_{\{1\}}$ and hence also on $Y$, so \emph{all} unconditional $\H_j$ will be false! It is not hard to see that this situation is typical when the covariates are compositional: for almost all compositional covariate distributions and (even sparse) response conditional distributions, all unconditional hypotheses will be false and hence any hypothesis test or selection method controlling an error rate with respect to such hypotheses will be unreliable for Type I errors with respect to any sparse true model. Switching to the conditional version of $\H_j$ creates the opposite problem: compositionality means that every $X_{\{j\}}$ is a deterministic function of $X_{\{j\}^\mathsf{c}}$, which immediately implies that each $X_{\{j\}}$ is conditionally independent of $Y$ given $X_{\{j\}^\mathsf{c}}$ and hence gives the degenerate conclusion that \emph{every} conditional $\H_j$ is true! Thus, any hypothesis test or selection method based on identifying false conditional $\H_j$'s will be guaranteed to have trivial power.\footnote{For instance, a level-$\alpha$ hypothesis test for $H_j$ must have power at most $\alpha$.}

Unfortunately, the conditional hypothesis $\H_j: Y\indep X_{\{j\}}\mid X_{\{j\}^\mathsf{c}}$ forms the basis of the majority of existing parametric and nonparametric hypothesis tests and variable selection techniques. 
For instance, parametric hypothesis tests and variable selection methods generally test hypotheses of the form $\beta_j=0$~
\citep{barber2015controlling, JJ19,montgomery2021introduction} where the parameter $\beta_j$ represents the contribution of $X_{\{j\}}$ to $Y$'s conditional model, and a zero value represents no contribution and thus conditional independence. Similarly, a standard nonparametric null hypothesis for covariate importance in causal inference is that the average treatment effect is zero \citep{chernozhukov2018double}, which is also implied by conditional independence. It has also become increasingly popular to explicitly test conditional independence \citep{candes2018panning,shah2020hardness} (or conditional mean independence \citep{lundborg2022projected} which is also implied by conditional independence) directly. Since the conditional independence hypothesis is true for all covariates in a regression problem with compositional covariates, all the aforementioned methods (and many more) are guaranteed to have trivial power in such settings, regardless of the true relationship between $Y$ and $X$.

We note here that a common workaround for conditional testing is to simply drop one arbitrarily chosen covariate, rendering the remaining covariates non-compositional. This indeed fixes the problem entirely, but \emph{only} if the dropped covariate is not part of the true model. If the dropped covariate is part of the true model, not only does it preclude identifying that true covariate, but, due to the compositionally induced dependence, it also renders false covariates 
non-null. In particular, in Example~\ref{ex:onebug} if we were to drop $X_{\{1\}}$, then conditional testing would of course be unable to correctly discover $X_{\{1\}}$, but would also consider all of $X_{\{2\}},\dots,X_{\{p\}}$ as non-null. Since an analyst will not know a priori which variables are in the true model, there is no clear way in general to know that dropping a covariate will actually help.

\subsection{Contributions}
This paper makes two core contributions:

\begin{enumerate}
    \item \textbf{Formalizing important covariates under compositionality:} 
    To overcome the aforementioned misalignment between hypotheses (conditional or unconditional) and true signals in a parsimonious regression model with compositional covariates, we define the set of important covariates as the minimal set of covariates that together render all other covariates unhelpful for predicting the response, namely, the Markov boundary. We argue that under compositionality of the covariates, the Markov boundary aligns with natural intuition for what covariates should be discovered, and we show that it is well-defined, nontrivial, and unique under mild conditions. Although our conditions are technical, we show they can be simplified in certain compositional cases of interest.

    \item \textbf{Methods for testing and variable selection with compositional covariates:} Our results about the Markov boundary reveal a direct connection between it and
    the \emph{bivariate} conditional independence hypotheses $\H_{i,j}:= Y \perp\!\!\!\perp X_{\{i,j\}^\mathsf{c}} \mid  X_{\{i,j\}}$, which are not rendered trivially true for compositional covariates like their univariate analogues. 
    Leveraging this connection, we frame testing for covariate importance as a partial conjunction hypothesis test and apply multiple testing methods to such tests for variable selection.
    We propose a number of methodological innovations specific to our problem for partially overcoming the conservativeness inherent in partial conjuction hypothesis testing, and we demonstrate that our methods are valid and powerful in a range of simulations.
\end{enumerate}

\subsection{Related work}\label{sec: related work}
Testing for and selecting important covariates in multivariate regression is a heavily studied topic. First, there are many methods for estimating covariate importance or the set of important variables, which can be useful but do not generally come with Type I error guarantees, e.g., \cite{NIPS2008_7380ad8a, schwarz1978estimating, tibshirani1996regression, bertsimas2016best, berger1996intrinsic}. As our ultimate goal in this paper is to perform rigorous hypothesis tests and \emph{controlled} variable selection, we focus our literature review on methods with Type I error guarantees, which either devise individual hypothesis tests (whose p-values can be fed into multiple testing procedures \citep{holm1979simple, benjamini1995controlling, benjaminiyekutieli} for controlled variable selection) or directly perform controlled variable selection without individual hypothesis testing. Either way, these methods broadly fall into parametric (e.g., \cite{barber2015controlling,javanmard2018debiasing,montgomery2021introduction, kuchibhotla2022post,xing2023controlling}) and nonparametric (e.g., \cite{candes2018panning, chernozhukov2018double, berrett2020conditional, dCRT,lundborg2022projected}) approaches. But the degeneracy of conditional covariate importance mentioned in Section~\ref{sec: background motivation} applies to all of these methods, parametric and nonparametric: they all define covariate importance in such a way that $Y\indep X_{\{j\}} \mid X_{\{j\}^\mathsf{c}}$ implies that $X_{\{j\}}$ is unimportant/null, and hence all of these methods treat \emph{all} compositional covariates as null and thus will have trivial power to identify any covariates of interest. It is worth noting that although prior works like \cite{candes2018panning} do, like us, define the set of important covariates as the Markov boundary, they justify this target under standard assumptions which fail to hold for compositional covariates, and in particular their inferential methods are entirely based on the identity that the Markov boundary is exactly the set of covariates $X_{\{j\}}$ for which $Y\nindep X_{\{j\}} \mid X_{\{j\}^\mathsf{c}}$; this identity fails to hold when the covariates are compositional, and hence new theory (which we present in Section~\ref{sec:well-defined variable selection target}) is needed to handle regression with compositional covariates. 

There does exist a substantial body of literature on analyzing compositional data, but much of it has been focused on analyzing properties of the distribution of the compositional data itself (see, e.g., \cite{aitchison1982statistical,egozcue2003isometric,filzmoser2009correlation,filzmoser2009principal,pawlowsky2011compositional, greenacre2021compositional,pal2022clustering}), as opposed to analyzing the conditional distribution of a response variable $Y$ conditioned on compositional covariates $X$. There is work on compositional outcome regression~\citep{fiksel2022transformation}, which is regression with a compositional response variable, but this is still different from our setting where the \emph{covariates} are compositional. We note that, primarily within the microbiome field, a popular group of methods known as the Differential Abundance (DA) analysis ~\citep{paulson2013differential,mallick2021multivariable} tests for marginal independence, i.e., $Y\indep X_{\{j\}}$ for each $j$, optionally adjusting for confounding covariates \textit{but not other microbes} (e.g. any host demographics, phenotypes, or exposure variables than can be associated with both $X$ and $Y$).\footnote{We note that ``differential abundance" analysis can also refer to methods that test for marginal independence between $Y$ and the latent \emph{absolute} abundance of the $j^{\text{th}}$ microbe \citep{zhou2022linda, wang2023robust,lin2024multigroup,zong2024mbdecoda}. Even moreso than canonical DA analysis, these methods differ fundamentally from the approach of this paper but require more space to fully explain, so we defer such discussion to \Cref{app:absabund}.
} 
Within canonical DA analysis, $X_{\{j\}}$ is usually treated as the outcome in regression modeling, against $Y$ as the exposure variable of interest and other covariates. The corresponding null hypothesis, however, is equivalent to treating individual $X_{\{j\}}$'s as covariates and $Y$ as the outcome (as they are treated in our paper), due to the tests' marginal nature. As such, DA methods do not account for the compositionality, or more generally the dependence, of microbial covariates.
Hence they all suffer from the issue highlighted in Section~\ref{sec: background motivation} that unconditional testing will usually treat all compositional covariates as \emph{non-null}, and thus cannot be expected to control any Type I error rate with respect to any sparse or parsimonious set of important covariates (see \citet[Appendix B]{candes2018panning} for related discussion).

Finally, there are works which take a conditional approach to testing or selecting important covariates in regression with compositional covariates, and as these are the works most closely related to ours, we review them more closely in the rest of this subsection. 
A number of such works assume $Y\mid X$ follows a linear or generalized linear model with respect to the log-transformed covariates~\citep{lin2014variable,shi2016regression,lu2019generalized,sohn2019compositional}. To address the unidentifiability of the coefficient vector induced by compositionality, these works enforce a constraint on the coefficient vector that it must sum to zero. To the best of our knowledge this constraint is purely statistically motivated and thus may impact the interpretation of the coefficient vector they treat as their inferential target. In particular, if in Example~\ref{ex:onebug} we were to model $Y\mid X \sim \mathrm{Bernoulli}(e^{-(\beta_0 + X\beta)})$ for an intercept $\beta_0\in\mathbb{R}$ and coefficient vector $\beta\in\mathbb{R}^p$ that sums to 0, then the ``correct" value of $\beta$ would be $(\frac{p-1}{p},-\frac{1}{p},\dots,-\frac{1}{p})$ which is not sparse, and hence inference based on zeros/non-zeros in $\beta$ will scientifically incorrectly treat $X_{\{2\}},\dots,X_{\{p\}}$ as non-null. \cite{srinivasan2021compositional} employs a similar linear model as these works but replaces the assumption that the coefficient vector sums to zero with an assumption that a non-empty set of variables can be screened away based on the data, and that set of variables is guaranteed to contain \emph{no} non-nulls---this is akin to a minimum signal size condition which may not hold in many settings. 

More broadly, when regression covariates contain \emph{factors}, i.e., covariates that take $k$ possible non-metric levels, they are often handled by \emph{contrast coding}, which creates $k$ binary features each representing the indicator function for one of the factor levels, but then drops one of the factor levels (called the \emph{reference} level) to address the compositionality of these $k$ covariates~\citep{montgomery2017design}. Contrast coding raises the challenge of selecting the reference level, since any downstream conclusions will depend on its choice (and one is generally not allowed to choose it based on the data itself). In particular, as mentioned at the end of Section~\ref{sec: background motivation}, dropping one of the covariates in Example~\ref{ex:onebug} and then defining importance via conditional dependence will correctly consider only $X_{\{1\}}$ as important if one of $X_{\{2\}},\dots,X_{\{p\}}$ is dropped, but will very wrongly consider $X_{\{2\}},\dots,X_{\{p\}}$ as important if $X_{\{1\}}$ is dropped, and, given the analyst does not know a priori that $X_{\{1\}}$ is special, it is unclear how an analyst can avoid such a mistake via dropping a variable.

Finally, the work that is closest to ours is \cite{SM-CH-LJ:2024}, which also adapts the usual conditional definition of covariate importance via conditional independence to a new definition catered to compositional covariates that avoids the degeneracy of conditional independence. \cite{SM-CH-LJ:2024} defines a covariate as important if and only if $Y\nindep X_{\{j\}} \mid \frac{X_{\{j\}^\mathsf{c}}}{\sum_{k\neq j}X_{\{k\}}}$, which says that a covariate is important if, when we intervene on its value while keeping the \emph{relative} proportions of the remaining covariates fixed, $Y$'s (conditional) distribution changes. While this interventional definition is appealing and may be appropriate in some settings, 
it automatically treats all levels of a one-hot-coded factor covariate as null, and furthermore
it defines all covariates as important in Example~\ref{ex:onebug}. In contrast, the definition of important covariates we lay out in this paper, using the Markov boundary, naturally handles factor covariates and correctly identifies only $X_{\{1\}}$ as important in Example~\ref{ex:onebug}.

\subsection{Notation}
In this paper, we use $[p]$ to denote the set $\{1, \dots, p\}$ for any $p \in \mathbb{N}$. For any set $A\subseteq [p]$, denote its complement $[p]\setminus A$ by $A^\mathsf{c}$. For any $p$-dimensional vector $X = (X_1, \ldots, X_p)$ and a subset $A \subseteq [p]$, $X_A$ denotes the subvector $(X_j : j \in A)$ of $X$. For an $m \times n$ matrix $P$ and indices $i \in [m]$ and $j \in [n]$, we denote the entry in row $i$ and column $j$ by $P_{i,j}$ and we define $\log(P)$ as the $m\times n$ matrix with $\log(P)_{i,j}=\log(P_{i,j})$.

\subsection{Outline}
In \Cref{sec:well-defined variable selection target}, we formally define important covariates in the compositional setting via the Markov boundary and prove that it remains well-defined under mild conditions despite compositionality. \Cref{sec: methodology} details our methods for testing and controlled variable selection and their theoretical guarantees. In \Cref{sec: simulations}, we assess the validity and power of our proposed methods across diverse data generation scenarios and hyperparameter settings. Finally, \Cref{sec: discussion} concludes the paper with contributions, limitations, and ideas for future research.

\section{Defining important compositional covariates}\label{sec:well-defined variable selection target}
As highlighted in the previous section, the typical ways of defining an important covariate are not well-suited to compositional covariates. In this section, we put forth the Markov boundary as a solution and argue that, unlike conditional or unconditional dependence, membership in the Markov boundary continues to capture the right notion of an important covariate even under compositionality. 

\subsection{The Markov boundary}\label{subsec:the Markov boundary}

Suppose we have an $\R^p$-valued covariate vector $X= (X_{\{1\}}, \dots, X_{\{p\}})$ and a response $Y$; we place no restriction on the form of $Y$, and in particular it can be Euclidean-valued or categorical. 
As $X$ is compositional, $\sum_{i=1}^p X_{\{j\}} = c$ for some constant $c$, and without loss of generality, we will assume $c=1$. As has been done in previous work \citep{candes2018panning} (though not for compositional covariates), we consider $Y$'s Markov boundary as a natural definition of the set of important covariates:


\begin{Definition}[Markov boundary~\citep{PEARLbook}]\label{def: Markov boundary}
    For random variables $X\in\mathbb{R}^p$ and $Y$ defined on the same probability space, 
    a subset $\M$ of $[p]$ is called a Markov boundary of $Y$
    if:
    \begin{enumerate}
        \item $Y\indep X_{\M^\mathsf{c}}\mid X_{\M}$, and
        \item $Y\nindep X_{{\M^\prime}^\mathsf{c}}\mid X_{{\M^\prime}}$ for all $\M^\prime \subsetneq \M$.
    \end{enumerate}
\end{Definition}

Unfortunately, when the covariates are compositional, there will in general be multiple Markov boundaries of $Y$, so the remaining subsections of this section are devoted to establishing mild conditions under which the above definition provides a well-defined, unique, and nontrivial set of important covariates. But in the remainder of this subsection we will first justify why \Cref{def: Markov boundary} provides a natural and intuitive definition for the set of important covariates. For readability, in the remainder of this subsection, we will suppose we are in the setting established by the end of next subsection, in which there is a unique nontrivial Markov boundary, and hence refer to \emph{the} Markov boundary instead of \emph{a} Markov boundary.

Item 1 in \Cref{def: Markov boundary} says that, after accounting for the covariates in the Markov boundary, all the remaining covariates provide no further information about $Y$. Item 2 says that the Markov boundary is the minimal such set, in the sense that no subset of it has the property in item 1. Together, this definition informally says that the Markov boundary is the minimal set of variables that, once known, allows us to drop all other variables without losing information about $Y$, and removing any variable from this set would lead to a strict loss of information about $Y$. When the covariates are not compositional, the Markov boundary is (under very mild conditions) equivalent to the set of important covariates defined via conditional dependence \citep{edwards2012introduction,candes2018panning}, which, as mentioned in the second-to-last paragraph of \Cref{sec: background motivation}, is the basis for covariate importance throughout the literature on parametric and nonparametric methods for identifying important covariates in regression.
Thus we find it to be a natural and intuitive target for variable selection with compositional covariates \emph{if} we can show it remains well-defined under compositionality.

\subsection{Conditions for uniqueness of the Markov boundary} \label{subsec:uniqueness theorem} 

The reader may be surprised that we can establish any conditions under which the Markov boundary will ever be unique with compositional covariates, since an immediate consequence of covariate compositionality is that \emph{every} set of size $p-1$ covariates will satisfy item 1 of \Cref{def: Markov boundary} and hence contain a subset satisfying condition 2 (and hence there can never be a truly unique Markov boundary when the covariates are compositional). However, we will show in this subsection that under very mild conditions, all but at most one Markov boundary will be of \emph{exactly} size $p-1$, and it is our position that, when $X$ is compositional, a Markov boundary of size $p-1$ is \emph{trivial} and not worth considering, since every set of covariates of size $p-1$ is information theoretically equivalent to the entire set of covariates due to compositionality. In \Cref{ex:onebug}, $\{1\}$ and $\{2, \dots, p\}$ are both technically Markov boundaries, but the latter is trivial and the former is the unique nontrivial Markov boundary. Our position is that the unique (under mild conditions) \emph{nontrivial} Markov boundary is the natural target for variable selection even though other Markov boundaries technically exist. 
First, we give an example to demonstrate why it is necessary to make some sort of assumptions in order to establish uniqueness of even of a nontrivial Markov boundary.

\begin{example}\label{ex:linear dependence}
Suppose $p=4$, $X_{\{1\}}, X_{\{2\}}, X_{\{3\}}$ are mutually independent, and $X_{\{4\}} = 1-X_{\{1\}}-X_{\{2\}}-X_{\{3\}}$. If $Y = X_{\{1\}} + X_{\{2\}} + \epsilon$ for $\epsilon$ independent of $X$, then:
\begin{itemize}
\item $Y \indep X_{\{3\}}, X_{\{4\}} \mid X_{\{1\}}, X_{\{2\}}$, since $\epsilon \indep X_{\{3\}}, X_{\{4\}}$,
\item $Y \indep X_{\{1\}}, X_{\{2\}} \mid X_{\{3\}}, X_{\{4\}}$, since $X_{\{1\}}+X_{\{2\}}=1-X_{\{3\}}-X_{\{4\}}$, and
\item $Y\nindep X_{\{1,2,3,4\}\setminus \{j\}}\mid  X_{\{j\}}$ for all  $j\in \{1,2,3,4\}$.
\end{itemize}
Thus, both $\{1,2\}$ and $\{3,4\}$ are Markov boundaries of size $<p-1$. 
\end{example}

Such a simple example may give the impression that it is not hard to find non-unique nontrivial Markov boundaries, but in fact it is constructed very carefully and we argue it should be vanishingly rare to encounter such an example in real data. 
In particular, the non-uniqueness of the nontrivial Markov boundary is a product of a precise interaction between $Y\mid X$, which depends on $X_{\{1\}}$ and $X_{\{2\}}$ through exactly their unweighted sum, and the compositionality of $X$, whose constraint also involves $X_{\{1\}}$ and $X_{\{2\}}$ through exactly their unweighted sum. So even changing $Y\mid X$ to $ Y = X_{\{1\}} + 0.999 X_{\{2\}} + \epsilon $ or $Y = X_{\{1\}} + X_{\{2\}}^{0.999}+\epsilon$ would result in a unique nontrivial  Markov boundary. Thus, it seems one would have to get \emph{very} unlucky to encounter a real data set whose alignment of $Y\mid X$ with the compositional constraint on $X$ is so perfect.

To formalize when the nontrivial Markov boundary is uniquely defined, we will represent it as a function of identifiable and testable statistical hypotheses. Let $\M$ denote any Markov boundary as defined by \Cref{def: Markov boundary}. For any $j\in\M$, if $X$ were \emph{not} compositional, we would expect $Y$ to be conditionally dependent of $X_{\{j\}}$ given $X_{\{j\}^\mathsf{c}}$ since $X_{\{j\}}$ contains unique conditional information about $Y$ not captured by $X_{\{j\}^\mathsf{c}}$; indeed the Markov boundary is typically (when the covariates are non-compositional) uniquely identified as the set of $j$ for which $Y\indep X_{\{j\}}\mid X_{\{j\}^\mathsf{c}}$ is false \citep{edwards2012introduction,candes2018panning}. The reason this identification strategy fails for compositional covariates is that any one covariate is perfectly determined by all the others, so, intuitively, $X_{\{j\}}$'s unique conditional information about $Y$ cannot shine through because it is entirely absorbed into the information in $X_{\{j\}^\mathsf{c}}$. An intuitive fix is to shift one more covariate from $X_{\{j\}^\mathsf{c}}$ to join $X_{\{j\}}$, i.e., for some $i\in \{j\}^\mathsf{c}$, consider the conditional relationship between $Y$ and $X_{\{i,j\}}$ given $X_{\{i,j\}^\mathsf{c}}$. If $j\in\M$, we would expect $Y \nindep X_{\{i,j\}} \mid X_{\{i,j\}^\mathsf{c}}$ to hold for any $i\in \{j\}^\mathsf{c}$, since $X_{\{j\}}$ contains unique conditional information about $Y$, and hence so should $X_{\{i,j\}}$, and that unique information is no longer entirely absorbed by the covariates that are conditioned on under compositionality, since $X_{\{i,j\}}\mid X_{\{i,j\}^\mathsf{c}}$ is \emph{not} deterministic under compositionality. 
This leads us to define 
$$\S = \{j\in [p]: \H_{i,j} \text{ is false for all } i\in \{j\}^\mathsf{c}\},$$ 
where the $\H_{i,j}: Y \indep X_{\{i,j\}} \mid X_{\{i,j\}^\mathsf{c}}$ are \emph{bivariate} conditional independence null hypotheses. $\S$ consists of the indices $j$ such that all the bivariate null hypotheses containing $j$ as one of the indices are false. After a brief remark, the remainder of this subsection will establish what our intuition suggests: under mild conditions, $\S$ (which is uniquely defined by construction) coincides with the unique nontrivial Markov boundary. 

\begin{remark}\label{remark:trivariate}
Key to our intuitive argument that bivariate conditional independence will allow us to identify the Markov boundary is the fact that compositionality places only a \emph{single} constraint on $X$. If there were a further constraint, e.g., $\sum_{j=1}^{p} jX_{\{j\}} = p/2$, then $X_{\{i,j\}}\mid X_{\{i,j\}^\mathsf{c}}$ would again be deterministic for all $i,j$ pairs, and we would have to instead consider \emph{tri}variate conditional independence, $\H_{i,j,k}: Y\indep X_{\{i,j,k\}}\mid X_{\{i,j,k\}^\mathsf{c}}$, and so on if there were yet more constraints on $X$. We will detail an important exception to this intuition in \Cref{subsubsec: simplification for factor covariates}, 
but in the remainder of this subsection it may help the reader to think of the compositional $X$ we consider as satisfying no additional deterministic constraints beyond compositionality, and our technical conditions will for the most part preclude additional constraints.
\end{remark}

First, we show that without any further conditions, any nontrivial Markov boundary contains $\S$:

\begin{Lemma}\label{Lemma: M supset S}
For any Markov boundary $\M$ such that $|\M|<p-1$, $\M \supseteq \S$.
\end{Lemma}

\begin{proof}
Since $\M$ is a Markov boundary, $Y \indep X_{\M^\mathsf{c}}  \mid X_{\M}$. The weak union property\footnote{(Weak union property) For any $\R^p$-valued random variable $X$ and four subsets $A$, $B$, $C$ and $D$ of $[p]$, $X_A\indep X_{B\cup D} \mid X_C \text{ implies } X_A\indep X_B \mid X_{C\cup D}$.}~\citep{PEARLbook} implies that for any $\{i,j\}\subseteq \M^\mathsf{c}$, $\H_{i,j}$ is true. Since $\S$ contains exactly the $i$ such that $\H_{i,j}$ is false for all $j\neq i$, that means that as long as $|\M^\mathsf{c}|>1$ (which holds by assumption), $i\in \M^\mathsf{c}$ implies $i \in \S^\mathsf{c}$. Thus $\M^\mathsf{c} \subseteq \S^\mathsf{c}$, i.e., $\M\supseteq \S$.
\end{proof}
If we can show $Y\indep X_{\S^\mathsf{c}}\mid X_{\S}$, then we are done, since $\M \supseteq \S$ by \Cref{Lemma: M supset S} but if $\M\supsetneq \S$ it would contradict property 2 of \Cref{def: Markov boundary} of a Markov boundary, so that only leaves the possibility that $\M = \S$. Such a result will require some technical (yet, we argue, mild) conditions, but the intuition is as follows. 
If there exists an $\M$ of size $<p-1$, then since $\S\subseteq\M$ by \Cref{Lemma: M supset S}, $|\S^\mathsf{c}|\ge 2$, i.e., given some $i\in \S^\mathsf{c}$, there exists by definition some other $j\in \S^\mathsf{c}$ such that $Y\indep X_{\{i,j\}}\mid X_{\{i,j\}^\mathsf{c}}$. Suppose we have another index $k\in \S^\mathsf{c}\setminus \{i,j\}$ such that $Y\indep X_{\{j,k\}}\mid X_{\{j,k\}^\mathsf{c}}$ (if there is no such $k$, then $|\S|=p-2$ and \Cref{Lemma: M supset S} immediately proves $\M=\S$ given $|\M|<p-1$). If we can apply the \emph{intersection property}\footnote{(Intersection property) For any $\mathbb{R}^p$-valued random variable $X$ and four subsets $A$, $B$, $C$, and $D$ of $[p]$, if $X_A \indep X_B \mid X_{C\cup D}$ and $X_A \indep X_C \mid X_{B\cup D}$, then $X_A \indep X_{B \cup C}\mid X_D$.} of conditional independence~\citep{PEARLbook} on these variables, then $Y\indep X_{\{i,j,k\}}\mid X_{\{i,j,k\}^\mathsf{c}}$. So, the intersection property `grew' the conditional independence statement from $Y\indep X_{\{i,j\}}\mid X_{\{i,j\}^\mathsf{c}}$ to $Y\indep X_{\{i,j,k\}}\mid X_{\{i,j,k\}^\mathsf{c}}$. If all the elements of $\S^\mathsf{c}$ can be iteratively added to the conditional independence statement via the intersection property, then we arrive at $Y\indep X_{\S^\mathsf{c}}\mid X_{\S}$ as desired. So, the key to proving $\M = \S$ is showing that we can apply the intersection property recursively, and the rest of this subsection describes sufficient conditions for this. 

The standard sufficient condition for the intersection property is that $X$ has a continuous density with respect to the Lebesgue measure \citep{Peters}. This fails for compositional $X$, whose support lies on the standard $(p-1)$-simplex, which has Lebesgue measure zero. So, we need to find a different set of assumptions which will be appropriate for compositional covariates. To do that, we first need to reintroduce some definitions from \cite{Peters} about a random variable $X\in\mathbb{R}^p$ with distribution $F$. 

\begin{Definition}[Coordinatewise connectivity (through $A$ and $B$) \citep{Peters}]
Let $A$ and $B$ be two subsets of $[p]$. Then two points $x,x^\prime$ $\in$ $\mathrm{support}(F)$ are said to be coordinatewise connected (through $A$ and $B$) if either $x_A =x^\prime_A$ or $x_B =x^\prime_B$.
\end{Definition}

\begin{Definition}[Equivalence (through $A$ and $B$) \citep{Peters}]
Any two points $x_1, x_2 \in \mathrm{support}(F)\subseteq \R^p$ are said to be equivalent (through $A$ and $B$) if there exists $z_1, z_2, \dots, z_L\in \mathrm{ support}(F)$ for some finite $L$ such that the pairs $(x_1, z_1), (z_1, z_2), \dots, (z_{L-1}, z_L), (z_L, x_{2})$ are each coordinatewise connected through $A$ and $B$.
\end{Definition}

So, equivalence extends coordinatewise connectivity in a pathwise way, in the sense that any two equivalent points can be connected by a sequence of consecutively coordinatewise connected points. Note that, as the name suggests, equivalence is indeed an equivalence relation, so we can define corresponding equivalence classes. The following lemma proves a version of the intersection property~\citep{PEARLbook} for two sets of variables such that $X$ has a single equivalence class through them.

\begin{Lemma}\label{Lemma: 1 equiv class}
For random variables $X\in\mathbb{R}^p$ and $Y$ defined on the same probability space, if $A\subseteq [p]$ and $B\subseteq [p]$ are such that $X_{A\cup B}\mid X_{(A\cup B)^\mathsf{c}}$ has only one equivalence class $($through $A\setminus B$ and $B\setminus A)$, 
then $Y\indep X_A\mid X_{A^\mathsf{c}}\ \text{and } Y\indep X_B\mid X_{B^\mathsf{c}}$ together imply $Y\indep X_{A\cup B}\mid X_{(A\cup B)^\mathsf{c}}$.
\end{Lemma}

The lemma and its proof in \Cref{subsec: proof of Lemma: 1 equiv class} are inspired by \citet[Corollary 1]{Peters}. 
Next, define 
$$\Delta = \left\{\{A,B\}: X_{A\cup B}\mid X_{(A\cup B)^\mathsf{c}} \text{ has 1 equivalence class (through } A\setminus B \text{ and }B\setminus A)\right\}.$$ 
So $\Delta$ denotes the pairs of sets on which we can apply \Cref{Lemma: 1 equiv class}. Further, define $\I$ as the set of distinct $\{i,j\}$ pairs whose bivariate null hypotheses are true: 
$$\I = \{\{i,j\}: \H_{i,j} \text{ is true}\}.$$
Now, let us define the following operation on sets.

\begin{Definition}[Set action]\label{def:set action}
Suppose we have a set $\Phi$ of pairs of sets (like $\Delta$). Consider a set $Q$ of sets (like $\I$). We define the action of $\Phi$ on $Q$ (denoted by $\Phi \circ Q$) as $$\Phi \circ Q := Q\cup \{A\cup B : A\in Q, B\in Q, \{A,B\}\in \Phi \}. $$
\end{Definition}

Note that both instances of the word ``sets" in \Cref{def:set action} will refer to sets of integers in $[p]$ in this paper.
Consider any set $C\in\Delta\circ \I$. By \Cref{def:set action}, there exist $A, B\in \I$ satisfying $\{A,B\}\in \Delta$ with $A\cup B = C$. Since $A,B\in \I$, we have, by \Cref{Lemma: 1 equiv class}, $Y\indep X_{A\cup B}\mid X_{(A\cup B)^\mathsf{c}}$, i.e, $Y\indep X_C\mid X_{C^\mathsf{c}}$. Since $C$ was generic, we conclude that
$Y\indep X_{C}\mid X_{C^\mathsf{c}}$ for all $C\in \Delta\circ \I$. Applying the same argument with $\Delta\circ\I$ replacing $\I$, we get the same conditional independence for all $C\in \Delta\circ (\Delta\circ\I)$. Continuing in this way, we get 
$$Y\indep X_{C}\mid X_{C^\mathsf{c}}  \quad\forall C\in (\Delta\circ)^k \,\I, \quad\forall k\in \mathbb{Z}_{\geq 0}.$$
The above argument is just a recursive way of combining sets for which the conditional independence statement $Y\indep X_C\mid X_{C^\mathsf{c}}$ is true through the intersection property. We continue this until we have a big enough set $C$. In particular, to show $\M \subseteq \S^\mathsf{c}$ it is enough to show that we can reach a $C \supseteq \S^\mathsf{c}$ in this way. \Cref{thm: uniqueness of MB}, whose proof is presented in \Cref{subsec: proof of uniqueness Theorem}, formalizes this argument.

\begin{Theorem}[Uniqueness of nontrivial Markov boundary]\label{thm: uniqueness of MB}
If $\S=[p]$ then no nontrivial Markov boundary exists. Otherwise, if $\S^\mathsf{c} \in (\Delta\circ)^{k}\, \I$ for some $k \in \mathbb{Z}_{\geq0}$, then $\S$ is the unique nontrivial Markov boundary.
\end{Theorem}
In addition to establishing when $\S$ is the unique nontrival Markov boundary for compositional $X$, \Cref{thm: uniqueness of MB} also says that there is no nontrivial Markov boundary when all the $\H_{i,j}$ are false.
We can interpret this as the setting where there is no parsimony or sparsity at all and thus it is natural to consider \emph{all} covariates as being important in this setting. Hence, although $\S=[p]$ is not a Markov boundary in this case, it remains an appropriate definition of the set of important covariates. 

\Cref{thm: uniqueness of MB} is applicable to different types of compositional covariates, including discrete covariates as found in many experimental designs, continuous covariates such as relative abundances of chemicals, and mixed covariate distributions such as for microbe relative abundances which have both discrete (at 0) and continuous components. 
Additionally, the theorem is not restricted solely to compositional data but can also apply to any covariate data satisfying at most one deterministic constraint; see \Cref{sec: discussion} for other examples. 
%
%
However, the condition that $\S^\mathsf{c} \in (\Delta\circ)^{k}\, \I $ is rather abstract, so in the next subsection we provide simpler and more interpretable corollaries for important classes of compositional covariates.

\subsection{Simplifications of Theorem \ref{thm: uniqueness of MB}}\label{subsec:simplified theorem for special cases}

\subsubsection{Continuous covariates}\label{subsubsec:simplified Theorem for continuous compositional covariates}

Continuous compositional distributions such as the Dirichlet distribution are often used to model compositional covariates. These arise commonly when a continuous whole is partitioned, such as component fractions of a chemical solution or amount of time spent on different activities during a day.
For continuous compositional distributions, the assumptions in \Cref{thm: uniqueness of MB} can be simplified substantially. This is formalized in the following corollary:

\begin{corollary}[Uniqueness of nontrivial Markov boundary for continuous distributions]\label{cor:continuous distributions}
If $\S=[p]$ then no nontrivial Markov boundary exists. Otherwise, if the following two assumptions on the distribution of $(Y,X)$ hold:
\begin{itemize}
    \item[(i)] For all $i, j\in \S^\mathsf{c}$, there exists a finite $t$ and a sequence $l_1, l_2, \dots, l_t$ such that $\H_{i,l_1}$, $\H_{l_1,l_2}$, $\dots$, $\H_{l_{t-1},l_t}$, $\H_{l_t,j}$ are all true,
    \item[(ii)] $X$ can be written as $X = \frac{Z}{\sum_{j\in [p]} Z_{\{j\}}}$ for a $\mathbb{R}^p$-valued random variable $Z$ such that
    there exists a function $f_Z$ that is a density for $Z$ and is continuous in $[0,\infty)^p\setminus \{\boldsymbol{0}_p\}$ and zero elsewhere, and for all $C\subseteq [p]$ and $c \in [0,1]^{|C|}$, the set
    $$\left\{z: \ \frac{z_C}{\sum_{j\in[p]} z_{\{i\}}}=c, \ f_Z(z)>0\right\}$$ is path connected if non-empty,
\end{itemize}
then $\S$ is the unique nontrivial Markov boundary.
\end{corollary}
The proof of \Cref{cor:continuous distributions} from \Cref{thm: uniqueness of MB} in \Cref{appendix:proof-continuous-covariates} requires a nontrivial adaptation of \citet[Proposition 1]{Peters} to the compositional setting.
Assumption \textit{(i)} places restrictions on the distribution of $Y \mid X$ and precludes pathological cases like \Cref{ex:linear dependence} which had $\{\H_{1,2},\H_{3,4}\}$ as the (non-path-connected) set of true nulls. On the other hand, assumption \textit{(ii)} is a condition on the marginal distribution of $X$ and is completely agnostic to the distribution $Y\mid X$. Its main role is to exclude the possibility of additional constraints on the distribution of $X$ aside from the sum constraint due to compositionality. Standard continuous compositional distributions such as the Dirichlet distribution (which is a normalization of independent Gamma distributed random variables) can be easily shown to satisfy assumption \textit{(ii)}.

\subsubsection{Categorical or factor covariates}\label{subsubsec: simplification for factor covariates}

In addition to relative abundance data, compositional data can arise when covariates represent categorical factors, which is common, for instance, in experimental design. The simplest case is when the covariates represent a single factor with $p$ levels; such covariates arise, e.g., in one-way ANOVA experimental designs.
In this case, $X_{\{j\}}$ is the indicator variable for the $j^{\text{th}}$ level of the factor; this is also sometimes called ``one-hot" encoding. 
Since the factor can only take one level at a time, $X$ sums to 1. We can generalize this to $K$ factors (as would arise in a $K$-factor design) by considering $X$ as the concatenation of $K$ single-factor versions of $X$, and 
notationally, we will say that the $p$ entries of $X$ can be partitioned into sets $F_1, \dots, F_K$, corresponding to each factor, so that for each $k\in [K]$, 
$\sum_{j\in F_k}X_{\{j\}} = 1$.\footnote{While this means $\sum_{j=1}^p X_{\{j\}} = K$ instead of 1, we still refer to $X$ as compositional.}
While this imposes $K$ constraints, contradicting the intuition of \Cref{remark:trivariate}, the disjointness of the constraints (since each only applies to $X_{F_k}$ and the $F_k$ are disjoint) makes it a special case that still permits the application of \Cref{thm: uniqueness of MB}. 

Note that, due to compositionality within each factor, $\H_{i,j}$ is always true if $i$ and $j$ correspond to different factors. 
This immediately implies $\S=\emptyset$ when $K>1$, so we present a modified definition of $\S$ for this setting:
$$\S_F = \bigcup_{k\in [K]}\left\{j\in F_k: \H_{i,j} \text{ is false for all } i\in F_k\setminus \{j\}\right\}.$$ 
Note that when $K=1$, $\S_F=\S$.
We also generalize the definition of a nontrivial Markov boundary: for a Markov boundary $\M$, we define the factored Markov boundary corresponding to the factor $k$ as $\M_k=\M\cap F_k$ and we call a factored Markov boundary $\M_k$ nontrivial if $|\M_k| < |F_k|-1$. 
\begin{corollary}[Uniqueness of nontrivial factored Markov boundary for regression with $K$ factors]\label{cor:factor covariates}
Consider regression with $K$-factor covariates, and fix some $k\in[K]$. If $F_k\subseteq \S_F$, then no nontrivial factored Markov boundary corresponding to factor $k$ exists. Otherwise, if the following two assumptions hold:
\begin{itemize}
    \item[(i)]For all $i,j\in \S_F^\mathsf{c} \cap F_k$, there exists a finite $t$ and a sequence $l_1, l_2, \dots, l_t \in F_k$ such that $\H_{i,l_1}$, $\H_{l_1,l_2}$, $\dots$, $\H_{l_{t-1},l_t}$, $\H_{l_t,j}$ are all true,
    \item[(ii)] Positive probability is assigned to all levels of factor $k$,
\end{itemize}
then $\S_F\cap F_k$ is the unique nontrivial Markov boundary corresponding to factor $k$.
\end{corollary}
Assumption \textit{(i)} is on $Y\mid X$ while \textit{(ii)} is only on the marginal distribution of $X$. 
The proof of \Cref{cor:factor covariates} in \Cref{appendix:proof of factor covariates} generalizes that of \Cref{cor:continuous distributions} to multiple compositional groups of covariates.

A generalization of $K$-factor covariates enforces that each observation take $L\ge 1$ levels from each factor; this arises frequently in gene knockout experiments where exactly $L>1$ out of $p$ genes are knocked out per individual, resulting in $\sum_{j\in F_k}X_{\{j\}}=L$ for all $k\in [K]$. \Cref{cor:gene knockout experiments} in \Cref{sec:gene knockout experiments} proves uniqueness of the Markov boundary for this setting.

\begin{remark}
    Although the Markov boundary can usually be thought of as a property of $Y\mid X$, when $X$ has highly restricted support $\mathrm{supp}(X)$, the Markov boundary can also depend on $\mathrm{supp}(X)$ as well as $Y\mid X$. This phenomenon is perhaps most intuitively clear when thinking about interactions: if $Y=X_{\{1\}} X_{\{2\}}$ and $X_{\{1\}},X_{\{2\}}\in\{0,1\}$, then we expect the Markov boundary will be $\{1,2\}$ and indeed it will be if the support of $(X_{\{1\}},X_{\{2\}})$ contains $(1,1)$ and at least one other point. But if instead $X_{\{1\}}+X_{\{2\}}=1$, then $Y\stackrel{\mathrm{a.s.}}{=}0$ and the Markov boundary will be empty (though well-defined). Thus, analysts should think carefully about the choice of $\mathrm{supp}(X)$ in an experimental design and how it may impact the Markov boundary, especially when interactions are of interest~\citep{dean1999design}. 
\end{remark}

\section{Hypothesis testing and variable selection}\label{sec: methodology}
\Cref{sec:well-defined variable selection target} identified sufficient conditions for the existence of a unique nontrivial Markov boundary and argued that it provides a natural definition of important compositional covariates in regression. It also argued that when no nontrivial Markov boundary exists, the natural definition of important compositional covariates in regression is simply $[p]$. In particular, \Cref{thm: uniqueness of MB} established two possibilities: either (a) $\S=[p]$ and no nontrivial Markov boundary exists, or (b) $\S\subsetneq [p]$ and is (under the conditions of \Cref{thm: uniqueness of MB}) the unique nontrivial Markov boundary. Note that the possibility that no nontrivial Markov boundary exists and $\S\subsetneq [p]$ is precluded by the assumption of the theorem that if $\S\subsetneq [p]$, then $\S^\mathsf{c} \in (\Delta\circ)^{k}\, \I$ for some $k \in \mathbb{Z}_{\geq0}$ and $\S$ is the unique nontrivial Markov boundary. 
So while \Cref{sec:well-defined variable selection target} was focused conceptually on the Markov boundary, what it really did was justify $\S$, which is always unique and well-defined, as an appropriate definition of the set of important covariates via its connection with the Markov boundary, under mild conditions. 
Thus, turning to methodology in this section, we will always assume the conclusion of \Cref{thm: uniqueness of MB}:
\begin{assumption}\label{asm:mb}
Either (a) $\S=[p]$ and no nontrivial Markov boundary exists, or (b) $\S\subsetneq [p]$ and is the unique nontrivial Markov boundary.
\end{assumption}
Then, we simply treat $\S$ as our inferential target, and in particular we will develop hypothesis tests for
\[ \H_{0j}: j\notin\S \]
as well as multiple testing procedures combining such tests for the task of controlled variable selection.

\subsection{Testing for the importance of a compositional covariate}\label{sec:testing single covariate}
First, as is always the case in constructing hypothesis tests, we need only consider the properties of a test when the hypothesis is true (since if the hypothesis is false, it cannot be falsely rejected and hence the Type I error is automatically controlled). And the existence of a true $\H_{0j}$ precludes case (a) of \Cref{asm:mb}, so we only need to construct hypothesis tests that are valid under case (b) of \Cref{asm:mb}. Thus, in the remainder of this subsection, we will focus solely on the case when $\S\subsetneq [p]$ is the unique nontrivial Markov boundary, as this is the only case we need to consider to establish valid hypothesis tests.


Just based on its definition, it may not be immediately clear how to test $\H_{0j}$. But since $\S$ is defined via the hypotheses $\H_{i,j}$, which are standard conditional independence hypotheses which many methods exist to test, our first goal in this section is to connect $\H_{0j}$ with the hypotheses $\H_{i,j}$. In particular, by definition, $\H_{0j} = \H_{0j}^{p-1}$, where for any integer $r$,  $\H_{0j}^r$ is defined as
\begin{equation}\label{eq: bcpch}
\H_{0j}^r : |\{i \in \{j\}^{\mathsf{c}}: \H_{i,j}\text{ false}\}|<r.
\end{equation}
But in fact, letting $s:=|\S|$, since we are only considering the case when $s<p$ (recall also that by definition, $s\neq p-1$, so $s<p \Rightarrow s<p-1$), then because $\S$ is a Markov boundary ($Y\indep X_{\S^{\mathsf{c}}} \mid X_{\S}$) and by the weak union property of conditional independence, $\H_{i,j}$ is true for all $i,j \in \S^{\mathsf{c}}$. So (under \Cref{asm:mb}) $\H_{0j} = \H_{0j}^{s+1}$, and combining this with the fact that $\H_{0j}^{r} \Rightarrow \H_{0j}^{r'}$ for any $r \le r'$, we get two useful conclusions: first, a hypothesis test for $\H_{0j}^r$ constitutes a valid test for $\H_{0j}$ so long as $r>s$, and second, tests for $\H_{0j}^r$ will tend to have power nonincreasing in $r$. Thus, if an analyst has any domain knowledge upper-bounding the size of $\S$, then they should test $\H_{0j}^{\overline{s}}$ where $\overline{s}$ is the smallest (strict) upper-bound on $s$ they are confident about, and that test will be valid for $\H_{0j}$ as long as $\overline{s}>s$ (or $\overline{s}=p-1$ remains always a valid choice, regardless of the value of $s$). 

In fact it is not uncommon to have at least a coarse upper-bound for $s$ (see \cite{filzmoser2018methods, chen2013variable, rivera2018balances} for examples of regression with compositional covariates in which such an upper-bound is available), since sparsity is a common belief in many applications. As we show in our simulations, even a relatively weak upper-bound on $s$ of $\overline{s}=p/2$ is sufficient to gain substantial power over using just $\overline{s}=p-1$. However, we emphasize that extra knowledge upper-bounding $s$ is not necessary for the methods we propose, as the choice $\overline{s}=p-1$ is always valid 
and can still provide good power.

Having equated our hypothesis of interest, $\H_{0j}$, with the hypothesis $\H_{0j}^{r}$ from \Cref{eq: bcpch} with $r=\overline{s}$, we now point out that $\H_{0j}^{r}$ is a \emph{partial conjunction hypothesis} (PCH) \citep{benjamini2008screening}. For a given integer $r$, a PCH states that strictly fewer than $r$ base hypotheses are false out of some collection of base hypotheses (in our case, the set of $p-1$ hypotheses $\{\H_{i,j}: i\in\{j\}^{\mathsf{c}}\}$), aligning exactly with \Cref{eq: bcpch}. PCH testing is a well-studied topic in the statistical literature (applications include meta-analyses~\citep{owen2009karl}, genomics~\citep{wang2019admissibility}, and neuroscience~\citep{friston1999multisubject}), allowing us to apply existing PCH tests to $\H_{0j}^{\overline{s}}$ and thus to $\H_{0j}$. 

PCH tests assume access to a p-value for each of the base hypotheses. Since our base hypotheses $\H_{i,j}$ are standard conditional independence hypotheses, there exist numerous methods for testing them given, say, a data set of $n$ i.i.d. observations of the random vector $(X,Y)$ (certain forms of dependence among observations, such as those arising from common experimental designs, are also compatible with conditional independence testing; see \Cref{sec: obtaining p-values}). 
This paper does not attempt to contribute to the substantial literature on conditional independence testing, and instead we take a conditional-independence-test-agnostic approach by simply assuming that for each $\H_{i,j}$, one of the many excellent available conditional independence tests has been applied to obtain a p-value $P_{i,j}$ that is valid for $\H_{i,j}$.\footnote{While one may expect $P_{i,j}=P_{j,i}$ since $\H_{i,j}=\H_{j,i}$, we will not assume this, though we will allow it.} In \Cref{sec: obtaining p-values} we discuss a bit further the topic of testing $\H_{i,j}$ and describe the particular choice we make for our simulations.

In this paper, we will consider two PCH tests as ways to combine our base p-values $P_{i,j}$ into a valid test for $\H_{0j}$. 
Recognizing that the PCH $\H_{0j}^{r}$ is equivalent to the statement that at least $(p-1)-(r-1)$ base hypotheses are null, the standard PCH testing strategy is to apply a global null hypothesis test to the $(p-1)-(r-1)$ largest base p-values, and the validity of such an approach follows from the super-uniformity of such p-values \citep{benjamini2008screening}.
The first PCH test we consider applies the Bonferroni global test, which, given $r=\overline{s}\le p-1$ as discussed earlier (which can either be $p-1$ or a known strict upper-bound on $s$), gives the p-value
\begin{align}\label{eq:bonferroni pch p-val}
P_{j}^{\mathrm{B}} = (p-\overline{s}) P_{(\overline{s}),j},
\end{align}
where for any $i\in [p-1]$, $P_{(i),j}$ denotes the $i^{\text{th}}$ smallest p-value among $\{P_{i^\prime,j}: i^\prime\in \{j\}^\mathsf{c}\}$. The following result, which given our results so far is an immediate consequence of \citet[Theorem 2]{benjamini2008screening}, proves the validity of $P_j^\mathrm{B}$ for testing $\H_{0j}$.
\begin{corollary}\label{cor: validity of bonferroni p-val}
Under \Cref{asm:mb}, if $\overline{s}>s$ or $\overline{s}=p-1$, then for any $\alpha\in[0,1]$, $\mathbb{P}_{\H_{0j}}(P_j^\mathrm{B}\le \alpha)\le \alpha$.
\end{corollary}
The appeal of the Bonferroni approach is that it requires no assumptions about the dependence among the base p-values $P_{i,j}$. This is always a nice property to have, but is perhaps of particular value in our setting where $P_{i,j}$ and $P_{i',j}$ are testing hypotheses which are quite strongly linked, since they are computed from the same data and both contain $X_{\{j\}}$ in the pair of covariates whose conditional independence with $Y$ is being tested. However, the Bonferroni approach tends to have low power when $r$ is much less than the number of base hypotheses, and, perhaps surprisingly and as we discuss in \Cref{appendix: simulation bonferroni}, $P^{\mathrm{B}}_j$ is not monotonic with respect to $\overline{s}$. 

The other PCH test we consider applies the Simes global test \citep{simes1986improved}, which requires some assumptions about the dependence among the $P_{i,j}$ but can have substantially more power than the Bonferroni approach. Given $r=\overline{s}$ as discussed earlier, the Simes approach to testing $\H_{0j}$ gives the p-value
\begin{align}\label{eq:Simes pch p-val}
P_{j}^{\mathrm{S}} = \min _{\overline{s} \leq i \leq p-1}\left\{\frac{p-\overline{s}}{i-\overline{s}+1} P_{(i),j}\right\}.
\end{align}
Validity of the Simes p-value requires a form of positive dependence among the null base p-values, namely, \emph{positive regression dependency on each one from a subset} (PRDS) \citep{benjaminiyekutieli}. Positive dependence may be plausible (if not easily provable) for the base p-values in our setting because the shared structure of their corresponding hypotheses is a form of alignment, not disalignment, and indeed in our simulations we never observed Type I error violations for the Simes approach despite having no formal proof of PRDS for the null base p-values in any setting. Formally, we have the following validity result for $P_j^\mathrm{S}$, which given our results so far follows from \citet[Theorem 1]{benjamini2008screening}.
\begin{corollary}\label{cor: validity of simes p-val}
Under \Cref{asm:mb}, if either $\overline{s}>s$ or $\overline{s}=p-1$, and the base p-values $\{P_{i,j}:\ i\in \S^\mathsf{c}\setminus\{j\}\}$ are $\mathrm{PRDS}$, then for all $\alpha\in[0,1]$, $\mathbb{P}_{\H_{0j}}(P_j^\mathrm{S}\le \alpha)\le \alpha$.
\end{corollary}

Note that the power benefit of the Simes approach over the Bonferroni approach is only available when $\overline{s}<p-1$, since when $\overline{s}=p-1$,
\[ P_j^\mathrm{S} = P_j^{\mathrm{B}} = \max_{i\in \{j\}^\mathsf{c}} P_{i,j},\]
which is a valid p-value under no assumptions at all (not even \Cref{asm:mb}).
\begin{corollary}\label{cor: validity of max p-val}
For all $\alpha\in[0,1]$, $\mathbb{P}_{\H_{0j}}(\max_{i\in \{j\}^\mathsf{c}} P_{i,j} \le \alpha)\le \alpha$.
\end{corollary}

Because all the methods described in this subsection apply PCH testing methods to bivariate conditional independence tests, we refer to this general approach as Bivariate Conditional PCH, or BCP for short.

\subsection{Controlled variable selection}
In addition to testing for the importance of an individual compositional covariate $X_j$, an analyst may be interested in finding a set of covariates $\hat{\S}$ that contains as many important covariates as possible while controlling some error rate, i.e., controlled variable selection.

\subsubsection{Variable selection controlling the familywise error rate}\label{subsubsec: FWER control}
To select a set $\hat{\S}$ controlling the \emph{familywise error rate} (FWER) at level $\alpha$, i.e., ${\P(\hat{\S}\not\subseteq \S)\le\alpha}$, one option is to simply apply Holm's procedure \citep{holm1979simple} to the p-values $\{P_j^\mathrm{B}\}_{j=1}^p$, which will control the FWER under the conditions of \Cref{cor: validity of bonferroni p-val}. Similarly, Holm's procedure applied to $\{P_j^\mathrm{S}\}_{j=1}^p$ will control the FWER when the conditions of \Cref{cor: validity of simes p-val} hold for all $j\in\S^\mathsf{c}$, but for expositional purposes we will focus on $\{P_j^\mathrm{B}\}_{j=1}^p$ first before returning to $\{P_j^\mathrm{S}\}_{j=1}^p$ later in this subsection. 

While the above procedure is straightforward and can be powerful, the structure of the BCP p-values can be leveraged for an even more powerful procedure that controls the FWER under identical assumptions. 
To see this, first recall that Holm's procedure sequentially builds a set of rejected indices $\hat{\S}$ by starting from $\hat{\S}_1=\emptyset$ and at the $j^{\text{th}}$ step, executes the following:
\begin{equation}\label{eq: holm}
\text{if }\min_{j\in \hat{\S}_j^{\mathsf{c}}}P_{j}^\mathrm{B}\le \frac{\alpha}{p-j+1}, \text{ set }\hat{\S}_{j+1} = \hat{\S}_j \cup \{\arg\min_{j\in \hat{\S}_j^\mathsf{c}}P_{j}^\mathrm{B}\}; \text{ else, return }\hat{\S}= \hat{\S}_j.
\end{equation}
The intuition behind the Holm's procedure's increasing (in $j$) p-value threshold $\frac{\alpha}{p-j+1}$ is that, at the $j^{\text{th}}$ step, the $j-1$ hypotheses in $\hat{\S}_j$ have already been confidently identified as false, and thus one only needs to correct for the multiplicity of the remaining $p-(j-1)=p-j+1$ hypotheses in $\hat{\S}_j^{\mathsf{c}}$. But we can also apply this same intuition to the BCP p-values being thresholded to make them smaller, as follows. For a set $A\subseteq\{j\}^{\mathsf{c}}$, let $P_{(i),j}(A)$ denote the $i^{\text{th}}$ smallest p-value among $\{P_{i',j}: i'\in \{j\}^{\mathsf{c}}\setminus A\}$, and define 
\begin{align}\label{eq:corrected Bonferroni p-value}
P^\mathrm{B}_j(A) = (p-\overline{s})P_{(\overline{s}-|A|),j}(A).    
\end{align}
Then our proposed Holm's procedure simply replaces the $P_j^\mathrm{B}$'s in \eqref{eq: holm} with $P_j^\mathrm{B}(\hat{\S}_j)$; see \Cref{Algorithm: adaptive-Bonferroni--Holm} in \Cref{app:fwer algs} for a full statement.\footnote{This procedure becomes undefined at the $(\overline{s}+1)^{\text{th}}$ step since $\overline{s}-|\hat{\S}_{\overline{s}+1}|=0$, but if we assume (as we do) that $\overline{s}>s$, then the procedure can simply reject all of $[p]$ if this step is ever reached. In practice it seems more sensible (but slightly less powerful) to simply stop at the $\overline{s}^{\text{th}}$ step if it is reached, unless $\overline{s}$ is set to its default value $p-1$ and is not known to be a strict upper-bound on $s$.}
The intuition is that, having already rejected $\hat{\S}_j$ before the $j^{\text{th}}$ step, we can not only discard its $j-1$ elements in the multiplicity correction, but we can also discard its elements from the set of base hypotheses $P_{i,j}$ we consider in constructing our PCH tests. Note that for any $A\subseteq A^\prime$ and $i>|A^\prime|$, $P_{(i-|A^\prime|),j}(A^\prime)\le P_{(i-|A|),j}(A)$, and thus also $P^\mathrm{B}_j(A^\prime)\le P^\mathrm{B}_j(A)$. In particular, $P^\mathrm{B}_j(\hat{\S}_j)\le P^\mathrm{B}_j(\emptyset)=P^\mathrm{B}_j$, guaranteeing that \Cref{Algorithm: adaptive-Bonferroni--Holm}'s rejections will always contain those of the original Holm's procedure with $P^\mathrm{B}_j$. The following result (proved in \Cref{appendix: proof of FWER-Bonferroni}) establishes that it also controls the FWER under identical conditions as the original Holm's procedure---in particular, under no assumptions at all on the dependence among the $P_{i,j}$ or $P_j^\mathrm{B}$.

\begin{Theorem}\label{thm: FWER control - Bonferroni}
Under \Cref{asm:mb}, if $\overline{s}>s$ or $\overline{s}=p-1$, then \Cref{Algorithm: adaptive-Bonferroni--Holm} controls the $\mathrm{FWER}$ at level $\alpha$.
\end{Theorem}

As in the previous section, if a user is willing to make an assumption of PRDS among $\{P_{i,j}: i\in\S^\mathsf{c}\setminus\{j\}\}$ for each $j$, then it is valid and more powerful to use the Simes p-values $P^\mathrm{S}_j$ in Holm's procedure, and this procedure also admits a uniform improvement by using 
\begin{align}\label{eq:corrected Simes p-value}
P_j^\mathrm{S}(\hat{\S}_j) := \min _{\overline{s}-|\hat{\S}_j| \leq i \leq p-1-|\hat{\S}_j|}\left\{\frac{p-\overline{s}}{i-\overline{s}+1+|\hat{\S}_j|} P_{(i),j}(\hat{\S}_j)\right\}    
\end{align}
in place of $P_j^\mathrm{S}$ at each step; see \Cref{Algorithm: adaptive-Simes-Holm} in \Cref{app:fwer algs} for a full statement. The following result (proved in \Cref{appendix: proof of FWER-Simes}) establishes that this modified Holm's procedure with $P_j^\mathrm{S}(\hat{\S}_j)$ also controls the FWER under identical conditions as the original Holm's procedure with $P_j^\mathrm{S}$.

\begin{Theorem}\label{thm: FWER control}
Under \Cref{asm:mb}, if either $\overline{s}>s$ or $\overline{s}=p-1$, and for each $j\in [p]$ the base p-values $\{P_{i,j}:\ i\in \S^\mathsf{c}\setminus\{j\}\}$ are $\mathrm{PRDS}$, then \Cref{Algorithm: adaptive-Simes-Holm} controls the $\mathrm{FWER}$ at level $\alpha$.
\end{Theorem}

Informally, \Cref{thm: FWER control} requires the $P_{i,j}$'s that share the same $j$ to be positively dependent, and as was mentioned immediately after $P_j^{\mathrm{S}}$ was introduced in \Cref{eq:Simes pch p-val}, since $P_{i,j}$ and $P_{i',j}$ are testing analogous hypotheses on two pairs of covariates that overlap in one covariate ($X_j$), we may expect them to be at least somewhat positively aligned, and thus positively dependent. Furthermore, \Cref{thm: FWER control} places no assumptions beyond this, and in particular allows $P_{i,j}$ and $P_{j,i}$ to be arbitrarily strongly dependent, which is important because $P_{i,j}$ and $P_{j,i}$ are in fact testing the \emph{same} hypothesis (on the same data), and hence are likely to be extremely strongly dependent, and could even be equal. Allowing for these forms of dependence seems to be important, as an alternative method for FWER control with multiple PCH p-values, AdaFilter-Bonferroni \citep{adafilter}, assumes independence among the base p-values in order to prove validity and we found that it failed to control the FWER in our simulations; see \Cref{simulation adafilter}. In contrast, \Cref{Algorithm: adaptive-Simes-Holm} controlled the FWER empirically in all our simulations.
We note that \cite{bogomolov2018assessing} provide another method for FWER control with PCH p-values, but it is is only applicable when the number of base hypotheses for each PCH is 2, making it unsuitable for our setting.

\subsubsection{Variable selection controlling the false discovery rate}\label{subsubsec: FDR control}
To select a set $\hat{\S}$ controlling the \emph{false discovery rate} (FDR) at level $\alpha$, i.e., ${\E\left[ \frac{|\hat{\S} \setminus \S|}{\max\{|\hat{\S}|,1\}}\right]\le\alpha}$, one option is to apply the Benjamini--Yekutieli procedure~\citep{benjaminiyekutieli} to the Bonferroni BCP p-values $\{P_j^\mathrm{B}\}_{j=1}^p$, which will control the FDR under the conditions of \Cref{cor: validity of bonferroni p-val}.
However, it is much more powerful to instead apply the Benjamini--Hochberg procedure \citep{benjamini1995controlling} to the Simes BCP p-values $\{P_j^\mathrm{S}\}_{j=1}^p$, and it turns out that this is valid when the base p-values are PRDS.

\begin{Theorem}[FDR Control]\label{thm: FDR control}
Under \Cref{asm:mb}, if either $\overline{s}>s$ or $\overline{s}=p-1$, and the $p(p-1)$ base p-values $\{P_{i,j}\}_{i\neq j}$ are $\mathrm{PRDS}$, then \Cref{algorithm: Simes-BH-oracle} controls the $\mathrm{FDR}$ at level $\alpha$.
\end{Theorem}

Note that \Cref{thm: FDR control} would be a direct corollary of standard FDR control results~\citep{benjaminiyekutieli} if it assumed the \emph{BCP} p-values were PRDS, but instead it assumes the \emph{base} p-values are PRDS, and thus relies instead on recent results from \citet[Theorem 3.1]{TPH_under_dependence} (see \Cref{appendix: proof of FDR control} for complete proof). As discussed in the paragraph after \Cref{thm: FWER control}, there are multiple sources of (potentially strong) positive dependence among the base p-values that are important to allow for. As with FWER, another method for FDR control with PCH p-values, AdaFilter-BH~\citep{adafilter}, assumes the base p-values $P_{i,j}$ are independent, and fails to control the FDR in our simulations (see \Cref{simulation adafilter}), underscoring the importance of allowing positive dependence among the p-values in the theoretical guarantees of our methods. In contrast, \Cref{algorithm: Simes-BH-oracle} controls the FDR empirically in all our simulations. We note that another FDR control method for PCH tests by \cite{heller2014replicability} has a time complexity exponentially increasing with $p$, making it impractical for our setting, and the FDR control method of \cite{bogomolov2018assessing} is only applicable when the number of base hypotheses for each PCH is 2, making it unsuitable for our setting.

\subsection{Addressing sparsity in the covariates}\label{subsec: sparse covariates conditioning method}
An interesting phenomenon arises when some of the compositional covariates $X_{\{j\}}$ are sparse, i.e., $\mathbb{P}(X_{\{j\}}\neq 0)$ is small, as would be expected in, e.g., the microbiome, as certain microbial taxa may be completely absent from most observations. In the non-compositional setting, a high degree of sparsity in a covariate may make the power to identify \emph{that} covariate as important quite low, but does not generally have much impact on the power to identify \emph{other} (non-sparse) important covariates. Surprisingly, when the covariates are compositional, even having just one sparse covariate can make it nearly impossible to identify \emph{any} important covariates via BCP, even if they are not sparse at all. To see why, let $\epsilon_j := \mathbb{P}(X_{\{j\}}\neq 0)$. If $\epsilon_j=0$, then $X_{\{j\}} \stackrel{\text{a.s.}}{=} 0$ and for any $i\neq j$, ${\H_{i,j}: Y \indep X_{\{i,j\}} \mid X_{\{i,j\}^\mathsf{c}}}$ is equivalent to $\H_{i}: Y \indep X_{\{i\}} \mid X_{\{i\}^\mathsf{c}}$, and as discussed throughout Section~\ref{sec: introduction}, this $\H_i$ is definitionally true for all $i$ when the covariates are compositional, guaranteeing that $\S=\emptyset$ and trivializing the power of any BCP methods introduced in this section thus far. This is due to exactly the issue raised in \Cref{remark:trivariate}, and the problem with sparse covariates is that they approximate this situation by having $X_{\{j\}}=0$ for \emph{most} observations, effectively leading to most observations providing no evidence against $\H_{0i}$ for any $i$ and, thus, BCP methods having very low power to reject any $\H_{0i}$.


\Cref{remark:trivariate} alluded to addressing a similar issue via higher-order conditional independence hypotheses, such as trivariate conditional independence, but this required that the extra constraint(s) (beyond compositionality) be deterministic (and known a priori), which will not typically be the case with sparse covariates for which $\epsilon_j$ is small but strictly positive. 
Our proposed solution is to \emph{condition on} covariates that are expected a priori to be sparse, giving up on having any power to identify them as members of $\S$ (which would likely have been low anyway, due to their sparsity), but salvaging the power to identify the rest of $\S$. Formally,
suppose the covariate vector $X$ is split into two parts: $\D$, the indices of covariates expected to be dense, and $\D^\mathsf{c}$, the indices of covariates expected to be sparse. Define 
$$\S_\D = \{j: \H_{i,j} \text{ is false for all } i\in \D\setminus \{j\}\},$$ 
as the analogue of $\S$ if we simply condition on $X_{\D^{\mathsf{c}}}$ and only consider $X_\D$ as candidate covariates for rejection. By simply treating all our random variables as conditional on $X_{\D^{\mathsf{c}}}$ and treating $X_\D$ as our covariates (which, conditional on $X_{\D^{\mathsf{c}}}$, remain compositional, since they sum to $1-|X_{\D^{\mathsf{c}}}|$), all the results and BCP methods of this paper extend directly to (multiple or single) testing $\H_{0i}^\D: j\in \S_\D$. And as long as none of the elements of $\D$ are sparse, no bivariate conditional independence tests involving a sparse covariate are required (indeed, only $P_\D$, the $|\D|\times|\D|$ submatrix of $P$ containing $P_{i,j}$ for $i,j\in\D$, is needed when conditioning on $X_{\D^\mathsf{c}}$), so they cannot impact power. 

Ultimately, however, our interest remains in our original target, $\S$, so to interpret the results of our BCP procedures conditional on $X_{\D^\mathsf{c}}$, we must connect $\S_\D$ with $\S$. Indeed the following theorem shows that under similarly mild conditions as \Cref{thm: uniqueness of MB}, $\S_\D = \S\cap\D$, so that BCP methods conditioning on $X_{\D^\mathsf{c}}$ remain valid for our original target $\S$, though they will only have power to detect its dense elements, $\S\cap\D$.


\begin{Theorem}[Connecting $\S_\D$ and $\S$]\label{thm: validity of sparse covariates conditioning method}
If $\S=[p]$, then $\S_\D=\D$ ($=\S\cap\D$). Otherwise, if $|\S^\mathsf{c}\cap \D| \neq 1$ and $\S^\mathsf{c}\in (\Delta \circ)^{k} \I$ for some $k\in\mathbb{Z}_{\geq 0}$ then $\S_\D=\S\cap \D$ which is the intersection of the unique nontrivial Markov boundary with $\D$.
\end{Theorem}

The proof of this theorem is presented in \Cref{sec: proof of theorem validity of sparse covariates}. In addition to the conditions in \Cref{thm: uniqueness of MB}, there is one additional condition in \Cref{thm: validity of sparse covariates conditioning method}: $|\S^\mathsf{c}\cap \D| \neq 1$. When $|\S^\mathsf{c}\cap \D| = 1$ (call the element in this intersection $j^\star$), the conditional compositionality of $X_\D$ results in $\S_\D=\D$ rather than the desired $\D\setminus\{j^\star\} =\S\cap\D$, so we must assume this case does not occur. 

From a methodological perspective, note that BCP procedures conditional on $X_{\D^\mathsf{c}}$ can not only increase power (when $X_{\D^\mathsf{c}}$ are sparse) over their unconditional counterparts, but also save on computation, as they only require computing bivariate conditional independence tests for pairs of elements of $\D$.

\subsection{Obtaining base p-values \texorpdfstring{$P_{i,j}$}{Pij}}\label{sec: obtaining p-values}
This section so far has described methods for testing and variable selection that take as input the bivariate conditional independence test p-values $P_{i,j}$ ($i\neq j$), but has left the choice of test to the user. As the topic of conditional independence testing is well-studied, we view this as a feature of our methods that they are flexible to this choice of test, as well as the assumptions under which the test is valid. For further details on methods to test conditional independence, and the assumptions under which they are valid, see, e.g., \citet{doran2014permutation,sen2017model,candes2018panning,shah2020hardness,berrett2020conditional,lundborg2022projected,dCRT,tansey2022holdout, kim2022local,shi2024azadkia}. 

While our position on which conditional independence test to use with our proposed BCP methods remains an agnostic one, we are forced to make a choice when we run our simulations; we choose the Conditional Randomization Test \citep{candes2018panning}, and in particular a computationally efficient variant called the distilled conditional randomization test (dCRT) with a lasso-based test statistic \citep{dCRT}. \Cref{appendix: using dCRT to obtain p-values} contains implementation details for the dCRT for bivariate conditional independence testing as it is used in our simulations.

We note here that while we originally assumed that the $P_{i,j}$ were computed from a dataset of $n$ i.i.d. samples of $(Y,X)$, nothing in our theory relied on this assumption. Indeed, many conditional independence tests, including the dCRT, can be readily extended to settings in which the samples are neither independent nor identically distributed (see \Cref{appendix: using dCRT to obtain p-values} for more discussion of this topic for the dCRT), and hence the BCP methods proposed in this paper will immediately apply in such settings as well.


\subsection{Computational speedups}\label{sec: computational speedups}
BCP variable selection methods require computing $p^2$ (or $|\D|^2$ when conditioning on $X_{\D^\mathsf{c}}$) base p-values, which can be quite expensive, especially in high dimensions, although it is trivially and fully parallelizable. Nevertheless, to help alleviate this issue, we discuss in this subsection some computational shortcuts for BCP methods.

Following the idea of `screening' in \citet{dCRT}, any subset of the $P_{i,j}$'s can be set to 1 without impacting the marginal validity of (all) the $P_{i,j}$'s (only making them conservative), even if this subset is chosen based on the entire data set. This can lead to a substantial computational speedup if there is a relatively inexpensive way to identify which $P_{i,j}$'s are likely to be large, since then these $P_{i,j}$'s can simply be set to 1 without computing them. And if the $P_{i,j}$'s set to 1 were indeed going to be large anyway (e.g., $>0.1$), then making them larger isn't likely to reduce power much, if at all, since any such base p-value would have been too large to contribute to a rejection anyway, even if fully computed. 

Such a subset of the $P_{i,j}$'s could be chosen via sparse regression of the response on the covariates, e.g., as the zero elements of the fitted coefficient vector of a single run of cross-validated LASSO. Another option, which can be used in conjunction with the aforementioned sparse regression approach, is possible if the $P_{i,j}$'s are being computed serially: when about to compute a given $P_{i,j}$, if enough other base p-values, $P_{i',j}$, in the same column have been computed already and are large, then there will already be no hope of the PCH p-value combining the base p-values of that column being small, and hence it is faster and no less powerful to simply set $P_{i,j}$ (and any other as-yet-uncomputed base p-values in the same column) to 1 without wasting computation on it.

In some cases, further speedups are possible within the computation of the conditional independence tests themselves; we detail one such for the dCRT in \Cref{appendix: simulation computational speedups} which also takes the form of setting some $P_{i,j}$'s to 1 without (fully) computing them, i.e., another form of data-dependent screening.

While data-dependent screening provably retains marginal validity of the $P_{i,j}$'s \citep{dCRT}, it does in general change the dependence structure of the base p-values. This will have no impact on the validity of our BCP procedures that do not assume anything about the dependence among the base p-values (e.g., Holm's procedure applied to Bonferroni PCH p-values), but could in principle impact the validity of BCP procedures that assume a form of positive dependence. However, we found all the speedups used in our simulations never led to empirical error control violations, and in fact led to nearly identical sets of rejections as the same procedures without any speedups (average Jaccard indices $>98\%$ in each setting we tried; see \Cref{appendix: simulation computational speedups}). Furthermore, putting these speedups together made a big computational difference, with about a $5\times$ speedup with $n=p=100$, and more than a $30\times$ speedup for $n=p=1000$; see \Cref{appendix: simulation computational speedups} for details.

\section{Numerical Experiments}\label{sec: simulations}
In this section we demonstrate our methods' performance in compositional settings (\Cref{subsec: comp-data-simulation}), non-compositional settings (for comparison with non-compositional methods; \Cref{subsec: non-compositional sims}), and compositional settings with sparse covariates as described in \Cref{subsec: sparse covariates conditioning method} (\Cref{subsec:  simulation conditioning on sparse covariates}). 
The appendix contains many more simulations, details, and plots, in particular: additional simulations to assess robustness under covariate model misspecification (\Cref{subsec: robustness study}), evaluate the performance of the multiple PCH testing method AdaFilter (\Cref{simulation adafilter}), and examine the impact of computational speedups on validity and average power  (\Cref{appendix: simulation computational speedups}). For each simulation presented in the main text, we include only key, representative figures to demonstrate the main findings, with additional figures and results provided in Appendices~\ref{appendix: simulation dirichlet}--\ref{appendix: simulation effect of conditioning}. Code to replicate all numerical experiments is available at \url{https://github.com/Ritwik-Bhaduri/Compositional_Covariate_Importance_Testing/}.

While \Cref{sec: methodology} considered combining base p-values into PCH p-values via both Bonferroni's and Simes' method, the simulations in this section will focus only on PCH p-values combined via Simes' approach (i.e., $P^{\mathrm{S}}_j$), while we provide all the same results for Bonferroni's approach in \Cref{appendix: simulation bonferroni}. 
This is because BCP methods using the Simes' approach have significantly more power than those using Bonferroni's while still empirically controlling their respective error rates well. See \Cref{appendix: simulation bonferroni} for more details on the comparison between the Bonferroni and Simes approaches within BCP.


We compare three different values of $\overline{s}$ for our methods: $p-1$, $p/2$, and $s+1$ (in all our simulations, $s+1<p/2$). These values represent increasing levels of knowledge about the sparsity of the data: $p-1$ indicates no knowledge of sparsity, $s+1$ indicates perfect knowledge, and $p/2$ indicates a moderate (and we argue realistic in many settings; see discussion in \Cref{sec:testing single covariate}) level of knowledge by assuming at least half of the covariates are irrelevant. 
Given $\overline{s}$, we call the BCP hypothesis test (with Simes' approach, i.e., $P_j^{\mathrm{S}}$) BCP($\overline{s}$), and for variable selection, we call our methods BCP($\overline{s}$)-BH (\Cref{algorithm: Simes-BH-oracle}) and BCP($\overline{s}$)-Holm (\Cref{Algorithm: adaptive-Simes-Holm}) for FDR and FWER control, respectively. In \Cref{subsec: comp-data-simulation} and \Cref{subsec: non-compositional sims}, we compare BCP methods to case-specific benchmark methods described in those respective subsections.

\subsection{Standard compositional covariate settings}\label{subsec: comp-data-simulation}

In these simulations, the rows of the covariate matrix are generated as independent copies of a compositional random variable $X\in \mathbb{R}^{100}$. 
We consider two different models of $X$ to show the validity of our method for different dependence structures among the compositional covariates. 
First we consider the Dirichlet distribution where $X\sim\text{ Dirichlet}(\alpha_1, \dots, \alpha_{100})$ with $\alpha_1= \dots = \alpha_{100}=2$. The other compositional distribution we consider is the logistic-normal distribution \citep{xia2013logistic}, where $X_{{j}} = \frac{e^{Z_{{j}}}}{\sum_{i=1}^{100} e^{Z_{{i}}}}$ for $j\in \{1,\dots,100\}$ and $Z\sim\mathcal{N}(\boldsymbol{0},\Sigma)$ with $\Sigma_{ij} = 0.6^{|i-j|}$.

The response $Y$ is generated from a linear model $\mathcal{N}(\log(X)\beta, 1)$. Our simulated data sets consist of $n=100$ i.i.d. samples from $(X,Y)$. 
For single testing, the coefficient of the variable being tested is varied on the x-axis of our plots, while the non-null nuisance coefficients are i.i.d. $\mathcal{N}(0,1)$. For variable selection, all non-null coefficients are $\mathcal{N}(0,\mathrm{SNR}^2)$, where the scalar value SNR is varied on the x-axis, and can be thought of as a signal-to-noise ratio as it modulates the signal in the non-nulls. In all these simulations, there is a unique non-trivial Markov boundary, also equal to $\S$, which is given by the set of non-zero coefficients; its size $s$ is always 10.

As we are not aware of any existing methods that can test for or select elements of the Markov boundary with error control guarantees when the covariates are compositional, we compare BCP methods to a benchmark that is invalid, but popular and powerful. Namely, we consider as a benchmark a Leave-One-Out (LOO) approach that first drops a column $j$ uniformly at random from the data, and then analyzes the resulting non-compositional data via a standard (i.e., non-compositional) conditional independence test directly analogous to the bivariate conditional independence tests used within the BCP methods, namely, the dCRT with a lasso-based test statistic (see \Cref{appendix: simulation dirichlet} for more details). For variable selection with FWER control or FDR control, we plug the $p-1$ p-values from this LOO approach into Holm's procedure (LOO-Holm) or the Benjamini--Hochberg procedure (LOO-BH), respectively.
This method is consistent with heuristics currently used in practice (like contrast coding) and is powerful. However, it lacks error control, as discussed near the end of \Cref{sec: related work}.

\begin{figure}[!]
    \centering
    \begin{subfigure}{\textwidth}
        \centering
        \includegraphics[width=\textwidth]{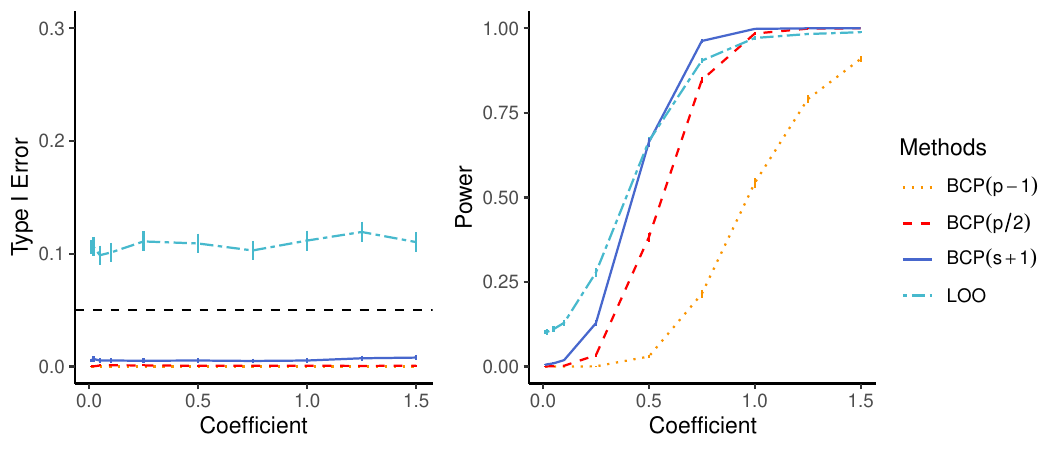}
        \caption{Comparison of type I error and power for single testing with Dirichlet covariates. The target type I error is 5\% and error bars correspond to $\pm 2$ Monte Carlo standard errors.}
        \label{fig:single_test_dirichlet}
    \end{subfigure}
    \begin{subfigure}{\textwidth}
        \centering
        \includegraphics[width=\textwidth]{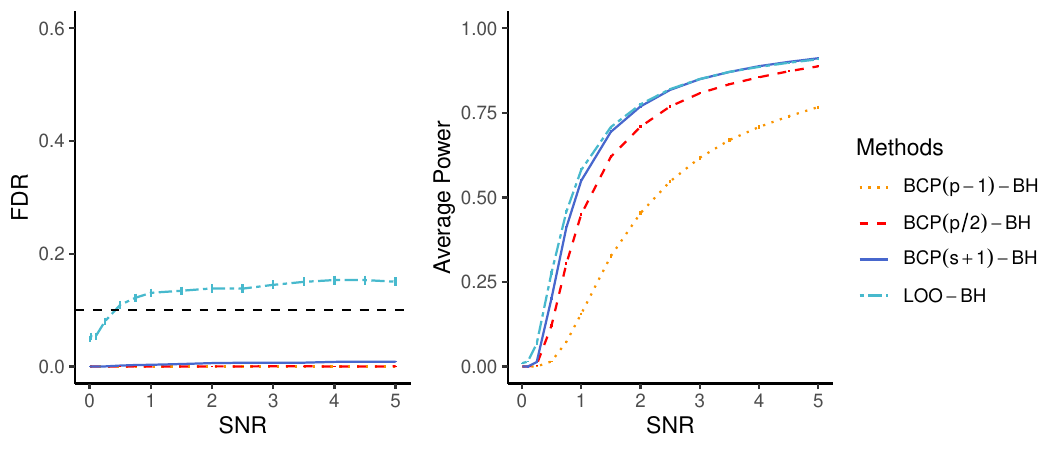}
        \caption{Comparison of FDR and average power for variable selection with Dirichlet covariates. The target FDR level is 10\% and error bars correspond to $\pm 2$ Monte Carlo standard errors.}
        \label{fig:fdr_dirichlet}
    \end{subfigure}%
    \caption{Comparison of methods for (a) single testing and (b) variable selelction with Dirichlet covariates. 
    }
\end{figure}
The left panel of \Cref{fig:single_test_dirichlet} shows that for single testing, the LOO method fails to control type I error at the nominal level of $5\%$, while the BCP methods all control type I error well below 5\%. Indeed, BCP methods were conservative across simulations due to their use of PCH p-values, which are well-known to be conservative
\citep{adafilter}. The right panel of \Cref{fig:single_test_dirichlet} shows the power of the BCP($\overline{s}$) methods decreases with the value of $\overline{s}$, as expected. BCP($s+1$), which assumes complete knowledge of $s$, achieves power similar to the LOO method and, in fact, beats the latter for higher signal strengths. The more realistic method---BCP($p/2$)---achieves power close to the oracle method BCP($s+1$) and LOO, and well above that of BCP($p-1$). \Cref{fig:fdr_dirichlet} shows analogous performance of BCP methods for variable selection:
LOO-BH fails to control FDR, while BCP($\overline{s}$)-BH methods are conservative but have average power similar to LOO-BH. Appendix \ref{appendix: simulation dirichlet} shows the results for FWER control are similar to those for FDR control. Furthermore, Appendix \ref{appendix: simulation logistic-normal} shows that both single testing and variable selection with logistic-normal covariates behave broadly similarly as with Dirichlet covariates. The only notable difference is a substantial drop in power for BCP($p-1$) (but not the other BCP methods), which could be attributed to the additional dependence structure in the data leading to the base p-values \(P_{i,j}\), where \(i\) is null and \(j\) is non-null, being less concentrated around zero. This effect would have more impact on the power of BCP methods with higher \(\overline{s}\).



\subsection{Non-compositional covariate settings}\label{subsec: non-compositional sims}

While BCP methods are primarily designed and recommended for regression problems with compositional covariates, nothing in their development or theory precludes their use on non-compositional covariates. In such a setting, we can benchmark them against state-of-the-art non-compositional methods that are both valid and powerful (unlike in the previous subsection, where there was no available valid benchmark).
In this subsection's simulation, $X\in\mathbb{R}^{100}$ is generated from $\mathcal{N}(\boldsymbol{0}, \Sigma)$ where $\Sigma$ is a Toeplitz matrix such that $\Sigma_{i,j} = 0.6^{|i-j|}$. The rest of the simulation setup including the conditional distribution of the response given the covariates and the sample size remain the same as in \Cref{subsec: comp-data-simulation}. 

As a benchmark, we consider a standard (i.e., non-compositional) conditional independence test directly analogous to the bivariate conditional independence tests used with the BCP methods, namely, the dCRT with a lasso-based test statistic. This Univariate testing approach is valid and powerful for testing whether a covariate is a member of the Markov boundary when the covariates are non-compositional, and for variable selection with FWER control or FDR control, these test's p-values can be plugged into Holm's procedure (Univariate-Holm) or the Benjamini--Hochberg procedure (Univariate-BH), respectively.

\begin{figure}[!]
    \centering
    \begin{subfigure}{\textwidth}
        \centering
        \includegraphics[width=\textwidth]{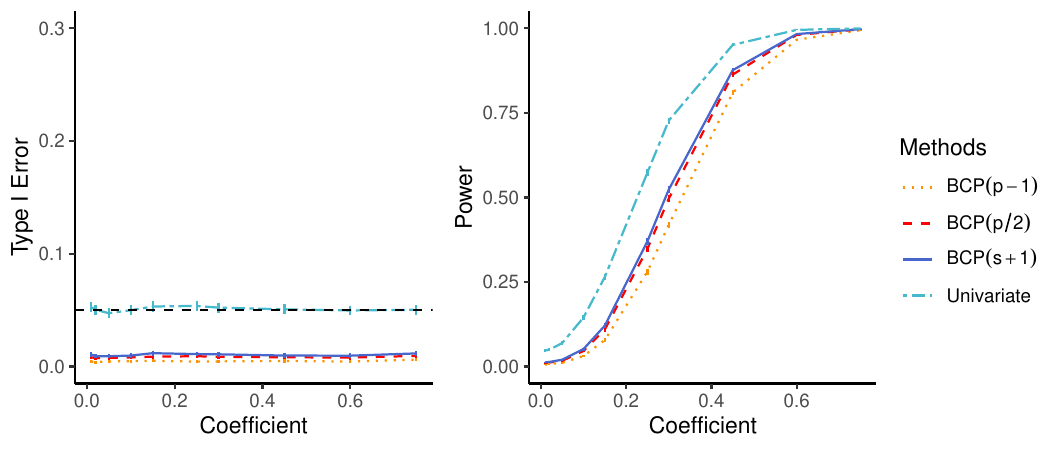}
        \caption{Comparison of type I error and power for single testing with multivariate normal covariates. The target type I error is 5\% and error bars correspond to $\pm 2$ Monte Carlo standard errors.}
        \label{fig:single_test_normal}
    \end{subfigure}
    \begin{subfigure}{\textwidth}
        \centering
        \includegraphics[width=\textwidth]{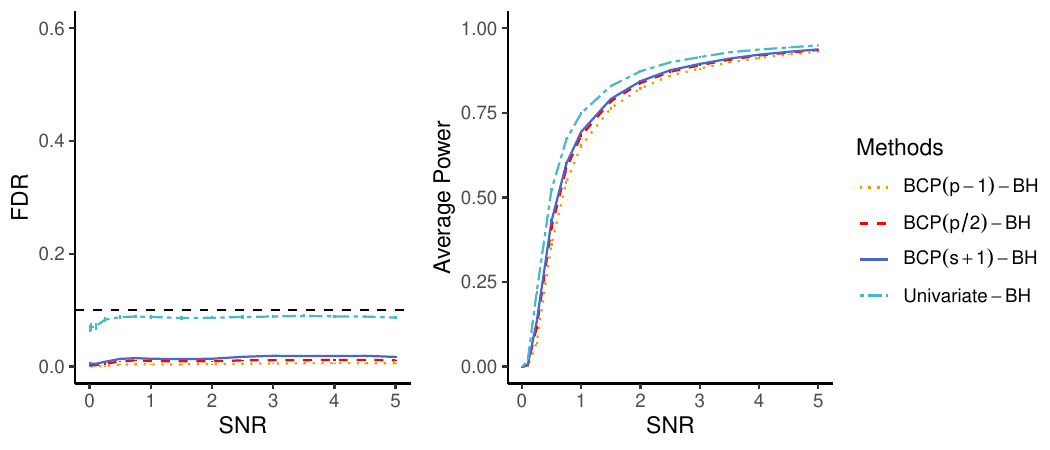}
        \caption{Comparison of FDR and average power for variable selection with multivariate normal covariates. The target FDR level is 10\% and error bars correspond to $\pm 2$ Monte Carlo standard errors.}
        \label{fig:fdr_normal}
    \end{subfigure}%
    \caption{Comparison of methods for (a) single testing and (b) multiple testing with multivariate normal covariates. 
    }
\end{figure}

\Cref{fig:single_test_normal} shows that all single-testing methods control the type I error and the non-compositional Univariate test
has the highest power, closely followed by the proposed BCP methods. 
\Cref{fig:fdr_normal} shows the same pattern for FDR-controlled variable selection except the gap between average power of the BCP methods and Univariate-BH is even narrower; \Cref{appendix: simulation multivariate normal} shows very similar results for FWER control. 
Thus, despite all the methodological overhead in the BCP approach that is needed to accommodate compositional covariates (and the significant corresponding conservativeness that comes with it), we find that in non-compositional settings, BCP methods are nevertheless not far behind state-of-the-art conditional independence tests in terms of power, with particularly competitive variable selection average power. 

\subsection{Sparse compositional covariate settings}\label{subsec:  simulation conditioning on sparse covariates}

To explore the BCP approach to sparse compositional covariates as presented in \Cref{subsec: sparse covariates conditioning method}, we generate covariates in this section 
from a 100-dimensional Dirichlet-Multinomial distribution as follows:
let $\alpha\in \mathbb{R}^{100}$ have entries $\alpha_j = (1+e^{\frac{j-50}{5}})^{-1}$ and sample $\theta \sim$ Dirichlet$(\alpha)$ and $X\sim$ Multinomial$(200, \theta)$. This $X$ distribution overall has about half of its entries non-zero but the sparsity is highly imbalanced, with the last entries much sparser than the first ones.
The response $Y$ is generated from a log-linear model $\mathcal{N}\left(\log(1+X)\beta,1\right)$, where the coefficient vector $\beta$ is zero except for 10 uniformly randomly selected entries which are sampled independently from $\mathcal{N}(0,1)$. The sample size was $n=100$.

The goal of this simulation is to study the effect of $\Dense$ on power, and since the density of the covariates decreases with increasing index, we choose $\Dense = [k]$ and vary $k$ from 30 to 100. 
In this subsection, we will consider BCP methods with the following three values of $\overline{s}$: $\PDense-1$, $p/4 (=25 <30\le \PDense)$, and $s+1$; this is because the conditional BCP approach of \Cref{subsec: sparse covariates conditioning method} requires $\overline{s}<\PDense$.

\begin{figure}[!]
\centering
\includegraphics[width=\textwidth]{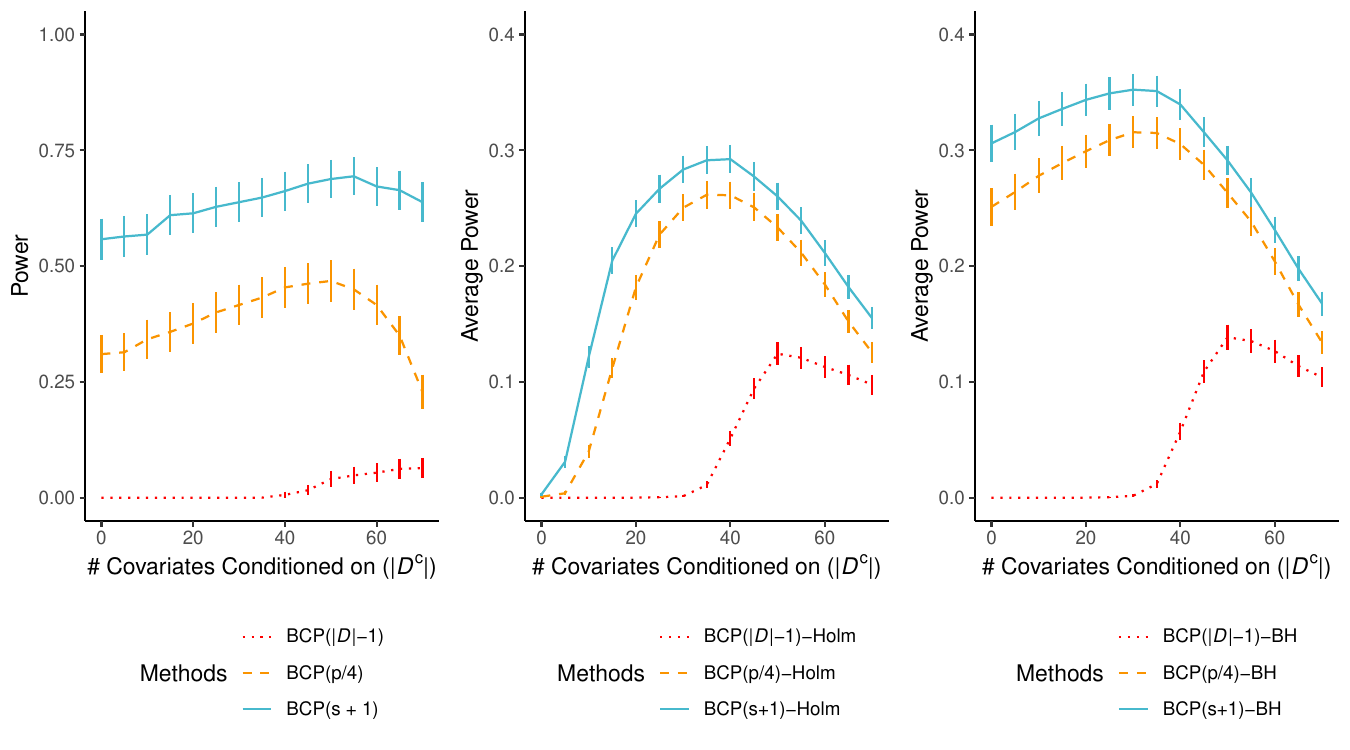} 
\caption{Effect of conditioning out the sparse covariates on power for single testing (left) and average power for FWER control (middle) and FDR control (right), respectively. Error bars correspond to $\pm2$ Monte Carlo standard errors.}
\label{fig:power_effect_of conditioning}
\end{figure}

\Cref{fig:power_effect_of conditioning} shows that conditioning on sparse covariates as proposed in \Cref{subsec: sparse covariates conditioning method} can yield substantial gains in power over not doing so (i.e., $|\Dense^\mathsf{c}|=0$) in both single testing and variable selection. All methods in all settings benefited from conditioning on at least the 35 sparsest covariates (the 35th-sparsest covariate is 8\% non-zero), and in some cases continued to benefit from conditioning on more than half the covariates. However, as expected, for nearly all methods there is a sweet spot, and conditioning on too many covariates starts to harm power as covariates that are conditioned on cannot be discovered (recall that in these simulations, the probability that a covariate is non-null is independent of its sparsity). \Cref{appendix: simulation effect of conditioning} shows the corresponding type I error, FWER, and FDR plots, all of which are controlled below their nominal levels.

\section{Discussion}\label{sec: discussion}
In this work, we have defined an interpretable model-free notion of an important covariate in regression problems with compositional covariates. With this notion as an inferential target, we then developed novel, valid, and powerful procedures for testing and controlled variable selection, integrating ideas from conditional independence testing and partial conjunction hypothesis testing.

While this paper has primarily focused on compositional covariates, its contents in fact apply more generally to covariates that satisfy any single deterministic equality constraint.
For example, the methods and theorems in this paper remain valid when the covariates are constrained to any $(p-1)$-dimensional linear subspace. 
Another allowable constraint is that the covariates lie on the surface of the ($(p-1)$-dimensional) unit $p$-sphere \citep{lin2019nonparametric}; such data arises in, e.g., 
geophysics~\citep{freeden1998constructive}, cosmology~\citep{dodelson2020modern}, plate tectonics~\citep{chang2000regression} and embryology~\citep{tyszka2005new}. 

A natural future direction for this work is the setting of multiple constraints, i.e., when the covariates are constrained to lie in a set of dimension $p-k$ for $k>1$. This setting arises in, e.g., feature engineering~\citep{zheng2018feature}, basis expansion~\citep{oreshkin2019n}, multi-factor experimental designs~\citep{box1957multi}, and experimental designs with interactions~\citep{dean1999design}. \Cref{subsubsec: simplification for factor covariates} showed that this paper's approach can already handle multi-factor experimental design 
since the constraints are on disjoint subsets of the indices, but for other type of multiple constraints, na\"{i}ve extension of the ideas in this paper would lead to computational intractability as discussed in \Cref{remark:trivariate}, with $k$ constraints requiring testing of $\mathcal{O}\left(p^{k+1}\right)$ $k$-variate conditional independence hypotheses. We hope that further insights and improvements along the lines of the computational speedups in \Cref{sec: computational speedups} can be discovered in future work to make such approaches tractable.

\section*{Acknowledgments}
RB and LJ were partially supported by DMS-2045981.

\bibliographystyle{apalike}
\bibliography{reference}

\appendix

\section{Differential Absolute Abundance Methods}
\label{app:absabund}
As mentioned in a footnote in \Cref{sec: related work}, there is a class of differential abundance methods in the microbiome field that test for marginal independence between $Y$ and the \emph{absolute} abundance of the $j^{\text{th}}$ microbe \citep{zhou2022linda, wang2023robust,lin2024multigroup,zong2024mbdecoda}. We refer to such methods collectively as differential absolute abundance (DAA) to distinguish them from canonical DA analysis. To understand how DAA methods relate to this paper, suppose that the observed compositional covariates $X$ represent the normalization of a latent noncompositional random vector $Z$, i.e., $X = Z/\sum_{j}Z_{\{j\}}$. Then DAA methods are designed to test hypotheses of the form $Y\indep Z_{\{j\}}$. Because DAA methods resemble DA methods but with $X_{\{j\}}$ replaced by $Z_{\{j\}}$, they are in some sense a step further removed from this paper than DA methods, which is why we defer their discussion to this appendix. 

First and foremost, DAA methods test a hypothesis in terms of the latent random variables $Z_{\{j\}}$. Thus, it represents fundamentally different \emph{scientific questions} than considered by DA methods or any of the methods in this paper, which are entirely in terms of the observable random variables $Y$ and $X$. In fact, $Z$ is only a meaningful quantity in certain applications where $X$ represents the normalization of some underlying random variable (for example, the microbiome field). Many other compositional random variables (e.g., factors in experimental design, chemical compositions, or time use within a day), however, are compositional but do not represent normalization from underlying noncompositional variables. There is thus no natural scientific question regarding the latent $Z_{\{j\}}$'s for these types of applications.

Beyond the conceptual difference, DAA methods differ from our approach in two additional, important ways. First, like DA methods, DAA methods test \emph{marginal} independence, meaning that unlike the multivariate approach taken in this paper, they do not account for any covariates other than the $j^{\text{th}}$ in their hypothesis. Second, the latency of $Z$ will generally make it (and hence nearly any hypothesis involving it) unidentifiable except under specific assumptions: most DAA methods assume highly parametric distributions on $Z_{\{j\}}$ and $X_{\{j\}}$; all require some form of sparsity conditions that not too many $Z_{\{j\}}$'s can be associated with $Y$. As such, DAA methods also differ from this paper in taking a parametric approach to inference, necessitated in part by the latency of $Z$.

\section{Proof of Theorems}
\subsection{Proof of Lemma~\ref{Lemma: 1 equiv class}}\label{subsec: proof of Lemma: 1 equiv class}
We are first going to prove Lemma \Cref{Lemma: 1 equiv class} holds for coordinatewise connected vectors (through $A\setminus B$ and $B\setminus A$). The proof for equivalent vectors will follow similarly.

Consider two vectors $(u,v,w)$ and $(u^\prime, v^\prime, w^\prime)$ in the support of $X_{A\cup B}$, conditioned on $X_{(A\cup B)^\mathsf{c}}=z$, that are coordinatewise connected through $A\setminus B$ and $B\setminus A$. Here, $(u,v,w)$ is shorthand for $X_{A\cap B} = u$, $X_{A\setminus B} = v$, and $X_{B\setminus A} = w$.

We will use the notation $p(\cdot \mid \cdot)$ to represent the conditional measure of $Y$ given a subvector of $X$. This notation is general and accommodates various scenarios:
\begin{enumerate}
    \item Conditional probability in the discrete case, i.e., $ \mathbb{P}(Y=y \mid X_{A\cup B}=(u,v,w)) $.
    \item Conditional density function in the continuous case, i.e., $f_Y(y \mid X_{A\cup B}=(u,v,w))$.
    \item More generally, using measure-theoretic notation, for any Borel set $\mathscr{B}$ in the sigma-algebra of $Y$ and the event $ X_{A\cup B}=(u,v,w) $, $ p(\mathcal{\mathscr{B}} \mid X_{A\cup B}=(u,v,w)) $ represents the conditional expectation $ \mathbb{E}[1_{\mathscr{B}} \mid \mathcal{F}_{X_{A\cup B}}] $ where $ \mathcal{F}_{X_{A\cup B}} $ is the sigma-algebra generated by $ X_{A\cup B} $. This accommodates cases where our random variables might not have purely discrete or continuous distributions.
\end{enumerate}

Without loss of generality (WLOG), assume the vectors are coordinatewise connected through $B\setminus A$. Then, $w = w^\prime$. So, we have,
\begin{align*}
        &p\left(Y \mid  \ X_{(A\cap B)}=u,X_{(A\setminus B)}=v, X_{(B\setminus A)}=w, X_{(A\cup B)^\mathsf{c}}=z\right)\\
        &=p\left(Y \mid  X_{(B\setminus A)}=w, X_{(A\cup B)^\mathsf{c}}=z\right) \quad \text{(by the conditional independence } Y\indep X_{A}\mid X_{A^\mathsf{c}}\text{)}\\
        &=p\left(Y \mid  X_{(B\setminus A)}=w^\prime, X_{(A\cup B)^\mathsf{c}}=z\right) \quad \text{(as } w=w^\prime\text{)}\\ 
        &=p\left(Y \mid  \ X_{(A\cap B)}=u^\prime,X_{(A\setminus B)}=v^\prime , X_{(B\setminus A)}=w', X_{(A\cup B)^\mathsf{c}}=z\right).
\end{align*}
Instead of the vectors being coordinatewise connected through $B\setminus A$, if they were connected through $A\setminus B$, we would have used the conditional independence $Y\indep X_{B}\mid X_{B^\mathsf{c}}$. So, it follows that $p(Y\mid u,v,w,z) = p(Y\mid u',v',w',z)$ whenever the vectors $(u,v,w)$ and $(u',v',w')$ from the support of $X_{A\cup B}$ conditioned on $X_{(A\cup B)^\mathsf{c}}=z$ are coordinatewise connected (via $A\setminus B$ and $B\setminus A$).

Now suppose $(u,v,w)$ and $(u^\prime,v^\prime,w^\prime)$ are equivalent (through $A\setminus B$ and $B\setminus A$). So there exist points $(u_i, v_i, w_i),\ i = 1, \dots, L$ for some $L$ such that $((u,v,w), (u_1,v_1,w_1))$, $((u_1,v_1,w_1), (u_2,v_2,w_2))$, $\dots$, $((u_{L-1},v_{L-1},w_{L-1}), (u_L,v_L,w_L))$ and $((u_L,v_L,w_L), (u^\prime,v^\prime,w^\prime))$ are coordinatewise connected. From the above argument it follows that 
\begin{align*}
    p\left(Y\mid u,v,w,z\right) = p\left(Y\mid u_1,v_1,w_1,z_1\right) = \dots = p\left(Y\mid u_L,v_L,w_L,z_L\right) = p\left(Y\mid u',v',w',z\right).
\end{align*}
This completes the proof.

\subsection{Proof of Theorem~\ref{thm: uniqueness of MB}}\label{subsec: proof of uniqueness Theorem}

We begin by proving the following result first.

\begin{proposition}\label{prop: uniqueness of MB 2}
    Suppose there exists $i\neq j$ such that $Y\indep X_{\{i,j\}}\mid X_{\{i,j\}^\mathsf{c}}$. If $\S^\mathsf{c}\in (\Delta \circ)^k \I$ for some $k\in \mathbb{Z}_{\geq 0}$ then there exists a unique set $\M$ which satisfies the following:
    \begin{enumerate}
        \item $Y\indep X_{\M^\mathsf{c}}\mid X_{\M}$
        \item For all $\M^\prime$ such that $Y\indep X_{{\M^\prime}^\mathsf{c}}\mid X_{{\M^\prime}}$, $|\M|\leq |\M^\prime|$
    \end{enumerate}
    Further, this unique set is given by $\M=\S$.
\end{proposition}

\begin{proof}
Recall that we defined $\mathcal{I}$ as $\I = \{\{i, j\} : \H_{i,j} \text{ is true}\}$ and $\S$ was defined as $\S = \{i \in [p] : \H_{i,j}$ \text{ is false for all} $\ j \in \{i\}^{\mathsf{c}}\}$. Finally, $\Delta = \{\{A,B\} : X_{A \cup B} \mid  X_{(A \cup B)^\mathsf{c}} \text{ has 1 equivalence class (through } A\setminus B \text{ and } B\setminus A)\}$. If $\H_{i,j}$ is true for some $i\neq j$, $\I$ is non-empty. 

Consider the case $|\I| = 1$, where $\I = \{\{i,j\}\}$. Define $\M = \{i,j\}^{\mathsf{c}}$. $\M$ satisfies condition \textit{1} by definition of $\I$. Now for assumption \textit{2}, consider any $\M^\prime$ such that $Y\indep X_{{\M^\prime}^\mathsf{c}}\mid X_{{\M^\prime}}$. If $|\M^\prime| < |\M| = p-2$, for any three distinct elements of ${\M^\prime}^\mathsf{c}$, $i^\prime,j^\prime,\text{ and }k^\prime $, by the weak union property~\citep{PEARLbook}, $Y\indep X_{\{i^\prime,j^\prime,k^\prime\}}\mid X_{\{i^\prime,j^\prime,k^\prime\}^\mathsf{c}}$. Using the weak union property again, $\{i^\prime,j^\prime\}, \ \{j^\prime,k^\prime\},\ \{k^\prime,i^\prime\} \in \I$ which contradicts $|\I|=1$. So, $\M$ satisfies both conditions \textit{1} and \textit{2}. To show that no other set satisfies these conditions, consider any other set $\M^\prime$ which satisfies them. Due to the previous argument $|\M^\prime| \geq p-2$. Also, $|\M^\prime|< p-1$ since $\M$ satisfies condition \textit{2}. If $|\M^\prime| = p-2$ but $\M^\prime \neq \M$, then this implies $\M^{\prime\mathsf{c}}\in \I$ and $\M^{\prime\mathsf{c}} \neq \{i,j\}$ which again contradicts $\I$. Therefore, when $|\I| = 1$ and $\I = \{\{i,j\}\}$, $\M = \{i,j\}^{\mathsf{c}}$ is the unique set satisfying conditions \textit{1} and \textit{2}. Also, in this case, $\S = \{i,j\}^{\mathsf{c}}$ making $\M=\S$ which proves the statement of the proposition for the case $|\I|=1$. 

In the case $|\I|\geq 2$, consider any set $C\in \I$. Note that $C$ is a set of two elements. For such a $C$, $Y \indep X_C \mid  X_{C^\mathsf{c}}$. If we take any two sets $C$ and $D$ in $\I$ such that $X_{C \cup D} \mid  X_{(C \cup D)^\mathsf{c}}$ has 1 equivalence class (through $C \setminus D$ and $D \setminus C$), then \Cref{Lemma: 1 equiv class} implies the following is true:
$$Y \indep X_{C \cup D} \mid  X_{(C \cup D)^\mathsf{c}}.$$
Since $\Delta$ is defined as all the sets of sets with one equivalence class, the same logic implies that for any set $C \in \Delta \circ \I$, $Y\indep X_C \mid X_{C^\mathsf{c}}$. This argument can be applied to elements of $\Delta \circ \I$ also to get the following, 
$$Y\indep X_C \mid X_{C^\mathsf{c}} \ \text{for all }C \in (\Delta\circ)^2 \I.$$
Proceeding recursively, we get, for any $k \in \mathbb{Z}_{\geq0}$,
$$Y\indep X_C \mid X_{C^\mathsf{c}} \ \text{for all }C \in (\Delta\circ)^k \I.$$
So, if $\S^\mathsf{c} \in (\Delta\circ)^k \I$ for some integer $k$, then $Y\indep X_{\S^\mathsf{c}} \mid X_{\S}$ and $\S$ satisfies condition \textit{1}, and thus $|\M| \leq |\S|$ for any $\M$ satisfying conditions \textit{1} and \textit{2}. Now, conditions \textit{1} and \textit{2} imply $\M$ is a Markov boundary. Hence, by \Cref{Lemma: M supset S}, for any such $\M$, $\S \subseteq \M$. Combining them, we get under the assumptions of the proposition, $\M = \S$ is unique set satisfying conditions \textit{1} and \textit{2}.
\end{proof}

Now that we have proven the proposition, we shall use it along with \Cref{Lemma: M supset S} to prove \Cref{thm: uniqueness of MB}. Under the assumptions of \Cref{thm: uniqueness of MB}, we have two cases: $\S = [p]$ and $|\S| < p$. If $\S = [p]$, then $Y\nindep X_{\{i,j\}}\mid X_{\{i,j\}^\mathsf{c}}$ for all $i\neq j$. For the sake of contradiction, suppose there exists a nontrivial Markov boundary, i.e., some set $\M$ such that $|\M|<p-1$ and $Y\indep X_{\M^\mathsf{c}}\mid X_{\M}$. Then there exists some $\{i,j\}\subseteq \M^\mathsf{c}$ such that $i\neq j$. Then, the weak union property~\citep{PEARLbook} implies $Y\indep X_{\{i,j\}}\mid X_{\{i,j\}^\mathsf{c}}$ which contradicts $\S=[p]$. So, in this case there does not exist a nontrivial Markov boundary. 

In the other case, that is when $|\S| < p$, there must exist some $i\neq j$ such that ${Y\indep X_{\{i,j\}}\mid X_{\{i,j\}^\mathsf{c}}}$. 
Now, consider any subset $\M^\prime$ of $[p]$ which is a Markov boundary and $|\M^\prime| < p-1$. So, $Y\indep X_{{\M^\prime}^\mathsf{c}} \mid X_{{\M^\prime}}$. \Cref{Lemma: M supset S} and $|{\M^\prime} ^\mathsf{c}| > 1$ imply $\M^\prime \supseteq \S$. So, any nontrivial Markov boundary must be a superset of $\S$. The only thing left to show is that it is indeed equal to $\S$. But \Cref{prop: uniqueness of MB 2} already tells us that $\S$ is the smallest set such that $Y\indep X_{\S}\mid X_{\S^\mathsf{c}}$. If $\M^\prime \supsetneq \S $, then $\M^\prime$ violates item \textit{2} in \Cref{def: Markov boundary}. Hence, $\M^\prime$ has to be equal to $\S$. So, any Markov boundary must either have size $= p-1$ or be equal to $\S$ which proves the Theorem.

\subsection{Proof of Corollary~\ref{cor:continuous distributions}}\label{appendix:proof-continuous-covariates}
Corollary~\ref{cor:continuous distributions} presents an application of \Cref{thm: uniqueness of MB} for the case of continuous distributions with further assumptions on the support of the conditional distribution. First the case $\S = [p]$ is identical to that in \Cref{thm: uniqueness of MB} and is omitted in this proof. 

For the case $|\S| < p$, we will apply \Cref{thm: uniqueness of MB} to prove the corollary. First, suppose that for any two subsets $A$ and $B$ of $[p]$ with $|B| = 2$ and $|A \cap B| = 1$, the support of the conditional distribution $X_{A\cup B} \mid X_{(A\cup B)^\mathsf{c}}=c$ has only one equivalence class through $A\setminus B$ and $B\setminus A$; we will prove this supposition shortly. Then $\Delta = \{\{A,B\}:\ |A\cap B| =1, \ |B|=2\}$. Due to assumption \textit{(i)}, for any two (two element) sets $C_1$ and $C_2$ in $\I$ there exists a sequence of sets $D_1, \dots, D_k$ for some $k\in \mathbb{Z}_{\geq0}$ such that $|C_1\cap D_1|=1$, $|D_1\cap D_2|=1$, $|D_2\cap D_3|=1$, $\dots$, $|D_{k-1}\cap D_k|=1$ and $|D_k\cap C_2|=1$. Consider all pairs of elements of $\I$. For each of them we get a sequence of sets $(D_i:\ i\in [k])$ connecting them for some $k\in \mathbb{Z}_{\geq0}$. Consider $k^\star$ to be the maximum such $k$. Note that $k^\star$ is finite as it is the maximum over $\binom{|\I|}{2}$ integers. Therefore, $\S^\mathsf{c}\in (\Delta\circ)^{k^\star}\I$ which satisfies the assumptions of \Cref{thm: uniqueness of MB}.

The only thing left to prove is that for any $A,B \subset [p]$ such that $|A\cap B|=1$ and $|B|=2$, the support of the conditional distribution $X_{A\cup B}\mid X_{(A\cup B)^\mathsf{c}}=c$ has only one equivalence class. First note that if $c < 0$ or $c> 1$, the support is empty and when $c = 1$, the support is equal to the singleton $\{\boldsymbol{0}_{|A\cup B|}\}$ and hence, has a single equivalence class. So, we consider the only nontrivial case which is $0\geq c < 1$. 
To do that, define the sets 
\begin{align*}
    \mathrm{Supp}_{>0}(Z) &= \{z\in \mathbb{R}^p:\ f_{Z}(z)>0\}\\
     \text{ and }\mathrm{Supp}_{>0}(X) &= \{x\in \mathbb{R}^p: \sum_{j\in[p]}x_{\{j\}}=1,\ f_{X}(x)>0\}.
\end{align*}
Here $f_Z$ is the density of $Z$ and $f_X$ is the density of $X$ with respect to the $p-1$--dimensional Hausdorff measure induced by $f_Z$. Note that $\mathrm{Supp}_{>0}(X)$ represents a set related to the support of $f_X$, however with the difference that $\mathrm{Supp}_{>0}(X)$ only includes points with positive density, as opposed to its closure which is commonly defined as the support of $f_X$. For some $C\subseteq [p]$ (think of $C$ as $(A\cup B)^\mathsf{c}$), and some $c\in \mathbb{R}^{|C|}$ such that there exists $x\in \mathrm{Supp}_{>0}(X)$ with $x_C=c$, consider the following set
$$\mathrm{Supp}_{>0}(X\mid X_C=c) = \left\{x \in \mathrm{Supp}_{>0}(X): \ x_C=c\right\}.$$
We want to show that $\mathrm{Supp}_{>0}(X\mid X_C=c)$ is open and path connected in $$\left\{x\in \mathbb{R}^{p}: \sum_{j\in [p]}x_{\{j\}}=1,\ x_C = c\right\}.$$

To do that, we will be using the following Lemma (see \Cref{appendix: proof of lemma support z to x} for proof):

\begin{Lemma}\label{lemma: support z to x}
Let $Z$ be a $p$-dimensional random vector. Define the function $g: [0,\infty)^p \setminus \{\boldsymbol{0}_p \}\rightarrow [0,\infty)^p $ as follows
$$x = g(z) := \left(\frac{z_{\{1\}}}{\sum_{k=1}^p z_{\{k\}}},\frac{z_{\{2\}}}{\sum_{k=1}^p z_{\{k\}}}, \dots, \frac{z_{\{p\}}}{\sum_{k=1}^p z_{\{k\}}}\right).$$
Let $f_Z$ denote the density of $Z$ with respect to the Lebesgue measure and $f_X$ that of $X$ ($:=g\circ Z$) with respect to the $(p-1)$--dimensional Hausdorff measure induced by $f_Z$. Then, if $f_Z$ is continuous in $[0,\infty)^p\setminus \{\boldsymbol{0}_p\}$, $\mathrm{Supp}_{>0}(X) = g\left(\mathrm{Supp}_{>0}(Z)\right)$.
\end{Lemma}

\paragraph{Step 1: Path connectedness of conditional support}
We begin by demonstrating that the support is path-connected. Consider any two points $x$ and $x'$ in $\mathrm{Supp}_{>0}(X \mid X_C = c)$. Since $x \in \mathrm{Supp}_{>0}(X \mid X_C = c)$, it follows that $f_X(x) > 0$. By \Cref{lemma: support z to x}, this implies there exists a $z \in \mathrm{Supp}_{>0}(Z)$ such that $x = g(z)$, where $g$ is defined as in the Lemma. Similarly, for $x'$, there exists a point $z' \in [0,\infty)^p \setminus \boldsymbol{0}_p$ such that $f_Z(z') > 0$ and $x' = g(z')$. Assumption \textit{(ii)} implies that there exists a path $z(t)$ in $\{ z \in \mathrm{Supp}_{>0}(Z) : g(z)_C = c \}$ connecting $z$ and $z'$. Formally, $z(\cdot)$ is a continuous map from $[0,1]$ to $[0,\infty)^p \setminus \boldsymbol{0}_p$ such that $z(0) = z$, $z(1) = z'$, $z(t) \in \mathrm{Supp}_{>0}(Z)$, and $g(z(t))_C = c$ for all $t \in [0,1]$.

Define $x(t) = g(z(t))$ for $t \in [0,1]$. For any $t \in [0,1]$, since $z(t) \in \mathrm{Supp}{>0}(Z)$, by \Cref{lemma: support z to x}, we have $x(t) \in \mathrm{Supp}_{>0}(X)$. Additionally, given that $g(z(t))_C = c$, it follows that $x(t) \in \mathrm{Supp}_{>0}(X \mid X_C = c)$. To demonstrate that $\mathrm{Supp}_{>0}(X \mid X_C = c)$ is path-connected, it is sufficient to show that $x(\cdot)$ is continuous on $[0,1]$. This holds because $x(\cdot) = g \circ z(\cdot)$ is the composition of two continuous functions. Therefore, $x$ and $x'$ are connected via the path $x(t)$, which proves that $\mathrm{Supp}_{>0}(X \mid X_C = c)$ is path-connected in $\{ x \in \mathbb{R}^p : \sum_{j=1}^p x_j = 1,\, x_C = c \}$.

\paragraph{Step 2: Openness of conditional support}
Now, we will show that $\mathrm{Supp}_{>0}(X \mid X_C=c)$ is open in $\left\{x\in \Delta^{p-1}:\ x_C = c\right\}$. To show that, consider any point $x\in \mathrm{Supp}_{>0}(X\mid X_C=c)$. By \Cref{lemma: support z to x}, there exists $z\in \mathrm{Supp}_{>0}(Z)$ such that $x=g(z)$. By assumption \textit{(ii)}, since $f_Z$ is continuous in $[0,\infty)^p\setminus\{\boldsymbol{0}_p\}$, there exists an open ball $B_{\epsilon}(z)$ (open in $[0,\infty)^p\setminus\{\boldsymbol{0}_p\}$) such that $B_{\epsilon}(z) \subseteq \mathrm{Supp}_{>0}(Z)$. Define $$B_{\frac{\epsilon}{t}}(x) = \left\{x^\prime \in \Delta^{p-1}:\ x^\prime_C = c,\ ||x-x^\prime|| < \frac{\epsilon}{t}\right\}, \quad t = \sum_{j=1}^p z_j.$$ 
It is enough to show that $B_{\frac{\epsilon}{t}}(x) \subseteq \mathrm{Supp}_{>0}(X \mid X_C=c)$. Consider some $x^\prime \in B_{\frac{\epsilon}{t}}(x)$. Define $z^\prime = t \cdot x^\prime$. Since $||z-z^\prime|| = ||tx-tx^\prime|| = t ||x-x^\prime|| < t\cdot \frac{\epsilon}{t} = \epsilon$, $z^\prime \in B_{\epsilon}(z)$ and hence $z^\prime \in \mathrm{Supp}_{>0}(Z)$. Therefore, by \Cref{lemma: support z to x}, $x^\prime\in \mathrm{Supp}_{>0}(X)$. Finally, since $x^\prime_C = c$, $x^\prime\in \mathrm{Supp}_{>0}(X\mid X_C = c)$. This shows that $\mathrm{Supp}_{>0}(X\mid X_C = c)$ is open in $\left\{x\in \Delta^{p-1}:\ x_C = c\right\}$.

\paragraph{Step 3: Proving uniqueness of equivalence class of conditional support}

Now that we have shown $\mathrm{Supp}_{>0}(X\mid X_C=c)$ is open and path connected in 
\sloppy$\left\{x\in \mathbb{R}^{p}: \sum_{j\in [p]}x_{\{j\}}=1,\ x_C = c\right\}$ for all $C\subseteq [p]$, and all $c\in \mathbb{R}^{|C|}$ such that there exists $x\in \mathrm{Supp}_{>0}(X)$ with $x_C=x$, we shall show that for any $A,B \subset [p]$ such that $|A\cap B| =1$ and $|B|=2$ and for all $c\in \mathbb{R}^{p-|A\cup B|}$ such that $\sum_{j\in [p-|A\cup B|]}c_{\{j\}} < 1$, $X_{A\cup B}\mid X_{(A\cup B)^\mathsf{c}}=c$ has only one equivalence class through $A\setminus B$ and $B\setminus A$. Note that we only need to consider 2-sized sets $B$ with one intersection with $A$ because we want to ``combine" the sets $\{i,j\}$ for $\{i,j\} \in \I$. Consider any two points $x$ and $x^\prime$ in $\mathrm{Supp}_{>0}(X)$ such that $x_{(A\cup B)^\mathsf{c}}= x^\prime_{(A\cup B)^\mathsf{c}}= c$. We shall show that $x$ and $x^\prime$ are equivalent through $A\setminus B$ and $B\setminus A$. Since $\mathrm{Supp}_{>0}(X\mid X_{(A\cup B)^\mathsf{c}}=c)$ is path connected in $\left\{x\in \mathbb{R}^{p}: \sum_{j\in [p]}x_{\{j\}}=1,\ x_C = c\right\}$, there exists a path $x(t)$ in the latter set connecting $x$ and $x^\prime$. Since $\mathrm{Supp}_{>0}(X\mid X_{(A\cup B)^\mathsf{c}}=c)$ is open in $\left\{x\in \mathbb{R}^{p}: \sum_{j\in [p]}x_{\{j\}}=1,\ x_C = c\right\}$, there exist open balls $B_{\epsilon_t}(x(t))$ (open in $\left\{x\in \mathbb{R}^{p}: \sum_{j\in [p]}x_{\{j\}}=1,\ x_C = c\right\}$) for all $t\in [0,1]$, such that $B_{\epsilon_t}(x(t))\subseteq \mathrm{Supp}_{>0}(X\mid X_{(A\cup B)^\mathsf{c}}=c)$ for all $t\in [0,1]$. Here the balls are defined with respect to the standard euclidean norm (or the $L^2$ norm). Now consider the sets $B_{\frac{\epsilon_t}{\sqrt{2}}}(x(t))$ for $t\in [0,1]$. $\bigcup_{t\in [0,1]}B_{\frac{\epsilon_t}{\sqrt{2}}}(x(t))$ is an open cover of $\{x(t): t\in [0,1]\}$. Also, since $[0,1]$ is compact and $x(t)$ is continuous, $\{x(t): \ t\in [0,1]\}$ is compact in $\mathrm{Supp}_{>0}(X\mid X_{(A\cup B)^\mathsf{c}}=c)$. This implies there exists a finite subcover $B_1, \ldots, B_n$ of $\{x(t):\ t\in[0,1]\}$. Call the respective centers $x_{t_1}, \ldots, x_{t_n}$. Remember that for all $t \in \{t_1, \dots, t_n\}$, the ball centered at $x(t)$ has $L^2$ radius $\frac{\epsilon_t}{\sqrt{2}}$. WLOG assume $B_1$ is centered at $x$ and $B_n$ is centered at $x^\prime$. Call the rest of the centers $y_1$, \ldots, $y_{n-2}$. If we can show that the centers of adjacent open balls are equivalent, then we will have shown $x$ and $x^\prime$ are equivalent. To show that, consider any two balls $B_j$ and $B_{j+1}$ with centers $w$ and $w^\prime$, respectively. Since $B_j$ and $B_{j+1}$ are adjacent open balls, they must have a non-empty intersection. Consider a particular point $w^\star \in B_j\cap B_{j+1}$. If we can show $w$ and $w^\star$ are equivalent, we can use the same argument to show $w^\star$ and $w^\prime$ are equivalent, thereby showing the equivalence of $w$ and $w^\prime$, which would complete the proof. Since $$\left(\sum_{j\in A\setminus B}w_{\{j\}} + w_{B\setminus A}^\star\right) + \left(\sum_{j\in A\setminus B}w_{\{j\}}^\star + w_{B\setminus A}\right) \leq \sum_{j\in [p]} w_{\{j\}} + \sum_{j\in [p]}w_{\{j\}}^\star = 2,$$
either $\sum_{j\in A\setminus B}w_{\{j\}} + w_{B\setminus A}^\star \leq 1$ or $\sum_{j\in A\setminus B}w_{\{j\}}^\star + w_{B\setminus A}\leq 1$. 

First, consider the case when $\sum_{j\in A\setminus B}w_{\{j\}} + w_{B\setminus A}^\star \leq 1$. Introduce the notation $w = (w_{A\setminus B}, w_{A\cap B}, w_{B\setminus A},c)$ and $w^\prime = (w_{A\setminus B}^\prime, w_{A\cap B}^\prime, w_{B\setminus A}^\prime,c)$ and define 
$$w^\dag = \left(w_{A\setminus B}, 1-\sum_{j\in A\setminus B}w_{\{j\}}-w^\star_{B\setminus A}-\sum_{j\in [p-|A\cup B|]}c_j, w_{B\setminus A}^\star, c\right).$$
Then the $L^2$ distance between $w$ and $w^\dag$ is as follows
\begin{align*}
||w-w^\dag||_2 &= \sqrt{\sum_{j\in [p]} \left(w_{\{j\}}-w_{\{j\}}^\dag\right)^2}\\
&= \sqrt{\left(w_{A\cap B} - \left(1-\sum_{j\in A\setminus B} w_{\{j\}} - w_{B\setminus A} ^\star - \sum_{j\in [p-|A\cap B|]}c_j\right)\right)^2 + \left(w_{B\setminus A} - w_{B\setminus A}^\star\right)^2 }\\
&= \sqrt{2 \cdot \left(w_{B\setminus A} - w^\star_{B\setminus A}\right)^2}\\
&\leq \sqrt{2 \cdot \sum_{j\in [p]} \left(w_{\{j\}}-w_{\{j\}}^\star\right)^2}\\
&\leq \sqrt{2\cdot \frac{\epsilon_{t_j}^2}{2}} \quad \left(\text{Since }w^\star \in B_{\frac{\epsilon}{\sqrt{2}}}(w)\right)\\
&= \epsilon_{t_j}
\end{align*}
Since $||w-w^\dag||_2 \leq \epsilon_{t_j}$, $w^\dag \in B_{\epsilon_{t_j}}(x(t_j))$ which implies $f_X(w^\dag)>0$ or in other words, $w^\dag \in \mathrm{Supp}_{>0}(X\mid X_C=c)$. Since $w$ and $w^\dag$ are coordinatewise connected through $A\setminus B$ and $w^\dag$ and $w^\star$ are connected through $B\setminus A$, $w$ and $w^\star$ are equivalent through $A\setminus B$ and $B\setminus A$.

Now, consider the other case: $\sum_{j\in A\setminus B}w_{\{j\}}^\star + w_{B\setminus A}\leq 1$. In this case, define $w^\dag$ as follows
$$w^\dag = \left(w^\star_{A\setminus B}, 1-\sum_{j\in A\setminus B}w^\star_{\{j\}}-w_{B\setminus A}-\sum_{j\in [p-|A\cup B|]}c_j, w_{B\setminus A}, c\right).$$
With this $w^\dag$ the $L^2$ distance can be bounded as
\begin{align*}
||w-w^\dag||_2 &= \sqrt{\sum_{j\in [p]} \left(w_{\{j\}}-w_{\{j\}}^\dag\right)^2}\\
&= \sqrt{\sum_{j\in A\setminus B}\left(w_{\{j\}} - w_{\{j\}}^\star\right)^2 + \left(w_{A\cap B} - \left(1-\sum_{j\in A\setminus B} w_{\{j\}}^\star - w_{B\setminus A} - \sum_{j\in [p-|A\cap B|]}c_j\right)\right)^2 }\\
&= \sqrt{\sum_{j\in A\setminus B}\left(w_{\{j\}} - w_{\{j\}}^\star\right)^2 + \left(w_{A\cap B}-w_{A\cap B}^\star + w_{B\setminus A} - w_{B\setminus A}^\star\right)^2}\\
&\leq \sqrt{\sum_{j\in A\setminus B}\left(w_{\{j\}} - w_{\{j\}}^\star\right)^2 + 2\left( w_{A\cap B}-w_{A\cap B}^\star \right)^2 + 2\left( w_{B\setminus A} - w_{B\setminus A}^\star\right)^2}\\
&\leq \sqrt{2 \cdot \sum_{j\in [p]} \left(w_{\{j\}}-w_{\{j\}}^\star\right)^2}\\
&\leq \sqrt{2\cdot \frac{\epsilon_{t_j}^2}{2}} \quad \left(\text{Since }w^\star \in B_{\frac{\epsilon}{\sqrt{2}}}(w)\right)\\
&= \epsilon_{t_j}
\end{align*}
Using the same argument as in the previous case, $||w-w^\dag||_2 \leq \epsilon_{t_j}$ implies $w^\dag \in B_{\epsilon_{t_j}}(x(t_j))$ which implies $w^\dag \in \mathrm{Supp}_{>0}(X\mid X_C=c)$. Further, since $w$ and $w^\dag$ are coordinatewise connected through $B\setminus A$ and $w^\dag$ and $w^\star$ are connected through $A\setminus B$, $w$ and $w^\star$ are equivalent through $A\setminus B$ and $B\setminus A$. So, in both the cases, 
$$w \text{ and } w^\star \text{are equivalent through }A\setminus B \text{ and } B\setminus A.$$
Using a similar construction, we can show that $w^\star$ and $w^\prime$ are equivalent through $A\setminus B$ and $B\setminus A$. Applying this argument recursively, we can show that all the centers of the balls $B_1, \dots, B_n$ are equivalent, thereby proving $x$ and $x^\prime$ are equivalent (through $A\setminus B$ and $B\setminus A$). Since $x$ and $x^\prime$ are chosen to be any two points in $\mathrm{Supp}_{>0}(X\mid X_C=c)$, this proves that $\mathrm{Supp}_{>0}(X\mid X_C=c)$ has only one equivalence class. This proves the conditions of \Cref{thm: uniqueness of MB} which completes the proof of the Corollary.

\subsection{Proof of Corollary \ref{cor:factor covariates}}\label{appendix:proof of factor covariates}
First, consider the case when $F_k \subseteq \S_F$. Suppose, for the sake of contradiction, there exists a nontrivial Markov boundary $\M$ i.e., $|\M^\mathsf{c} \cap F_k| \geq 2$. Then there exists $i\neq j$ such that $\{i,j\} \subseteq \M^\mathsf{c} \cap F_k$. By the weak union property~\citep{PEARLbook}, $Y\indep X_{\{i,j\}}\mid X_{\{i,j\}^\mathsf{c}}$. Therefore $i\notin \S_F$ and $j\notin \S_F$ which implies $|F_k \setminus \S_F|\geq 2$ which contradicts $F_k\subseteq \S_F$.

Now, in the other case, consider some nontrivial Markov boundary $\M$. To show $\M = \S_F \cap F_k$ we will follow the following steps:
\begin{enumerate}
    \item $\Delta$ includes sets of sets of the following form:
    \begin{enumerate}
        \item $\{A,B\}$ where both $A,B\subseteq F_k$ for $|A\cap B|\geq1$.
        \item $\{A,B\}$ where $A\subseteq F_{k}$ and $B\subseteq F_{k}^\mathsf{c}$.
    \end{enumerate}
    \item For any Markov boundary $\M$ with $|\M^\mathsf{c}\cap F_k|\geq 2$, $\M_k \supseteq \S_F\cap F_k$.
    \item If $|\M^\mathsf{c}\cap F_k|\geq 2$, $\M_k = \S_F\cap F_k$.
\end{enumerate}
\paragraph{Step 1: $\Delta$ contains sets of form $\boldsymbol{\{A,B\}}$, $\boldsymbol{A,B \subseteq F_k}$}

Consider two subsets $A$ and $B$ of $F_k$ such that $|A\cap B| \geq 1$. Given $X_{(A\cup B)^\mathsf{c}} = c$, let us consider the support of $X_{A\cup B}$. If $\sum_{j\in (A\cup B)^\mathsf{c} \cap F_k} c_j = 1$, then the support of $X_{A\cup B}\mid X_{(A\cup B)^\mathsf{c}}=c$ only includes the $\boldsymbol{0}_{|A\cup B|}$ vector. So, it trivially contains a single equivalence class. Conversely, when $\sum_{j\in (A\cup B)^\mathsf{c} \cap F_k} c_j = 0$, the support of $X_{A\cup B}\mid X_{(A\cup B)^\mathsf{c}}$ contains all the canonical vectors of length $|A\cup B|$. This is true because the rows of $X$ are randomly selected such that all treatment combinations have positive probability. Consider any two vectors $x$ and $y$ in the support of $X_{A\cup B}\mid X_{(A\cup B)^\mathsf{c}}$. Note that, since both vectors $x$ and $y$ have a single non zero entry (=1), there can only be one of the two following cases, with the roles of $x$ and $y$ being reversible:
\begin{enumerate}[label=(\alph*)]
    \item $x_{A\setminus B}= y_{A\setminus B}=\boldsymbol{0}_{|A\setminus B|}$ or $x_{B\setminus A}= y_{B\setminus A}=\boldsymbol{0}_{|B\setminus A|}.$
    \item $\sum_{j \in A\setminus B}x_{\{j\}}= 1, y_{A\setminus B}=\boldsymbol{0}_{|A\setminus B|}$ and $x_{B\setminus A}=\boldsymbol{0}_{|B\setminus A|}, \sum_{j \in B\setminus A}y_{\{j\}}=1.$
\end{enumerate}

In case (a), the two vectors $x$ and $y$ are connected through $A\setminus B$ and $B\setminus A$ respectively. This shows that any two vectors in the support of $X_{A\cup B}\mid X_{(A\cup B)^\mathsf{c}}$ are coordinatewise connected through $A\setminus B$ and $B\setminus A$, thereby proving that support of $X_{A\cup B}\mid X_{(A\cup B)^\mathsf{c}}$ has only one equivalence class (through $A\setminus B$ and $B\setminus A$). 

In case (b), consider a canonical vector $z$ with $\sum_{j\in A\cap B}z_{\{j\}}=1$. If $|A\cap B|>1$, then $z$ could be any canonical vector with $\sum_{j\in A\cap B}z_{\{j\}}=1$ and if $|A\cap B|=1$, then $z$ is the canonical vector with $z_{A\cap B}=1$ and zeros in all other indices. Then, $x_{B\setminus A}=z_{B\setminus A}$ and $y_{A\setminus B} = z_{A\setminus B}$. Note that $z$ belongs to the support of $X_{A\cup B}\mid X_{(A\cup B)^\mathsf{c}}$. This shows that $x$ and $y$ are equivalent, which shows that support of $X_{A\cup B}\mid X_{(A\cup B)^\mathsf{c}}$ has only one equivalence class (through $A\setminus B$ and $B\setminus A$). 

Therefore, in both cases (a) and (b), the support of $X_{A\cup B}\mid X_{(A\cup B)^\mathsf{c}}$ has only one equivalence class proving $\Delta$ contains all sets of the form $\{A,B\}$ where $|A\cap B| \geq 1$ and $A,B\subseteq F_k$.
    
\paragraph{Step 2: $\Delta$ contains sets of the form $\boldsymbol{\{A,B\},\ A\subseteq F_k,\ B\subseteq F_k^\mathsf{c}}$} 
Consider $A\subseteq F_{k}$ and $B\subseteq F_{k}^\mathsf{c}$. Without loss of generality, we can assume that both $A$ and $B$ are non-empty. If both $A$ and $B$ are singletons, then, the support of $X_{A\cup B}\mid X_{(A\cup B)^\mathsf{c}}$ only has one element which is determined by $X_{(A\cup B)^\mathsf{c}}$. Consider the case when one of $A$ and $B$ is a singleton. Suppose, $A$ is a singleton. Then for all points in the support of $X_{A\cup B}\mid X_{(A\cup B)^\mathsf{c}}$ will be either have $x_{A}=0$ or $x_{A}=1$. In both the cases, any two points in the support is coordinatewise connected through $A\setminus B$ $(=A)$. The case when $B$ is a singleton is same as then $B$ is a proper subset of some other factor ($k^\prime$) and $A$ is a subset of $F_{k^\prime}^\mathsf{c}$. 

So, the only case left is when both $|A|\geq 2$ and $|B|\geq 2$. Consider any point $x$ in the support of $X_{A\cup B}\mid X_{(A\cup B)^\mathsf{c}} = c$ for some vector $c$ in the support of $X_{(A\cup B)^\mathsf{c}}$. If $\sum_{j\in A^\mathsf{c}\cap F_{k}}x_{\{j\}}=1$, then for all points in the support, $x_{A}=\boldsymbol{0}_{|A|}$ making the support coordinatewise connected through $A\setminus B$ $(=A)$. If $\sum_{j\in A^\mathsf{c}\cap F_{k}}x_{\{j\}}=0$, then for any point $x$ in the support $x_A = e_i$ for some $i$. Here $e_i$ represent canonical vectors of length $|A|$. Consider any two points $x$ and $y$ in the support of $X_{A\cup B}\mid X_{(A\cup B)^\mathsf{c}} = c$. Define $z = (x_A, y_B)$. Since $z$ satisfies the criteria $\sum_{j\in F_{k}}z_{\{j\}} = 1$, $z$ lies in the support. $x$ and $z$ are coordinatewise connected through $A\setminus B$ $(=A)$ and $y$ and $z$ are coordinatewise connected through $B\setminus A$ $(=B)$. This implies $x$ and $y$ are equivalent which proves the support of $X_{A\cup B}\mid X_{(A\cup B)^\mathsf{c}}$ has only one equivalence class. Therefore $\Delta$ contains all sets of the form $\{A,B\}$ such that $A\subseteq F_{k}$ and $B\subseteq F_{k^\mathsf{c}}$. 

\paragraph{Step 3: For $\boldsymbol{\M}$ with $\boldsymbol{|\M^\mathsf{c}\cap F_k|\geq 2}$, $\boldsymbol{\M_k\supseteq \S_F\cap F_k}$}
Consider any Markov boundary such that $|\M^\mathsf{c}\cap F_k|\geq 2$ and some $i\in \M^\mathsf{c}\cap F_k$. Since $|\M^\mathsf{c}\cap F_k|\geq 2$, there exists some $j\in (\M^\mathsf{c}\cap F_k)\setminus \{i\}$. By the weak union property, since $\{i,j\}\subseteq \M^\mathsf{c}$, $Y\indep X_{\M^\mathsf{c}}\mid X_{\M}$ implies $Y\indep X_{\{i,j\}}\mid X_{\{i,j\}^\mathsf{c}}$. This means $i\in \S_F^\mathsf{c}\cap F_k$. Therefore, $\M^\mathsf{c}\cap F_k\subseteq \S_F^\mathsf{c}\cap F_k$ which proves $\M_k = \M\cap F_k\supseteq \S_F\cap F_k$.

\paragraph{Step 4: For $\boldsymbol{\M}$ with $\boldsymbol{|\M^\mathsf{c}\cap F_k|\geq 2}$, $\boldsymbol{\M_k = \S_F\cap F_k}$}
To show this, it is enough to show that $\M_k\not\supset \S_F\cap F_k$. We will first begin by showing that $$Y\indep X_{\S_F^\mathsf{c}\cap F_k}\mid X_{(\S_F\cap F_k) \cup F_k^\mathsf{c} }.$$
This is equivalent to showing $\S_F^\mathsf{c}\cap F_k\in (\Delta\circ)^l\I$ for some finite $l$. Consider the sets in $\I$ which have both the elements in $F_k$. We call the collection of those sets $\I_k$. Formally,
$$\I_k = \I\cap \{\{i,j\}:\ i,j\in F_k\}.$$
Then $\S_F^\mathsf{c}\cap F_k $ can be written as the union of all the sets in $\I_k$. Consider any $A, B\in \I_k$. By step 1, if $|A\cap B|\geq 1$, $\{A,B\}\in \Delta$. Since all the elements in $\I_k$ are assumed to be connected by a sequence of true bivariate conditional hypotheses (of the form $\H_{i,j}:\ \{i,j\}\subseteq F_k$), we have $\S^\mathsf{c}\cap F_k \in (\Delta \circ)^{l_k}\I_k$ where $l_k=|\I_k|$. Therefore, $$Y\indep X_{\S_F^\mathsf{c}\cap F_k}\mid X_{(\S_F\cap F_k) \cup F_k^\mathsf{c} }.$$
Now, $\M$ is a Markov boundary implies $Y\indep X_{\M^\mathsf{c}}\mid X_{\M}$ which further implies, by the weak union property, 
$$Y\indep X_{\M^\mathsf{c}\cap F_k^\mathsf{c}}\mid X_{\M\cup F_k}.$$
Define $A = \S_F^\mathsf{c}\cap F_k$ and $B = \M^\mathsf{c}\cap F_k^\mathsf{c}$. Since $A$ and $B$ satisfy the form of sets considered in Step 2, $\{A,B\}\in \Delta$. Therefore $Y\indep X_{A\cup B} \mid X_{A^\mathsf{c} \cap B^\mathsf{c}}$. Putting the expressions of $A$ and $B$ we get,
$$Y\indep X_{(\S_F^\mathsf{c}\cap F_k) \cup (\M^\mathsf{c}\cap F_k^\mathsf{c})}\mid X_{(\S_F\cap F_k) \cup (\M\cup F_k^\mathsf{c})}.$$
So, $(\S_F\cap F_k) \cup (\M\cup F_k^\mathsf{c})$ is a Markov blanket. If $\S_F\cap F_k \subsetneq \M\cap F_k$, then $(\S_F\cap F_k) \cup (\M\cup F_k^\mathsf{c}) \subsetneq \M$ which is a contradiction to the definition of Markov boundary. Therefore $\S_F\cap F_k = \M\cap F_k$. This completes the proof of Step 4 and hence, \Cref{cor:factor covariates}.

\subsection{Proof of Theorem~\ref{thm: validity of sparse covariates conditioning method}}\label{sec: proof of theorem validity of sparse covariates}
If $\S = [p]$, $Y\nindep X_{\{i,j\}}\mid X_{\{i,j\}^\mathsf{c}}$ for all $i\neq j$. In particular, $Y \nindep X_{\{i,j\}} \mid X_{\{i,j\}^\mathsf{c}}$ for all $\{i,j\}\subseteq \D$. Therefore, $\S_\D = \D$ which proves first part of the Theorem. For the other case, first note that $\S_\D \supseteq \S \cap \D$. When $\S^\mathsf{c}\cap \D = \emptyset$, $\D \subseteq \S$ which implies $\S_\D \subseteq \S \cap \D$. Combining them we get, $\S_\D = \S \cap \D$.

Now consider the case $|\S^\mathsf{c} \cap \D| > 1$. For the sake of contradiction, suppose there exists $i \in \S_\D \setminus \S$. Since $|\S^\mathsf{c} \cap \D| \geq 2$, there exists $j\neq i$ such that $j\in \D \setminus \S$. Therefore $\{i, j\} \subseteq \S^\mathsf{c}$. Since the conditions of \Cref{thm: uniqueness of MB} are satisfied in this case, we have $\S = \M$ and hence, $Y\indep X_{\S^\mathsf{c}}\mid X_{\S}$. By the weak union property~\citep{PEARLbook} $Y\indep X_{\{i,j\}}\mid X_{\{i,j\}^\mathsf{c}}$. This implies $i\notin \S_\D$ which is a contradiction. Therefore, $\S_\D\subseteq \S\cap \D$. Since we already showed that $\S_\D \supseteq \S \cap \D$, $\S_\D = \S \cap \D$. 

Finally, as we already mentioned, the conditions of \Cref{thm: uniqueness of MB} are satisfied here. So, $\S = \M$. Combining everything, $\S_\D = \S \cap \D = \M \cap \D$ which proves the Theorem.

\subsection{Proof of Theorem~\ref{thm: FWER control - Bonferroni}}\label{appendix: proof of FWER-Bonferroni}
Since we are considering Bonferroni's p-values after $h-1$ rejections (call the corresponding rejection set $\hat{\S}$), we have the following p-values to consider:
$$P_{j}^\mathrm{B}(\hat{\S})=(p-\overline{s})P_{(\overline{s}-h+1),j}(\hat{\S}) \text{ for all }j\in \hat{\S}^\mathsf{c}.$$
Now, these p-values are valid for the null variables $j$ (i.e., $j \in \S^\mathsf{c}$) if there are at least $p-\overline{s}-1$ null p-values in the vector $P_{\hat{\S}^\mathsf{c}\setminus \{j\} , j}$. If there are no null variables in the rejection set after the $h-1^\text{th}$ step---$\hat{\S}$, then there are indeed $p-s-1 (\geq p-\overline{s}-1)$ null p-values in this vector which is ensured by our assumption. 

So, $P_{j}^\mathrm{B}(\hat{\S})$ are valid for all null $j$ if and only if $\hat{\S}$ does not contain any null variables. Coming back to the proof of validity of the Bonferroni--Holm method (described in \Cref{Algorithm: adaptive-Bonferroni--Holm}), this method introduces an order among the hypotheses $\H_j: j\in [p]$ which is the order in which each column is considered for rejection (if we were to reject all the variables). Let the hypotheses arranged according to this order be represented as $\H_{(1)}, \H_{(2)}, \dots, \H_{(p)}$. Suppose $\H_{(h)}$ is the first falsely rejected null hypothesis and let us denote the corresponding variable as $j^\star$. Since $j^\star$ is a null variable and all the $h-1$ variables corresponding to the hypotheses $\H_{(1)}, \H_{(2)}, \dots, \H_{(h-1)}$ which were rejected prior to it were non-null, Bonferroni's p-value $P_{j^\star}^\mathrm{B}(\hat{\S})$ 
is superuniform. 
So, we have, 
$$P_{j^\star}^\mathrm{B}(\hat{\S}) \succeq \text{Uniform}[0,1].$$
Let us now consider the cutoff values. Since only non-null variables are rejected before step $h$, $h<s+1$. This implies $p-h+1 \geq p-s$. Therefore, $j^\star $ is rejected implies $P_{j^\star}^\mathrm{B}(\hat{\S})\leq \dfrac{\alpha}{p-h+1}\leq \dfrac{\alpha}{p-s}$. Using these results, we can upper bound the FWER as follows:
\begin{align*}
    \mathrm{FWER} = \mathbb P(V > 1) &= \P\left(\bigcup_{j\in \S^\mathsf{c}}\{j\text{ is rejected}\}\right) =  \P\left(P_{j^\star}^\mathrm{B}(\hat{\S}) \leq \dfrac{\alpha}{p_0}\right)\\
    &=  \P\left(\bigcup_{j\in \S^\mathsf{c}} \left\{P_{j}^\mathrm{B}(\hat{\S}) \leq \dfrac{\alpha}{p_0}\right\}\right)\leq \sum_{j\in \S^\mathsf{c}} \mathbb P\left(P_{j}^\mathrm{B}(\hat{\S}) \leq \dfrac{\alpha}{p_0}\right)\\
    &\leq \sum_{j\in \S^\mathsf{c}} \frac{\alpha}{p-s} = \alpha.
\end{align*}

\subsection{Proof of Theorem~\ref{thm: FWER control}}\label{appendix: proof of FWER-Simes}

The proof of this theorem follows the exact same structure as the proof of \Cref{thm: FWER control - Bonferroni}. The only difference is that, we are using Simes' p-values here instead of Bonferroni's p-values. For any column $j$, after $h-1$ rejections and the corresponding rejection set $\hat{\S}$, we have the following p-values to consider:
$$P_j^\mathrm{S}(\hat{\S}) := \min _{\overline{s}-h+1 \leq i \leq p-h}\left\{\frac{p-\overline{s}}{i-\overline{s}+h} P_{(i),j}(\hat{\S})\right\}
$$
For a null variable $j$, this p-value is valid when there are at least $p-\overline{s}-1$ null p-values in the vector $P_{[p]\setminus \{\hat{\S} \cup \{j\}\} , j}$ and $P_{[p]\setminus \{\hat{\S} \cup \{j\}\} , j}$ satisfy the PRDS assumption (see \Cref{appendix: proof of FDR control} for definition) which is true by the assumptions of this Theorem. Therefore, we only need to show that there are at least $p-\overline{s}-1$ p-values in $P_{[p]\setminus \{\hat{\S} \cup \{j\}\} , j}$. This is exactly what we have showed in \Cref{appendix: proof of FWER-Bonferroni}. So, using a similar argument, we can show that FWER $\leq \alpha$.

\subsection{Proof of Theorem~\ref{thm: FDR control}}\label{appendix: proof of FDR control}
Simes p-values are defined as follows:
$$
P_j^\mathrm{S} = \min_{s\leq i \leq p-1} \frac{p-\overline{s}}{i - \overline{s} + 1} P_{(i),j}.
$$
Simes p-values are superuniform under the partial conjunction null assuming the base p-values are PRDS~\citep[Theorem 1]{benjamini2008screening}. Formally, the PRDS assumption, defined in \cite{benjaminiyekutieli}, is as follows: 

\begin{Definition}[Positive Regression Dependency on the Subset (PRDS)]~\citep{benjaminiyekutieli}
  We say that p-values $ P^\prime := (P^\prime_1, P^\prime_2, \dots, P^\prime_m)$ satisfy the $\mathrm{PRDS}$ property on a subset $\mathbf{I}_0 \subseteq [m]$ if, for any increasing set $\mathbf{D}\in [0,1]^m$ and each $j \in \mathbf{I}_0$, the probability $\mathbb{P}(P^\prime \in \mathbf{D} \mid  P^\prime_j = x)$ is non-decreasing in $x$. Here, $\mathbf{D}$ is called a non-decreasing set if $x \in \mathbf{D}$ and $y \geq x$ implies that $y \in \mathbf{D}$ as well. We simply say that the p-values $P^\prime$ are $\mathrm{PRDS}$ if the p-values satisfy are $\mathrm{PRDS}$ on the the entire set $[m]$.
\end{Definition}

To prove the validity of our method, it is enough to show that the PRDS assumption on the base p-values implies that the Simes p-values are also PRDS~\citep[Theorem 1.2]{benjaminiyekutieli}. We will do so using \citet[Theorem 3.1]{TPH_under_dependence}. Specifically, we will be examining conditions outlined in item 1 of the Theorem---we restate a special case of their Theorem below in our notation, but first we need a few new definitions. 

\begin{Definition}[\cite{blanchard2008two}]
A \textbf{threshold collection} $\Delta$ is function
    $$\Delta:(j,r)\in [m]\times \mathbb{R}^{+} \rightarrow \Delta(j,r)\in \mathbb{R}^{+},$$
    which is non-decreasing in its second input. A \textbf{factored threshold collection} is a threshold collection of the form
    $$\Delta(j,r)=\frac{\alpha w_j\beta(r)}{m}\ \forall (j,r)\in [m]\times \mathbb{R}^{+},$$
    where $w_j$ is a prior weight for the $i^{\text{th}}$ hypothesis and $\beta:\mathbb{R}^{+}\rightarrow \mathbb{R}^{+}$ is a non-decreasing function called a shape function. Given a threshold collection $\Delta$, the $\boldsymbol{\Delta}$-\textbf{thresholding-based multiple testing procedure} at rejection volume $r$ is a multiple testing procedure that has the following rejection set when applied to the vector of p-values $P^{\prime} = \{P^{\prime}_1, \dots, P^{\prime}_m\}$:
    $$L_{\Delta}(r) = \{j \in [m]:\ P^\prime_j \leq \Delta(j, r)\}.$$
\end{Definition}

\begin{Definition}[\cite{TPH_under_dependence}] Consider a multiple testing procedure applied on a set of p-values $\{P^{\prime}_j:\ j\in [m]\}$ corresponding to hypotheses $\{\H_{j}:\ j\in [m]\}$ with the set of indices of rejected hypotheses $\mathcal{R} = \mathcal{R}(P^{\prime}) \subseteq [m]$. \begin{enumerate}
    \item The procedure is called \textbf{self consistent} with respect to threshold collection $\Delta$ if the rejection set $\mathcal{R}$ satisfies $\mathcal{R} \subseteq L_\Delta(|\mathcal{R}|)$ almost surely.
    \item The procedure is \textbf{non-increasing} if $|\mathcal{R}(P^{\prime})|$ is non-increasing in each p-value $P^{\prime}_{j},\ j\in [m]$.
    \item The procedure is said to follow the \textbf{natural monotonicity property} if given a set of rejected hypotheses $\mathcal{R}_1$, decreasing a certain p-value will result in rejecting a superset $\mathcal{R}_2$. 
\end{enumerate}
\end{Definition}
\begin{Definition}[\cite{wright1992adjusted}]
The adjusted p-value for a particular hypothesis within
a collection of hypotheses of a certain multiple testing procedure is the smallest overall significance level at which the multiple testing procedure rejects said particular hypothesis.     
\end{Definition}
When the multiple testing procedure is Benjamini--Hochberg at level $\alpha$~(BH) applied on a set of p-values $\{P^{\prime}_j:\ j\in [m]\}$, the $j^{\text{th}}$ adjusted p-value is
$$\max_{k\in [m]:\ P^\prime_k \geq P^\prime_j}\left\{\frac{P^\prime_k \cdot m}{\text{rank}(P^\prime_{k}; P^{\prime})}\right\},$$
where $\text{rank}(P^\prime_{k}; P^{\prime})$ is the rank of $P^\prime_{k}$ in $P^{\prime}$~\citep{yekutieli1999resampling}.

\begin{Theorem}[\cite{TPH_under_dependence}, Theorem 3.1]\label{thm: 3.1 Bogomolov}
Consider $M$ base null hypotheses with p-values ${P_1^\prime, \ldots, P_M^\prime}$, grouped into $m$ groups $A_j$, for $j \in [m]$, with $|A_j| = m_j$.
Let $\mathcal{P}$ be a self-consistent multiple testing procedure with respect to thresholds of the form $\Delta(j,r) = \alpha r/m\ \forall\ j \in [m]$. Consider a family of partial conjunction hypotheses $\{\H_{0,j}^{r_j/m_j},\ j \in [m]\}$ for the $m$ groups of p-values. For each $j\in [m]$, let $\mathcal{P}_j$ be a multiple testing procedure satisfying the natural monotonicity property and for each $j \in [m]$, the partial conjunction p-value $P_j^{r_j/m_j}$ is connected to $\mathcal{P}_j$ in the following way:
$$P_j^{r_j/m_j} = \max\left\{P_j^{1/(m_j - r_j + 1)}[A]:\ A\subseteq A_j, |A|= m_j - r_j + 1\right\},$$
where $P_j^{1/(m_j - r_j + 1)}[A]$ is the minimum adjusted p-value according to $\mathcal{P}_j$ applied on the p-values with indices in $A$. Suppose the procedure $\mathcal{P}$, the procedures $\{\mathcal{P}_j, j\in [m]\}$, and the base p-values satisfy the following conditions:
\begin{enumerate}
    \item[(a)] The p-values $\{P_1^\prime, \ldots, P_M^\prime\}$ are $\mathrm{PRDS}$ on the set of true base null hypotheses.
    \item[(b)] The multiple testing procedure $\mathcal{P}$ is non-increasing.
    \item[(c)] For each $j \in [m]$, the multiple testing procedure $\mathcal{P}_j$ is self-consistent with respect to thresholds of form $\Delta(i,r) = \alpha r/m_j$.
\end{enumerate}    
Then, if $\mathbf{I}_0$ denotes the set of true partial conjunction hypotheses $\H_{0,j}^{r_j/m_j}$ and $\mathcal{R}$ denotes the set of rejections when the procedure $\mathcal{P}$ is applied on $\{P_j^{r_j/m_j}, j\in [m]\}$ at level $\alpha$, the $\mathrm{FDR}$ of the resulting procedure
$$\mathbb{E}\left[\frac{\left|\mathcal{R}\cap \mathbf{I}_0\right|}{\left|\mathcal{R}\right|}\right]\leq \alpha \frac{|\mathbf{I}_0|}{m} \leq \alpha.$$
\end{Theorem}
For \Cref{algorithm: Simes-BH-oracle}, the multiple testing procedure $\mathcal{P}$ is $\mathrm{BH}$ and $m=p$. We will be showing shortly that the multiple testing procedures $\mathcal{P}_1, \mathcal{P}_2, \dots, \mathcal{P}_m$ connected to the partial conjunction p-values in \Cref{algorithm: Simes-BH-oracle} are also BH. For all $j\in [p]$, the partial conjunction p-value $P_{j}^{r_j/m_j} = P_{j}^{\mathrm{S}}$ with $r_j =\overline{s}$, $m_j=p-1$, and $A_j = \{(i,j):\ i\in \{j\}^\mathsf{c}\}$. 

We will begin the proof of validity of \Cref{algorithm: Simes-BH-oracle} by proving that BH is self-consistent, non-increasing and has the natural monotonicity property. Specifically, we will consider BH applied on a vector of $\tilde{m}$ p-values $\tilde{P} = (\tilde{P}_1, \tilde{P}_2, \dots, \tilde{P}_{\tilde{m}})$. To show self-consistency, we will show the rejection set
$$\mathcal{R} = \mathcal{R}(\tilde{P}) \subseteq \left\{j\in [m]:\ \tilde{P}_j \leq \frac{\alpha |\mathcal{R}|}{m}\right\}.$$
We begin by noting the following identities about the rejection set of the BH procedure.
$$\mathcal{R} = \{j\in [\tilde{m}]: P_{j} \leq \tilde{P}_{(|\mathcal{R}|)}\}$$ and $$|\mathcal{R}| = \max \left\{j\in [\tilde{m}] : \ \tilde{P}_{(j)}\leq \frac{\alpha j}{\tilde{m}}\right\}.$$
Here for any $j\in [\tilde{m}]$, $\tilde{P}_{(j)}$ denotes the $j^{\text{th}}$ ordered p-value among $\{\tilde{P}_{1}, \tilde{P}_{2}, \dots, \tilde{P}_{\tilde{m}}\}$. In particular, $\tilde{P}_{(|\mathcal{R}|)} \leq \frac{\alpha |\mathcal{R}|}{\tilde{m}}$. This, combined with the first identity implies that for any $j\in \mathcal{R}$, $\tilde{P}_{j} \leq \frac{\alpha |\mathcal{R}|}{\tilde{m}}$ which proves self-consistency of BH. This proves that the methods $\mathcal{P}, \mathcal{P}_1, \mathcal{P}_2, \dots, \mathcal{P}_m$ are all self-consistent.

To show that BH is non-increasing and follows the natural monotonicity property, we consider some $j^\star\in [\tilde{m}]$ and define $\tilde{Q} = (\tilde{P}_1, \dots, \tilde{P}_{j^\star-1}, \tilde{Q}_{j^\star}, \tilde{P}_{j^\star+1}, \dots, \tilde{P}_{\tilde{m}})$ where $\tilde{Q}_{j^\star} \geq \tilde{P}_{j^\star}$. To prove BH satisfies the natural monotonicity property, we will show that $\mathcal{R}(\tilde{P}) \supseteq \mathcal{R}(\tilde{Q})$ and to show that BH is non-increasing, we will show that $|\mathcal{R}(\tilde{P})|\geq |\mathcal{R}(\tilde{Q})|$. For the proof, first observe that the ordered vectors of $\tilde{P}$ and $\tilde{Q}$ satisfy $\tilde{P}_{(j)}\leq \tilde{Q}_{(j)}$ for all $j\in [\tilde{m}]$. Therefore 
$$\left\{j\in [\tilde{m}]:\ \tilde{P}_{(j)}\leq \frac{\alpha j}{\tilde{m}} \right\} \supseteq \left\{j\in [\tilde{m}]:\ \tilde{Q}_{(j)}\leq \frac{\alpha j}{\tilde{m}} \right\}.$$
Since $|\mathcal{R}(\tilde{P})| = \max\left\{j\in [\tilde{m}]:\tilde{P}_{(j)}\leq \frac{\alpha j}{\tilde{m}} \right\}$ and  $|\mathcal{R}(\tilde{Q})| = \max\left\{j\in [\tilde{m}]:\tilde{Q}_{(j)}\leq \frac{\alpha j}{\tilde{m}} \right\}$, we get that $|\mathcal{R}(\tilde{P})| \geq |\mathcal{R}(\tilde{Q})|$ proving BH is non-increasing. Then, the natural monotonicity property holds because of the following:
\begin{align*}
\mathcal{R}(\tilde{P}) = \left\{j\in [\tilde{m}]:\ \tilde{P}_{j}\leq \frac{\alpha |\mathcal{R}(\tilde{P})|}{\tilde{m}} \right\} &\supseteq \left\{j\in [\tilde{m}]:\ \tilde{P}_{j}\leq \frac{\alpha |\mathcal{R}(\tilde{Q})|}{\tilde{m}} \right\} \\
&\supseteq \left\{j\in [\tilde{m}]:\ \tilde{Q}_{j}\leq \frac{\alpha |\mathcal{R}(\tilde{Q})|}{\tilde{m}} \right\} = \mathcal{R}(\tilde{Q}).
\end{align*}
The first inclusion ($\supseteq$) holds because $|\mathcal{R}(\tilde{P})| \geq |\mathcal{R}(\tilde{Q})|$, and the second inclusion follows from the fact that $\tilde{Q}_j \geq \tilde{P}_j$ for all $j \in [m]$. This proves that the BH procedure, and by extension the multiple testing procedures $\mathcal{P}_1, \mathcal{P}_2, \dots, \mathcal{P}_m$, follow the natural monotonicity property.

Next, we will establish the connection between the Simes' partial conjunction p-values and BH. 
For that, we again consider a vector of p-values $\tilde{P} = (\tilde{P}_1, \tilde{P}_2, \dots, \tilde{P}_{\tilde{m}})$ and some $\tilde{r} \leq \tilde{m}$. We want to show that the Simes' PCH p-value $P^{\tilde{r},\tilde{m}}$ (as defined in \Cref{eq:Simes pch p-val}, with $P_{\cdot, j}$ replaced by $\tilde P$, $\overline{s}$ by $\tilde r$, and $p-1$ by $\tilde m$) satisfies the following
$$
P^{\tilde{r},\tilde{m}}=\max\left\{P^{1,\tilde{m}-\tilde{r}+1}[A]: A\subseteq [\tilde{m}], |A|=\tilde{m}-\tilde{r}+1\right\},
$$
where $P^{1,\tilde{m}-\tilde{r}+1}[A]$ is the minimum adjusted p-value according to BH applied to the subvector $\tilde{P}_A$. The set of adjusted p-values for BH applied to $\tilde{P}_A$ is 
$$\left\{\frac{\tilde{m}-\tilde{r}+1}{\text{rank}(\tilde{P}_{j}; \tilde{P}_A)}\tilde{P}_j : j\in A\right\},$$
where $\text{rank}(\tilde{P}_{j}; \tilde{P}_A)$ is the rank of $\tilde{P}_{j}$ in $\tilde{P}_A$. Therefore, it is enough to show
\begin{align}\label{eq: proof of theorem 3.3 WTP}
    \min_{\tilde{r}\leq j \leq \tilde{m}} \frac{\tilde{m}-\tilde{r}+1}{i - \tilde{r} + 1} \tilde{P}_{(j)} = \max_{A\subseteq [\tilde{m}]: |A| = \tilde{m}-\tilde{r}+1}\ \min_{j\in A} \frac{\tilde{m}-\tilde{r}+1}{\text{rank}(\tilde{P}_{j}; \tilde{P}_A)}\tilde{P}_{j}.
\end{align}
Now, 
$$\min_{j\in A} \frac{\tilde{m}-\tilde{r}+1}{\text{rank}(\tilde{P}_{j}; \tilde{P}_A)}\tilde{P}_{j} = \min_{j\in [\tilde{m}-\tilde{r}+1]} \frac{\tilde{m}-\tilde{r}+1}{j}(\tilde{P}_{A})_{(j)},$$
where $(\tilde{P}_{A})_{(j)}$ denotes the $j^\mathrm{th}$ ordered p-value in $\tilde{P}_A$. Define $A_{\max}$ as the set of $\tilde{m}-\tilde{r}+1$ maximum p-values in $\tilde{P}$. 
Then $\left(\left(\tilde{P}_{A_{\max}}\right)_{(1)}, \left(\tilde{P}_{A_{\max}}\right)_{(2)},\dots, \left(\tilde{P}_{A_{\max}}\right)_{(\tilde{m}-\tilde{r}+1)}\right)$ is coordinatewise greater than or equal to $\left(\left(\tilde{P}_{A}\right)_{(1)}, \left(\tilde{P}_{A}\right)_{(2)},\dots, \left(\tilde{P}_{A}\right)_{(\tilde{m}-\tilde{r}+1)}\right)$ which implies 
$$\min_{j\in [\tilde{m}-\tilde{r}+1]} \frac{\tilde{m}-\tilde{r}+1}{j}(\tilde{P}_{A_{\max}})_{(j)} \geq \min_{j\in [\tilde{m}-\tilde{r}+1]} \frac{\tilde{m}-\tilde{r}+1}{j}(\tilde{P}_{A})_{(j)}\quad \mathrm{for\ any\ } A\subseteq [\tilde{m}], \ |A| = \tilde{m}-\tilde{r}+1.$$
Therefore, the RHS of \Cref{eq: proof of theorem 3.3 WTP} equals
\begin{align*}
\max_{A \subseteq [\tilde{m}]: |A| = \tilde{m}-\tilde{r}+1} \min_{j \in A} \frac{\tilde{m}-\tilde{r}+1}{\text{rank}(\tilde{P}_{j}; \tilde{P}_A)} \tilde{P}_{j} &= \min_{j \in A_{\max}} \frac{\tilde{m}-\tilde{r}+1}{\text{rank}(\tilde{P}_{j}; \tilde{P}_{A_{\max}})} \tilde{P}_{j} \\
&= \min_{j \in \{\tilde{r}, \dots, \tilde{m}\}} \frac{\tilde{m}-\tilde{r}+1}{j - \tilde{r} + 1} \tilde{P}_{(j)} \\ 
&= \mathrm{LHS\ of\ \Cref{eq: proof of theorem 3.3 WTP}}.
\end{align*}
Therefore, the Simes' PC p-values and BH are connected as stated in \Cref{thm: 3.1 Bogomolov} which proves the connection of the Simes' PC p-values $P^{\mathrm{S}}_{j}$ and BH in \Cref{algorithm: Simes-BH-oracle}.

Finally, we assume that the p-values $\{P_{i,j}: i,j\in [p], i\neq j\}$ are PRDS. We have demonstrated that \Cref{algorithm: Simes-BH-oracle} satisfies all the conditions of \Cref{thm: 3.1 Bogomolov}, which implies that it controls the FDR at level $\alpha$.

\subsection{Proof of Lemma \ref{lemma: support z to x}}\label{appendix: proof of lemma support z to x}
Define the transformation $h: \mathbb{R}^{p}\rightarrow \mathbb{R}^{p+1}$ as follows:
$$g(Z) = (X, T) \quad X_{\{j\}} = \frac{Z_{\{j\}}}{\sum_{k = 1}^p Z_{\{k\}}} \forall j \in [p],\ T = \sum_{k=1}^p Z_{\{k\}}$$

$X$ lies on the $p-1$ dimensional simplex. For any $d \in \mathbb{N}$ define $H^d$ and $\lambda_d$ as the $d$-dimensional Hausdorff measure and Lebesgue measure in $\mathbb{R}^d$ respectively. So $f_Z$ is the Radon--Nikodym derivative or simply the probability density function of $Z$ with respect to $\lambda_p$. $f_X$ is the Radon--Nikodym derivative of the push forward measure of $X$ induced by $h$ and the measure of $Z$ with respect to $H^{p-1}$ if and only if $P(X\in B) = \int_{B} f_X dH^{p-1}$ for all open sets $B$ on the $p-1$ dimensional simplex. Now
\begin{align*}
P(X\in B) &= \int_{h^{-1}(B\times [0,\infty))} f_Z d\lambda_p = \int_{B\times [0,\infty)} f_Z(xt) D(x,t) dH^{p}
\end{align*}
where $D(x,t)$ is the appropriate volume element associated with this change of variables. If $J(x,t)$ is the Jacobian of the transformation $h$ at $(x,t)$, then $D(x,t) = \text{det}\left(J(x,t)^T J(x,t)\right)^{-\frac{1}{2}}$. Now, 
\begin{align*}
    \frac{\partial x_i}{\partial z_i} = \frac{\sum_{j\neq i} z_i}{(\sum_{j=1}^p z_j)^2} = \frac{1}{t}-\frac{x_i}{t}, \quad \frac{\partial x_i}{\partial z_j} = \frac{-z_i}{(\sum_{j=1}^p z_j)^2} = -\frac{x_i}{t}, \quad \frac{\partial t}{\partial z_i} = 1
\end{align*}
implying $J(x,t) = \begin{pmatrix}
    \frac{1}{t}I_p - \frac{1}{t}x\boldsymbol{1}_{p}^T\\
    \boldsymbol{1}_{p}^T
\end{pmatrix}$ and hence $J(x,t)^T J(x,t) = \frac{1}{t^2} I_p+(1 + \frac{1}{t^2} x^T x)\boldsymbol{1}_{p} \boldsymbol{1}_{p}^T - \frac{1}{t^2} \boldsymbol{1}_{p} x^T - \frac{1}{t^2} x\boldsymbol{1}_{p}^T$.
\begin{align*}
    \therefore \text{det}(J(x,t)^T J(x,t)) &= \frac{1}{t^{2p}}\text{det}\left(I_p + (t^2 + x^T x)\boldsymbol{1}_{p} \boldsymbol{1}_{p}^T - \boldsymbol{1}_{p} x^T - x\boldsymbol{1}_{p}^T\right)\\
    &= \frac{1}{t^{2p}}\text{det}\left(I_p + \begin{pmatrix}
        \sqrt{t^2 + x^T x}\boldsymbol{1}_{p} & -\boldsymbol{1}_{p} & - x
    \end{pmatrix} \begin{pmatrix}
        \sqrt{t^2 + x^T x}\boldsymbol{1}_{p}^T \\ x^T \\ \boldsymbol{1}_{p}^T
    \end{pmatrix}\right)\\
    &= \frac{1}{t^{2p}}\text{det}\left(I_3 + \begin{pmatrix}
        p(t^2 + x^T x) & - p \sqrt{t^2 + x^T x} & -\sqrt{t^2 + x^T x}\\
        \sqrt{t^2 + x^T x} & -1 & -x^Tx\\
        p\sqrt{t^2 + x^T x} & -p & -1
    \end{pmatrix}\right)\\
    &= \frac{1}{t^{2p}}\text{det}\left(\begin{pmatrix}
        1+p(t^2 + x^T x) & - p \sqrt{t^2 + x^T x} & -\sqrt{t^2 + x^T x}\\
        \sqrt{t^2 + x^T x} & 0 & -x^Tx\\
        p\sqrt{t^2 + x^T x} & -p & 0
    \end{pmatrix}\right)\\
    &= \frac{1}{t^{2p}} p t^2 = \frac{p}{t^{2(p-1)}}.
\end{align*}
So, $D(x,t)\propto t^{p-1}$ which implies the following
\begin{align*}
P(X\in B) &\propto \int_{B\times [0,\infty)}f_Z(xt)t^{p-1}dH^p\\
&\propto \int_{B}\int_{[0,\infty)}f_Z(xt)t^{p-1}d\lambda_1 dH^{p-1}
\end{align*}
Define$$f_X(x) = c\int_{[0,\infty)}f_Z(xt)t^{p-1}d\lambda_1 = c\int_0^\infty f_Z(xt)t^{p-1}dt$$
where $c$ is a constant $\left(= \int_{\Delta^{p-1}}\int_{[0,\infty)}f_Z(xt)t^{p-1}d\lambda_1 dH^{p-1}\right)$. Then by Radon--Nikodym Theorem, $f_X$ is the density of $X$ with respect to $H^{p-1}$.

Now, to show the Lemma, we first prove the following claims
\paragraph{Claim 1: $\boldsymbol{\mathrm{Supp}_{>0}(X)\supseteq g\left(\mathrm{Supp}_{>0}(Z)\right)}$}
Consider a point $z\in \mathrm{Supp}_{>0}(Z)$ and define $x = g(z)$. We want to prove that $x \in \mathrm{Supp}_{>0}(X)$. Since $f_Z$ is continuous in $\mathbb{R}^p\setminus \boldsymbol{0}_p$ and $f_Z(z)>0$, $f_Z(z^\prime) > 0$ for all $z^\prime$ in some open neighborhood of $z$. Formally, there exists $\epsilon>0$ such that $f_Z(z^\prime) > 0\ \forall \ z^\prime \in B_{\epsilon}(z) \cap [0,\infty)^p \setminus \boldsymbol{0}_p$. Then, we have the following

\begin{align*}
    f_X(x) = c\int_0^\infty f_Z(tx) \cdot t^{p-1} dt \geq c\int_{\min(\{t:\ xt \in B_\epsilon(z)\cap [0,\infty)^p \setminus \boldsymbol{0}_p\})}^{\max(\{t:\ xt \in B_\epsilon(z)\cap [0,\infty)^p \setminus \boldsymbol{0}_p\})} f_Z(tx) \cdot t^{p-1} dt.
\end{align*}
Note that for any $z\in [0,\infty)^p \setminus \boldsymbol{0}_p$ and $x = g(z)$, $\{t:\ xt \in B_\epsilon(z)\cap [0,\infty)^p \setminus \boldsymbol{0}_p\}$ is an open interval. This, together with the fact that $f_Z(tx) > 0 $ and $t > 0$ for all $t \in \{t:\ xt \in B_\epsilon(z)\cap [0,\infty)^p \setminus \boldsymbol{0}_p\}$, imply that, $f_X(x)>0$ which proves Claim 1.

\paragraph{Claim 2: $\boldsymbol{\mathrm{Supp}{>0}(X) \subseteq g\left(\mathrm{Supp}{>0}(Z)\right)}$}

It is enough to show that for any point $x \in \Delta^{p-1}$ such that $f_X(x) > 0$, there exists $z \in [0, \infty)^p \setminus \boldsymbol{0}_p$ such that $f_Z(z) > 0$. Assume the contrapositive, i.e., $f_Z(z) = 0$ for all $z \in \mathbb{R}^n$ such that $x = g(z)$. Then, $\int_0^\infty f_Z(tx) t^{p-1} dt = 0$, which implies $f_X(x) = 0$, contradicting our assumption. Therefore, there must exist $z$ such that $f_Z(z) > 0$ and $x = g(z)$. This proves Claim 2 and, in turn, the Lemma.

\section{Application of Theorem \ref{thm: uniqueness of MB} for gene knockout experiments.}\label{sec:gene knockout experiments}
In \Cref{cor:factor covariates}, we considered regression problems with factor covariates where each observation is assigned a single level from each factor. We can generalize this setup where each observation receives multiple levels from a single factor. This arises frequently in gene knockout experiments where exactly $L$ out of $p$ genes are knocked out per individual, resulting in the sum of each row of the binary design matrix equal to $L$. For simplicity we consider only one factor with $L$ levels and define $\S$ and the nontrivial Markov boundary as in \Cref{subsec:the Markov boundary}.
\begin{corollary}[Uniqueness of nontrivial Markov boundary for gene knockout experiments]\label{cor:gene knockout experiments}
    Consider a regression problem where the covariate $X$ denotes levels of a single factor. If $\S = [p]$ then no nontrivial Markov boundary exists. Otherwise if $(Y,X)$ satisfies the following two assumptions:
    \begin{itemize}
        \item[(i)] For all $i, j\in \S^\mathsf{c}$, there exists a finite $t$ and a sequence $l_1, l_2, \dots, l_t$ such that $\H_{i,l_1}$, $\H_{l_1,l_2}$, $\dots$, $\H_{l_{t-1},l_t}$, $\H_{l_t,j}$ are all true,
        \item[(ii)] Positive probability is assigned to all vectors $x$ such that $\sum_{j\in [p]}X_{\{j\}} = L$,
    \end{itemize}
    then $\S$ is the unique nontrivial Markov boundary.
\end{corollary}

\begin{proof}[Proof of \Cref{cor:gene knockout experiments}]\label{sec:proof of corollary for gene knockout}
If $\S = [p]$, \Cref{thm: uniqueness of MB} proves there does not exist any nontrivial Markov boundary. Lets consider the other case, that is $|\S| < p$. By assumption \textit{(i)}, to prove \Cref{cor:gene knockout experiments} using \Cref{thm: uniqueness of MB}, using the same logic as in \Cref{appendix:proof-continuous-covariates}, it is enough to show that $\{A,B\}\in \Delta$ for any $A,B\subseteq [p]$ such that $|A\cap B|=1$ and $|B|=2$. So, we choose any such $A,B$ and consider the distribution of $X_{A\cup B}\mid X_{(A\cup B)^\mathsf{c}} = c$. Define $L_1 = L-||c||_1 = |X_{A\cup B}|$ and $L_2=L-L_1$. If $L_1=0$, then the support of $X_{A\cup B}\mid X_{(A\cup B)^\mathsf{c}} = c$ is equal to the set $\{\boldsymbol{0}_p\}$ which is path connected. Otherwise if $L_1 > 0$, consider any two points $x$ and $y$ in the support of $X_{A\cup B}\mid X_{(A\cup B)^\mathsf{c}} = c$. If $x_{B\setminus A}=y_{B\setminus A}$, then they are coordinatewise connected through $B\setminus A$ implying $\{A,B\}\in \Delta$. In the other case, when $x_{B\setminus A}\neq y_{B\setminus A}$, assume WLOG that $x_{B\setminus A} = 1$ and $y_{B\setminus A} = 0$. There are three cases based on the values of $x_{A\cap B}$ and $y_{A\cap B}$:
\begin{itemize}
    \item[Case 1: ]($x_{A\cap B}=y_{A\cap B}=0$)\\
    We introduce the following notation: for any vector $w \in \mathbb{R}^p$, we denote $w = (w_{A \setminus B}, w_{A \cap B}, w_{B \setminus A}, w_{(A\cup B)^\mathsf{c}})$. In the case $x_{A\cap B}=y_{A\cap B}=0$, define $z = (x_{A\setminus B}, 1, y_{B\setminus A}, c)$. $\sum_{j\in [p]}z_{\{j\}} = L$ which implies $z$ is in the support of $X_{A\cup B}\mid X_{(A\cup B)^\mathsf{c}} = c$ by assumption \textit{(ii)}. $x$ and $z$ are coordinatewise connected through $A\setminus B$ and $y$ and $z$ are coordinatewise connected through $B\setminus A$ proving the equivalence of $x$ and $y$ through $A\setminus B$ and $B\setminus A$.
    \item[Case 2: ]($y_{A\cap B}=1$)\\
    Define $z = (y_{A\setminus B}, 0, x_{B\setminus A},c)$. Similar to case 1, since $\sum_{j\in [p]}z_{\{j\}} = L$, assumption \textit{(ii)} implies that $z$ is in the support of $X_{A\cup B}\mid X_{(A\cup B)^\mathsf{c}} = c$. Also, $x$ and $z$ are coordinatewise connected through $B\setminus A$ and $y$ and $z$ are connected through $A\setminus B$ which shows that $x$ and $y$ are equivalent through $A\setminus B$ and $B\setminus A$. 
    \item[Case 3: ]($x_{A\cap B}=1$ and $y_{A\cap B}=0$)\\
    Since $\sum_{j\in A\setminus B}x_{\{j\}} - \sum_{j\in A\setminus B}y_{\{j\}}=2$, $x_{A\setminus B}$ must have at least two zeros. Define $x_{A\setminus B}^\prime$ as $x_{A\setminus B}$ with one additional 1 in place of a zero. Then if we define $x^{(1)}= (x_{A\setminus B}^\prime, 0, x_{B\setminus A},c)$, $\sum_{j\in [p]}z^{(1)}_{\{j\}}=L$ which implies $x^{(1)}$ belongs to the support of $X_{A\cup B}\mid X_{(A\cup B)^\mathsf{c}}=c$. Define another point $x^{(2)}=(x_{A\setminus B}^\prime, 1, 0, c)$. Again, entries of $x^{(2)}$ also add up to $L$ implying $x^{(2)}$ belongs to the support of $X_{A\cup B}\mid X_{(A\cup B)^\mathsf{c}}=c$ by assumption \textit{(ii)}. Finally, we note that the following:
    \begin{itemize}
        \item[-] $x$ and $x^{(1)}$ are coordinatewise connected through $B\setminus A$.
        \item[-] $x^{(1)}$ and $x^{(2)}$ are coordinatewise connected through $A\setminus B$.
        \item[-] $x^{(2)}$ and $y$ are coordinatewise connected through $B\setminus A$.
    \end{itemize}
    Together, they imply $x$ and $y$ are equivalent through $A\setminus B$ and $B\setminus A$. 
\end{itemize}
Therefore, in all the cases, $x$ and $y$ are equivalent through $A\setminus B$ and $B\setminus A$ which shows that $X_{A\cup B}\mid X_{(A\cup B)^\mathsf{c}}$ has only one equivalence class. Hence, $\{A,B\}\in \Delta$. This is true for any $A,B$ such that $|A\cap B| = 1$ and $|B|=2$. Hence, we can apply \Cref{thm: uniqueness of MB} which completes the proof of \Cref{cor:gene knockout experiments}.
\end{proof}

While \Cref{cor:gene knockout experiments} is stated for a single factor, one can extend it to the $L$-factor case similar to \Cref{cor:factor covariates}. 

\section{Pseudocode of Algorithms}

\subsection{FWER control}\label{app:fwer algs}
\begin{algorithm}[H]
  \caption{Method: (for FWER control under arbitrary dependence)}
  \label{Algorithm: adaptive-Bonferroni--Holm}
  \KwData{P-value matrix $P \in [0,1]^{p \times p}$, desired FWER threshold $\alpha$, and a strict upper bound on the number of non-nulls $\overline{s}$}
  \KwResult{Rejected indices $\hat{\mathcal{S}}$}
  
  \textbf{Initialize:} $\hat{\mathcal{S}} = \emptyset$
  \SetKwFor{For}{for}{:}{end}

  \For{$k \in [p]$}{
  
  \For{$j \in [p]$}{
    Define Bonferroni's p-value $P^{\mathrm{B}}_{j}(\hat{\mathcal{S}})$ as in \Cref{eq:corrected Bonferroni p-value}. 
  }
  
    Define $j^\star = \underset{j \in [p]}{\text{argmin}} P^{\mathrm{B}}_{j}(\hat{\mathcal{S}})$.
  
  \If{$P^{\mathrm{B}}_{j^\star}(\hat{\mathcal{S}}) \leq \frac{q}{p - k + 1}$}{
    Add $j^\star$ to $\hat{\mathcal{S}}$\;
  }
  \Else{
  \Return $\hat{\mathcal{S}}$.
  }
  \Return $\hat{\mathcal{S}}$.
  }
\end{algorithm}

\begin{algorithm}[H]
  \caption{Method: (for FWER control under positive dependence)}\label{Algorithm: adaptive-Simes-Holm}
  \KwData{P-value matrix $P \in [0,1]^{p \times p}$, desired FWER threshold $\alpha$, and a strict upper bound on the number of non-nulls $\overline{s}$}
  \KwResult{Rejected indices $\hat{\mathcal{S}}$}
  
  \textbf{Initialize:} $\hat{\mathcal{S}} = \emptyset$
  
  \SetKwFor{For}{for}{:}{end}

  \For{$k \in [p]$}{
  
  \For{$j \in [p]$}{
    Define Simes' p-value $P^{\mathrm{S}}_{j}(\hat{\mathcal{S}})$ as in \Cref{eq:corrected Simes p-value}. 
  }
  
    Define $j^\star = \underset{j \in [p]}{\text{argmin}} P^{\mathrm{S}}_{j}(\hat{\mathcal{S}})$.
  
  \If{$P^{\mathrm{S}}_{j^\star}(\hat{\mathcal{S}}) \leq \frac{q}{p - k + 1}$}{
    Add $j^\star$ to $\hat{\mathcal{S}}$\;
  }
  \Else{
  \Return $\hat{\mathcal{S}}$.
  }
  \Return $\hat{\mathcal{S}}$.
  }
\end{algorithm}

\subsection{FDR control}
\subsection{FDR control}
\begin{algorithm}[H]
  \caption{Detecting Markov boundary while Controlling FDR}\label{algorithm: Simes-BH-oracle} 
  \KwData{P-value matrix $P \in [0,1]^{p \times p}$, desired FDR threshold $\alpha$, and a strict upper bound on the number of non-nulls $\overline{s}$}
  \KwResult{Rejected indices $\hat{\mathcal{S}}$}
  \textbf{Initialize:} $\hat{\mathcal{S}} = \emptyset$

  \SetKwFor{For}{for}{:}{end}
  \For{$j \in [p]$}{
    Define Simes' p-value $P^{\mathrm{S}}_{j}$ as in \Cref{eq:Simes pch p-val}. 
  }
  
    Sort Simes' p-values in ascending order: $P^{\mathrm{S}}_{(1)} \leq P^{\mathrm{S}}_{(2)} \leq \dots \leq P^{\mathrm{S}}_{(p)}$\;

    Define $k^{\star} = \max\left\{k\in [p]:\ P^{\mathrm{S}}_{(k)} \leq \frac{k}{p} \cdot \alpha\right\}$\;

    Define $\hat{\mathcal{S}} = \left\{j : P^{\mathrm{S}}_{j} \in \left\{P^{\mathrm{S}}_{(1)}, \dots, P^{\mathrm{S}}_{(k^\star)}\right\}\right\}$\;
    
    \Return $\hat{\mathcal{S}}$\;
\end{algorithm}

\section{Additional simulation details}\label{appendix: additional simulations}

\subsection{Conditional independence testing for simulations}\label{appendix: using dCRT to obtain p-values}
In \Cref{sec: obtaining p-values}, we mentioned that any conditional independence test could be used to obtain the $P_{i,j}$'s. For our simulations, we use the distilled conditional randomization test (dCRT)~\citep{dCRT}, which is a computationally efficient special case of the conditional randomization test~\citep{candes2018panning}, to test the hypotheses $\H_{i,j}: Y\indep X_{\{i,j\}}\mid X_{\{i,j\}^\mathsf{c}}$ for $i,j\in [p]$. To describe how we implement the dCRT, we need to introduce some notation first. Let $\mathbb{Y}$ denote the vector of observed outcomes and let $\mathbb{X}$ denote the design matrix. $\mathbb{X}_{\{i,j\}}$ denotes the matrix of the $i^{\text{th}}$ and $j^{\text{th}}$ columns of $\mathbb{X}$ and $\mathbb{X}_{\{i,j\}^\mathsf{c}}$ denotes the matrix with all the columns of $\mathbb{X}$ except $i$ and $j$. 

The first step of the dCRT is to draw samples $\tilde{\mathbb{X}}_{\{i,j\}}^{(1)}, \ \tilde{\mathbb{X}}_{\{i,j\}}^{(2)}, \dots,\ \tilde{\mathbb{X}}_{\{i,j\}}^{(K)}$ independently from the distribution of $\mathbb{X}_{\{i,j\}}\mid \mathbb{X}_{\{i,j\}^\mathsf{c}}$. Since, in our simulations, rows of $\mathbb{X}$ consist of independent observations of the random variable $X$, each row of $\tilde{\mathbb{X}}_{\{i,j\}}^{(k)}$ can be sampled independently from the conditional distribution $X_{\{i,j\}}\mid X_{\{i,j\}^\mathsf{c}}$. 

Now we define the test statistics for dCRT. Let $\hat{\mathbb{Y}}\in \mathbb{R}^n$ denote the fitted values obtained from a 5-fold cross-validated LASSO fitting $\mathbb{Y}$ to $\mathbb{X}_{\{i,j\}^\mathsf{c}}$. Also, for any vector $y\in \mathbb{R}^n$ and matrix $z\in \mathbb{R}^{n\times 2}$, define $T(y,z)$ as the $R^2$ from the regression of $y$ on $\log(z)$. Then the dCRT p-value is defined as follows:
\begin{align}\label{eq: pval}
  P_{i,j} &= \frac{1}{K+1} \left(1 + \sum_{k=1}^K \1_{\left\{T\left(\mathbb{Y}-\hat{\mathbb{Y}}, \tilde{\mathbb{X}}_{\{i,j\}}^{(k)}\right) \geq T\left(\mathbb{Y}-\hat{\mathbb{Y}}, \mathbb{X}_{\{i,j\}}\right)\right\}}\right).
\end{align}
Under $\H_{i,j}$, the dCRT p-value $P_{i,j}$ has a super-uniform distribution because it is a special case of the conditional randomization test, which is proved to be valid in \citet{candes2018panning}.

\subsection{Dirichlet}\label{appendix: simulation dirichlet}
As mentioned in the previous subsection, we use thedCRT~\citep{dCRT} to test the base hypotheses $\H_{i,j}$. For each ${i,j}\subset [p]$, we calculate the p-value $P_{i,j}$ as described in \Cref{appendix: using dCRT to obtain p-values}. For the benchmark LOO methods, we select some $j\in [p]$ at random and for each $i\in \{j\}^\mathsf{c}$, we define $d_Y^{\{i\}}$ as the prediction from a 5-fold cross-validated LASSO regression of $Y$ on $X_{\{i,j\}^\mathsf{c}}$. Since $X_{\{j\}^\mathsf{c}}$ is non-compositional, we draw resamples from $X_{\{i\}}\mid X_{\{i,j\}^\mathsf{c}}$ (which follows a Beta distribution in this case). Then, we define the 
r-squared of the regression of $Y-d_Y^{\{i\}}$ on $\log(X_i)$ as the test statistic $T$ and calculate p-values as in \citet[Algorithm 1]{dCRT}. In this way, we only get p-values for all $i\in \{j\}^\mathsf{c}$.

\begin{figure}[ht]
\centering
\includegraphics[width=\textwidth]{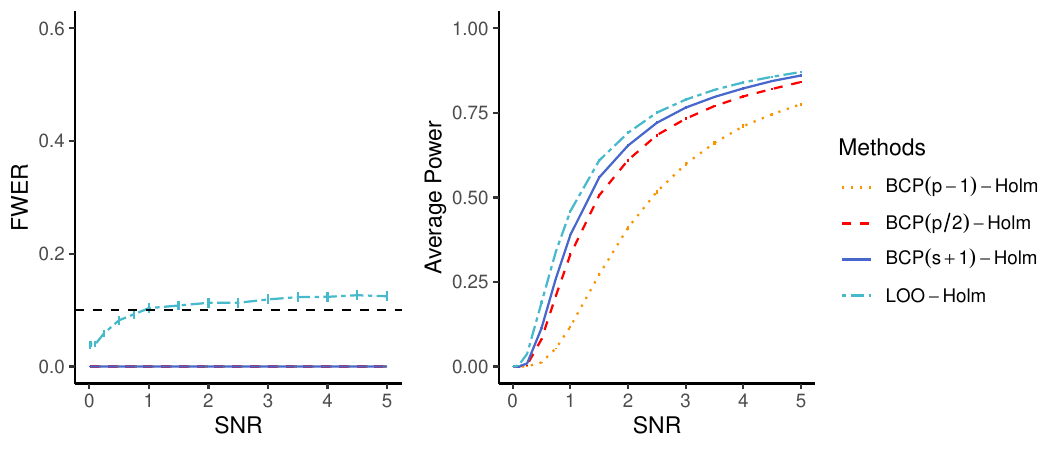}
    \caption{Comparison of FWER and average power for variable selection with Dirichlet covariates. The target FWER level is 10\% and error bars correspond to $\pm2$ Monte Carlo standard errors.}
    \label{fig:fwer_dirichlet}
\end{figure}
\Cref{fig:fwer_dirichlet} shows that like in the case of FDR control (\Cref{fig:fdr_dirichlet}), the proposed methods---BCP($p-1$)-Holm, BCP($p/2$)-Holm and BCP($s+1$)-Holm---control the FWER at a level much below the nominal level of 10\% while the benchmark method LOO-Holm violates the FWER level. In terms of average power, BCP($p/2$)-Holm and BCP($s+1$)-Holm achieve comparable average power to LOO-Holm.

\subsection{Logistic-normal}\label{appendix: simulation logistic-normal}
For this distribution, the conditional distribution of $X_{\{i,j\}}$ given $X_{\{i,j\}^\mathsf{c}}$ lacks a closed form, which is why we use MCMC to generate dCRT resamples. These resamples although not independent, are exchangeable, which is sufficient for validity. 
\begin{figure}[!]
    \centering
    \begin{subfigure}{\textwidth}
        \centering
\includegraphics[width=0.9\textwidth]{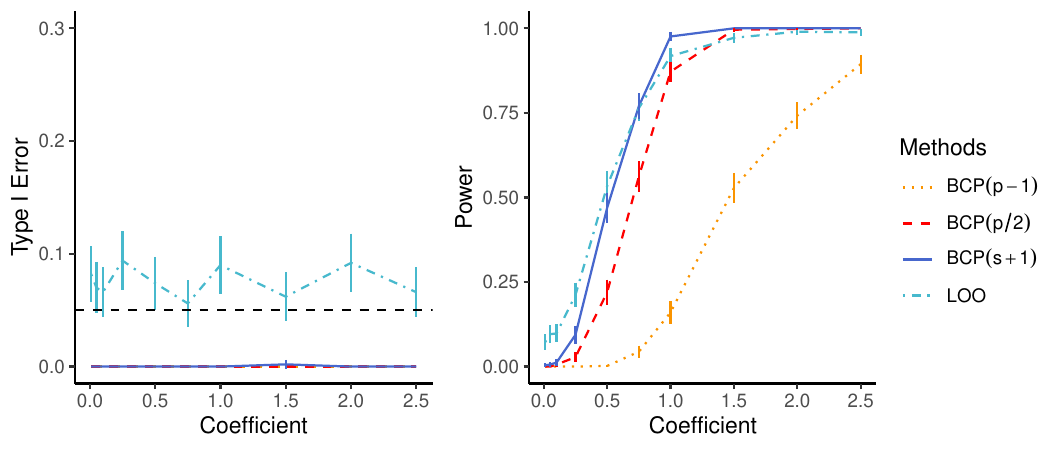}
        \caption{Comparison of type I error and power for single testing with Logistic-normal covariates. The target type I error is 5\% and error bars correspond to $\pm2$ Monte Carlo standard errors.}
\label{fig:single_test_logistic-normal}
    \end{subfigure}  
    \begin{subfigure}{\textwidth}
        \centering \includegraphics[width=0.9\textwidth]{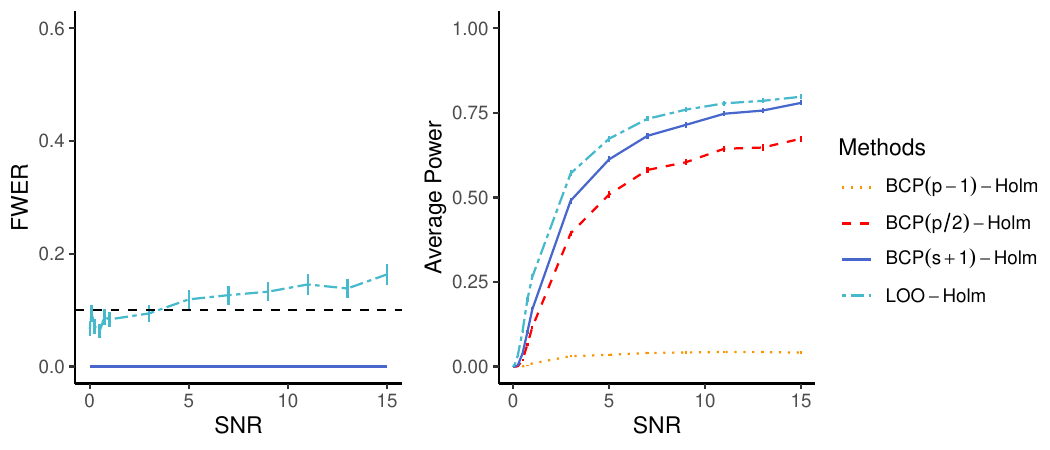}
        \caption{Comparison of FWER and average power for variable selection with Logistic-normal covariates. The target FWER level is 10\% and error bars correspond to $\pm2$ Monte Carlo standard errors.}\label{fig:fwer_logistic-normal}
    \end{subfigure}
    \begin{subfigure}{\textwidth}
        \centering
\includegraphics[width=0.9\textwidth]{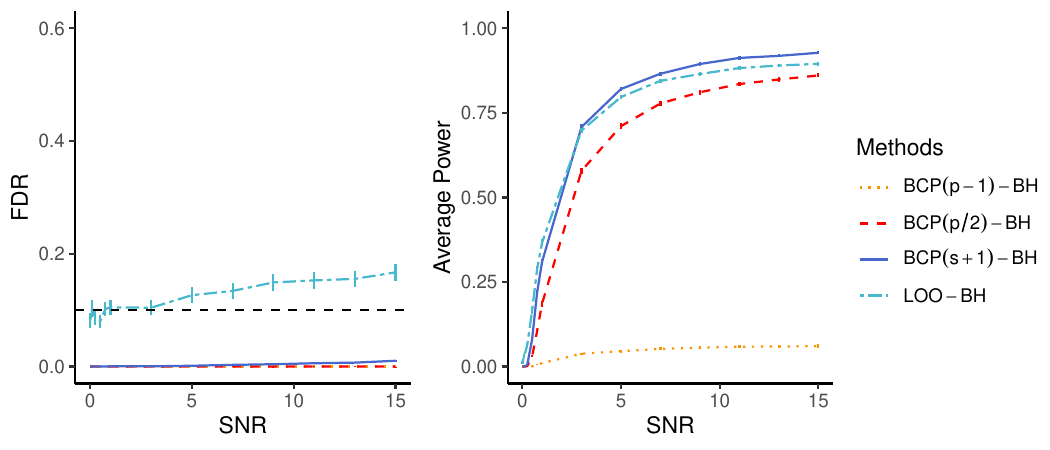}
        \caption{Comparison of FDR and average power for variable selection with Logistic-normal covariates. The target FDR level is 10\% and error bars correspond to $\pm2$ Monte Carlo standard errors.}\label{fig:fdr_logistic-normal}
    \end{subfigure}
    \caption{Comparison of methods for (a) single testing, (b) FWER control and (c) FDR control with Logistic-normal covariates.}
    \label{fig: logistic_normal}
\end{figure}
\Cref{fig: logistic_normal} shows that LOO methods fail to control respective error rates while the proposed methods control the same (\Cref{fig: logistic_normal}). BCP($p/2$) and BCP($s+1$) methods achieve power and average power similar to the benchmarks. In contrast to the Dirichlet case, BCP($p-1$) methods have considerably lower average power than the rest of the methods in multiple testing scenarios.

\subsection{Multivariate normal}\label{appendix: simulation multivariate normal}

\begin{figure}[ht]
\centering
\includegraphics[width=\textwidth]{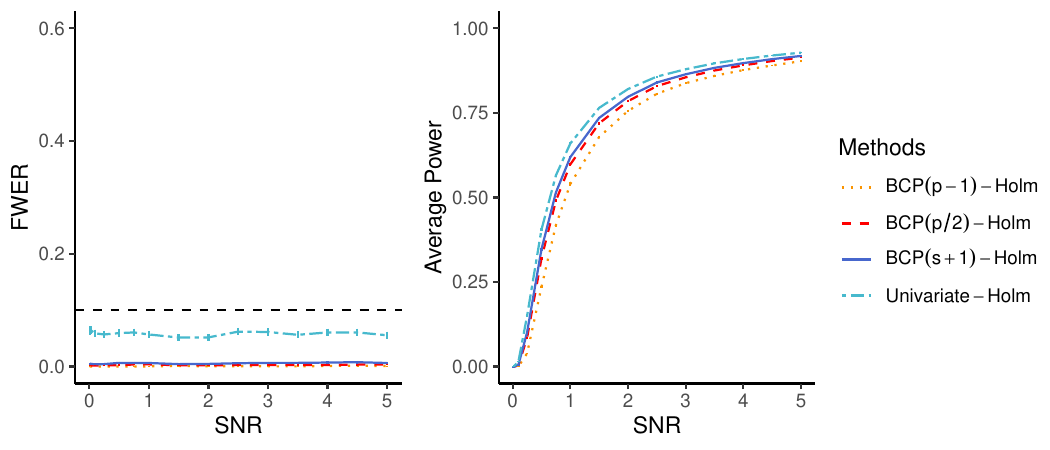} 
\caption{Comparison of FWER and average power for variable selection with multivariate normal covariates. The target FWER level is 10\% and error bars correspond to $\pm2$  Monte Carlo standard errors.}
\label{fig:fwer_normal}
\end{figure}

\Cref{fig:fwer_normal} shows that the three proposed methods achieve average power very close to the benchmark Univariate-Holm while controlling FWER at a level much lower than the nominal level 10\%.

\subsection{Studying the effect of conditioning on sparse covariates}\label{appendix: simulation effect of conditioning}
\Cref{fig:power_effect_of conditioning} illustrated that we can gain substantial power by conditioning out sparse covariates. As \Cref{fig:error_effect_of conditioning} shows, the error rates for the proposed methods remain well below the nominal levels, 5\% for single testing and 10\% for multiple testing, irrespective of the number of covariates we condition out.
\begin{figure}[ht]
\centering
\includegraphics[width=\textwidth]{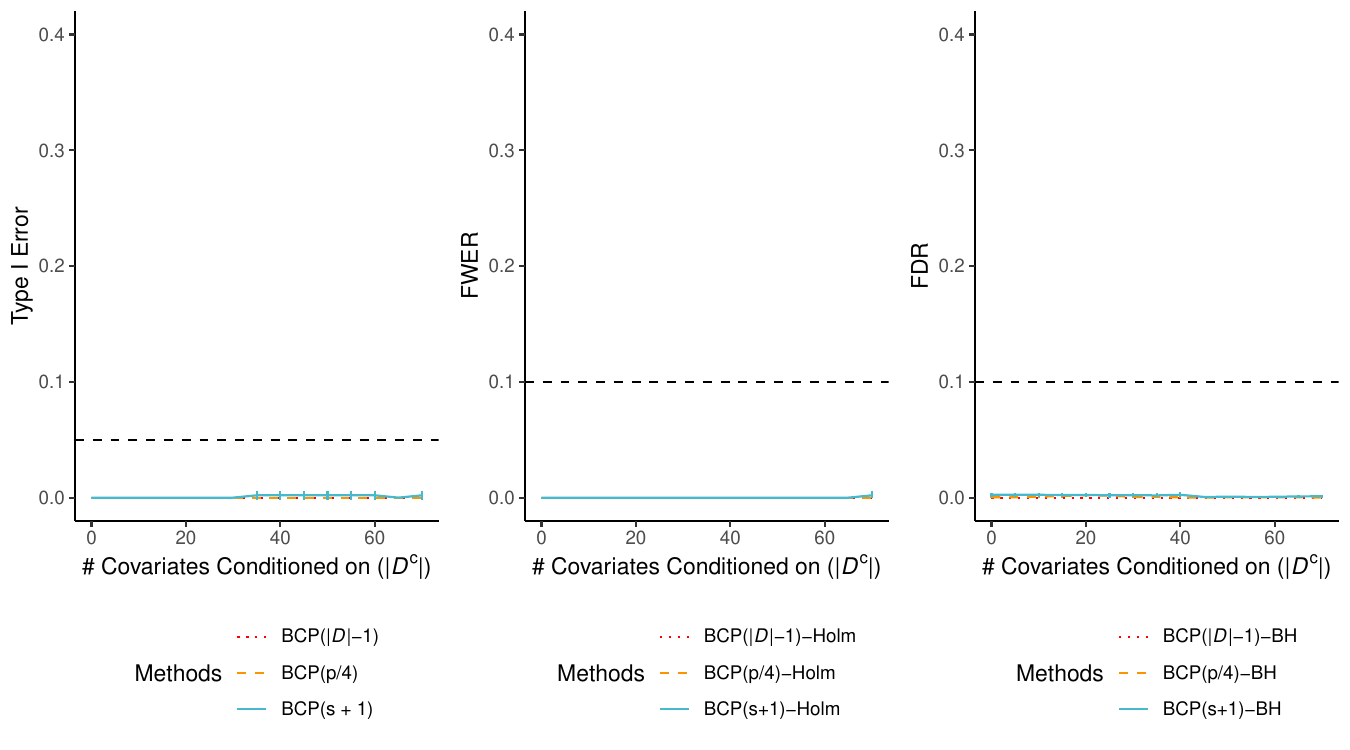}
\caption{Effect of conditioning out the sparse covariates on type I error (left), FWER (middle), and FDR (right), respectively. The target type I error, FWER, and FDR are 5\%, 10\%, and 10\%, respectively, and error bars correspond to $\pm2$  Monte Carlo standard errors.}
\label{fig:error_effect_of conditioning}
\end{figure}

\subsection{Robustness study}\label{subsec: robustness study}
In this simulation, we study the effect of conditioning out the sparse covariates. To do this, we generate the data from a Dirichlet distribution where the parameter vector is a vector of 2's of length 100. So, the distribution of $X_{\{i,j\}}$ given $X_{\{i,j\}^\mathsf{c}}$ is a scaled beta distribution which can be used in the dCRT step. However, here we consider the scenario where the parameters for the Dirichlet distribution are not known. Instead, we estimate the parameters of the Dirichlet distribution using the $\texttt{fit\_dirichlet}$ function~\citep{minka2000estimating} in $\texttt{R}$. Figures \ref{fig:fwer dirichlet_estimated} and \ref{fig:fdr dirichlet_estimated} show that estimating parameters has little impact on the average power and the error rates of the proposed methods BCP($p-1$), BCP($p/2$), and BCP($s+1$).

\begin{figure}[!]
    \centering
    \begin{subfigure}{\textwidth}
        \centering
\includegraphics[width=0.9\textwidth]{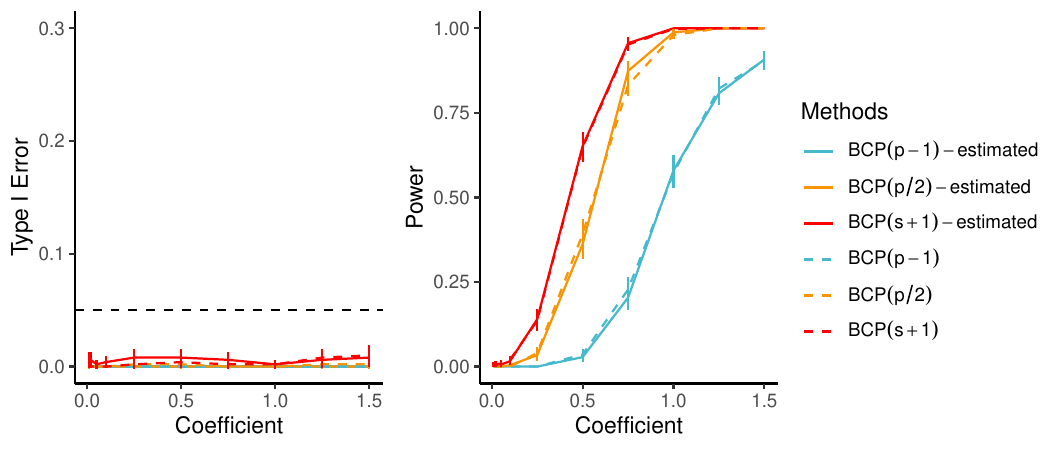}
        \caption{Effect of sampling from an estimated distribution on single testing methods. The target type I error is 5\% and error bars correspond to $\pm 2$ Monte Carlo standard errors.}
\label{fig:single_test_dirichlet_estimated}
    \end{subfigure}  
    
    \begin{subfigure}{\textwidth}
        \centering
        \includegraphics[width=0.9\textwidth]{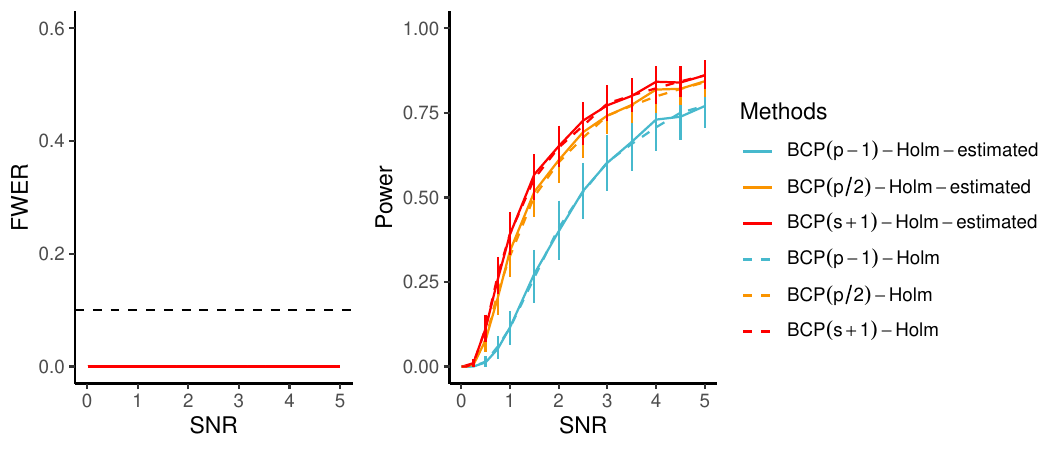}
        \caption{Effect of sampling from an estimated distribution on FWER control methods. The target FWER level is 10\% and error bars correspond to $\pm 2$ Monte Carlo standard errors.}
        \label{fig:fwer dirichlet_estimated}
    \end{subfigure}

    \begin{subfigure}{\textwidth}
        \centering
        \includegraphics[width=0.9\textwidth]{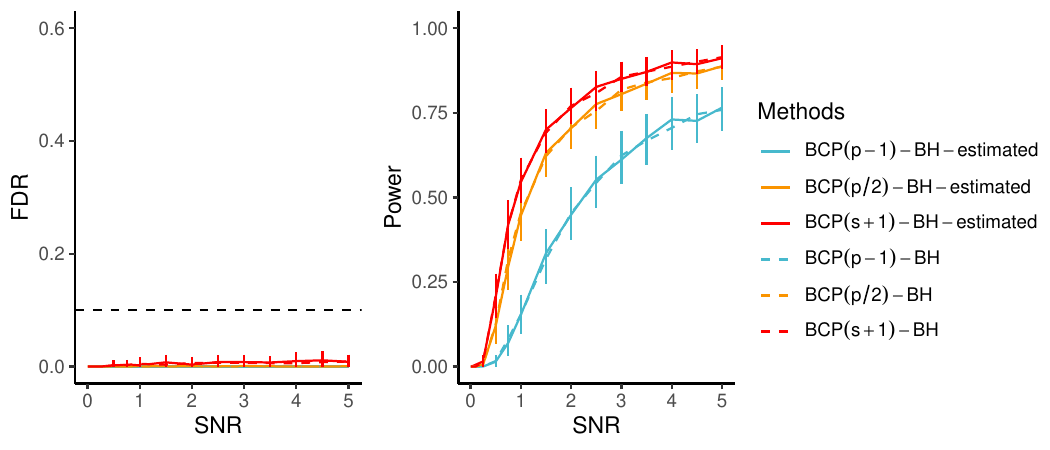}
        \caption{Effect of sampling from an estimated distribution on FDR control methods. The target FDR level is 10\% and error bars correspond to $\pm 2$ Monte Carlo standard errors.}
        \label{fig:fdr dirichlet_estimated}
    \end{subfigure}

    \caption{Simulation studies showing the effects of sampling from an estimated Dirichlet distribution on (a) FWER and (b) FDR control methods. 
    }
    \label{fig:dirichlet_estimated_combined}
\end{figure}

\subsection{Bonferroni p-values}\label{appendix: simulation bonferroni}
All of the plots in \Cref{sec: simulations} used PCH p-values combined via Simes' approach. In this section, we present the same results for Bonferroni's approach. 
\begin{figure}[ht]
\centering
\begin{subfigure}{\textwidth}
    \centering
    \includegraphics[width=0.85\textwidth]{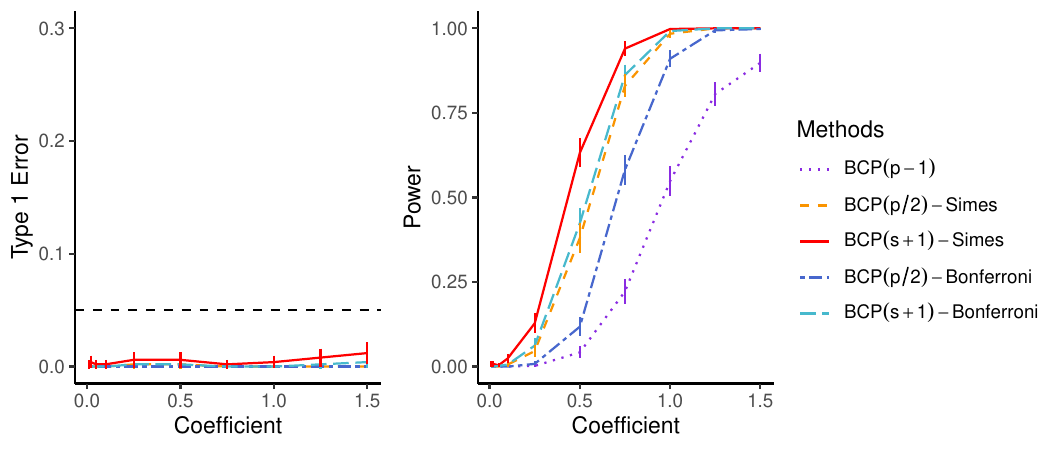}
    \caption{Comparison of type I error and power for single testing with Dirichlet covariates. The target type I error is 5\% and error bars correspond to $\pm2$ Monte Carlo standard errors.}
    \label{fig: single test bonferroni dirichlet}
\end{subfigure}

\begin{subfigure}{\textwidth}
    \centering
    \includegraphics[width=0.85\textwidth]{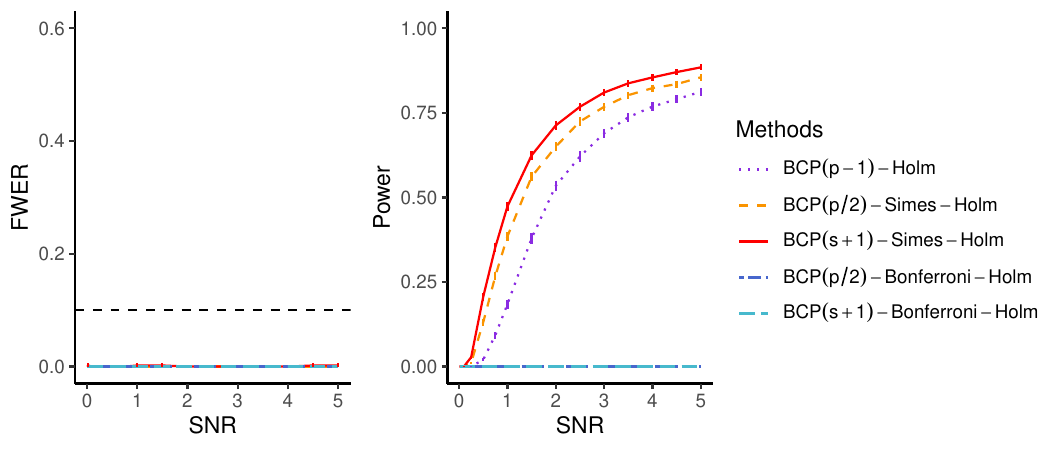}
    \caption{Comparison of FWER and average power for variable selection with Dirichlet covariates. The target FWER level is 10\% and error bars correspond to $\pm2$ Monte Carlo standard errors.}
    \label{fig: fwer bonferroni dirichlet}
\end{subfigure}

\begin{subfigure}{\textwidth}
    \centering
    \includegraphics[width=0.85\textwidth]{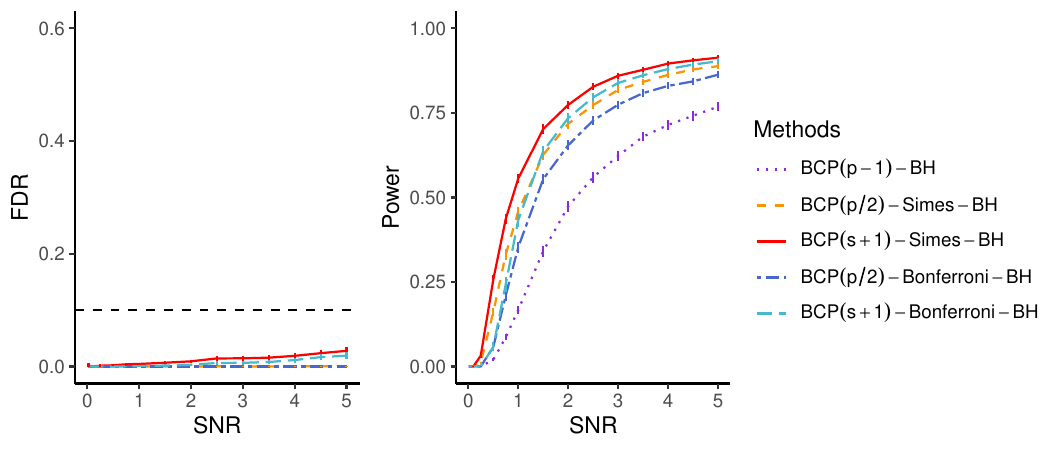}
    \caption{Comparison of FDR and average power for variable selection with Dirichlet covariates. The target FDR level is 10\% and error bars correspond to $\pm2$ Monte Carlo standard errors.}
    \label{fig: fdr bonferroni dirichlet}
\end{subfigure}

\caption{Comparison of BCP$(\overline{s})$  with Simes and Bonferroni PC p-values.}
\label{fig: bonferroni dirichlet}
\end{figure}

\begin{figure}[ht]
\centering
\begin{subfigure}{\textwidth}
    \centering
    \includegraphics[width=0.85\textwidth]{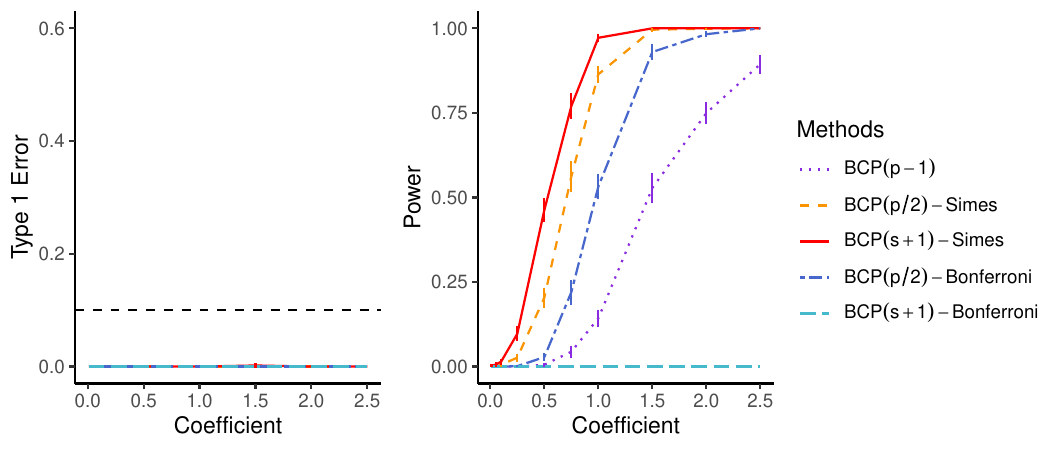}
    \caption{Comparison of type I error and power for single testing with Logistic-normal covariates. The target type I error is 5\% and error bars correspond to $\pm2$ Monte Carlo standard errors.}
    \label{fig: single test bonferroni logistic-normal}
\end{subfigure}

\begin{subfigure}{\textwidth}
    \centering
    \includegraphics[width=0.85\textwidth]{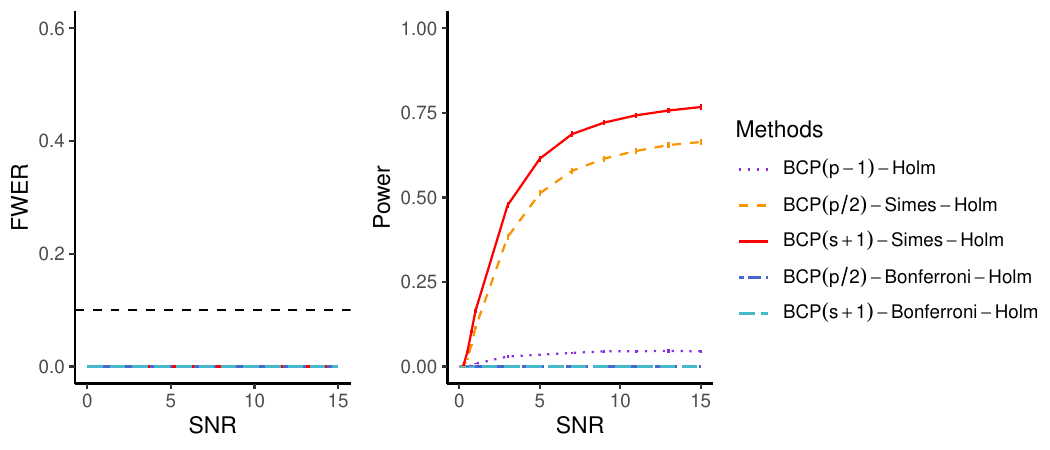}
    \caption{Comparison of FWER and average power for variable selection with Logistic-normal covariates. The target FWER level is 10\% and error bars correspond to $\pm2$ Monte Carlo standard errors.}
    \label{fig: fwer bonferroni logistic-normal}
\end{subfigure}

\begin{subfigure}{\textwidth}
    \centering
    \includegraphics[width=0.85\textwidth]{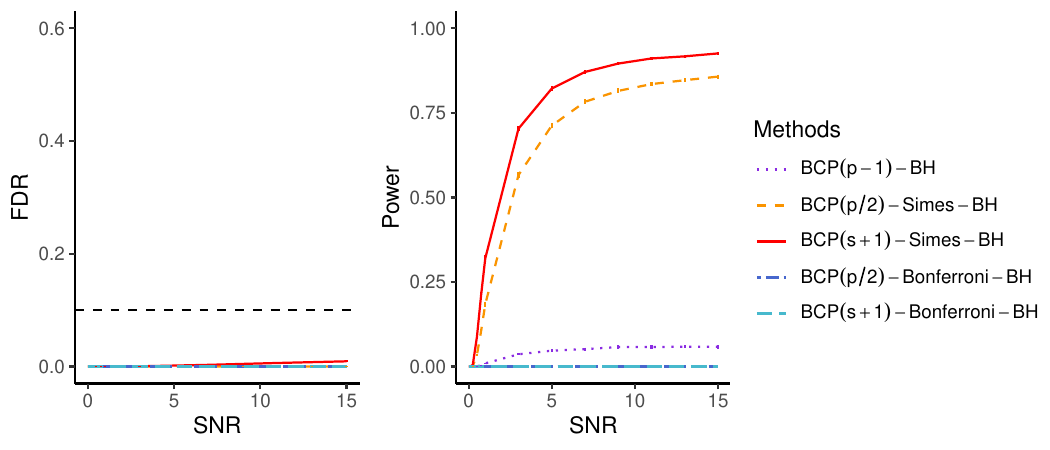}
    \caption{Comparison of FDR and average power for variable selection with Logistic-normal covariates. The target FDR level is 10\% and error bars correspond to $\pm2$ Monte Carlo standard errors.}
    \label{fig: fdr bonferroni logistic-normal}
\end{subfigure}

\caption{Comparison of BCP$(\overline{s})$  with Simes and Bonferroni PC p-values for Logistic-normal covariates.}
\label{fig: bonferroni logistic-normal}
\end{figure}

\begin{figure}[ht]
\centering
\begin{subfigure}{\textwidth}
    \centering
    \includegraphics[width=0.85\textwidth]{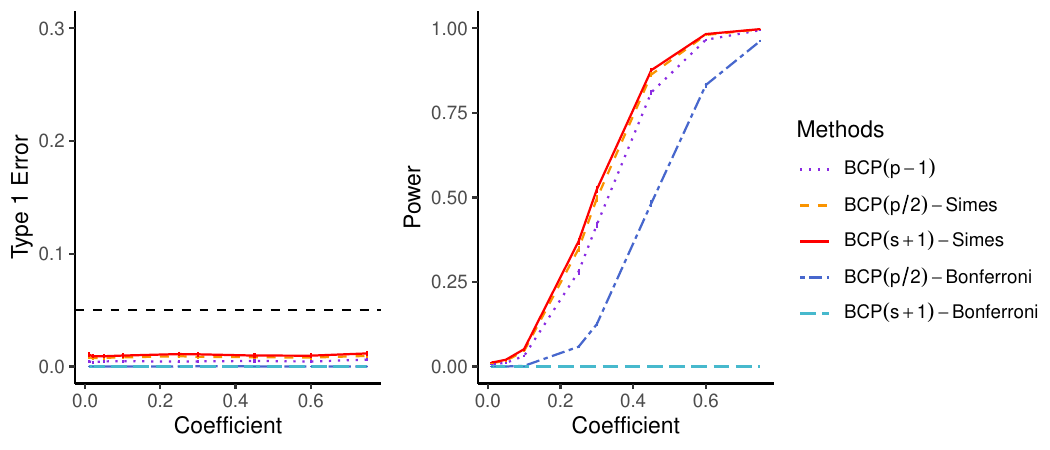}
    \caption{Comparison of type I error and power for single testing with multivariate normal covariates. The target type I error is 5\% and error bars correspond to $\pm2$ Monte Carlo standard errors.}
    \label{fig: single test bonferroni normal}
\end{subfigure}

\begin{subfigure}{\textwidth}
    \centering
    \includegraphics[width=0.85\textwidth]{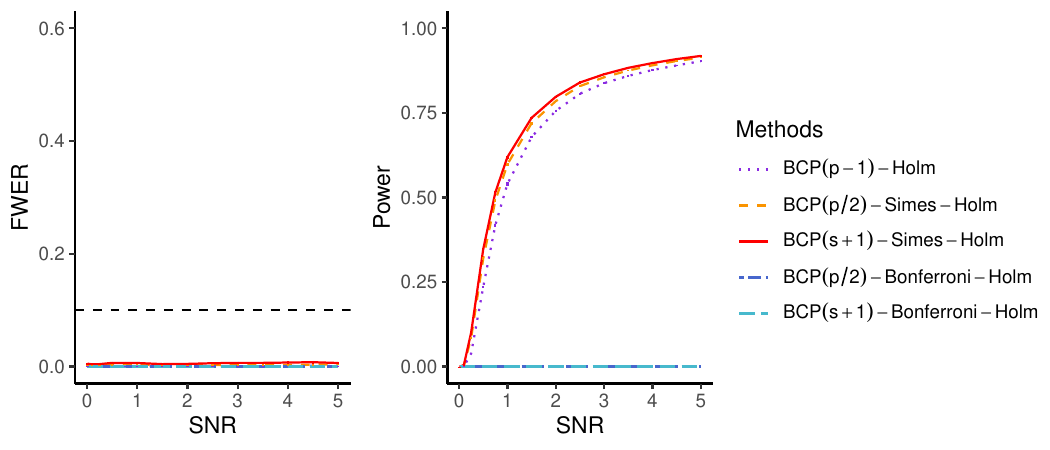}
    \caption{Comparison of FWER and average power for variable selection with multivariate normal covariates. The target FWER level is 10\% and error bars correspond to $\pm2$ Monte Carlo standard errors.}
    \label{fig: fwer bonferroni normal}
\end{subfigure}

\begin{subfigure}{\textwidth}
    \centering
    \includegraphics[width=0.85\textwidth]{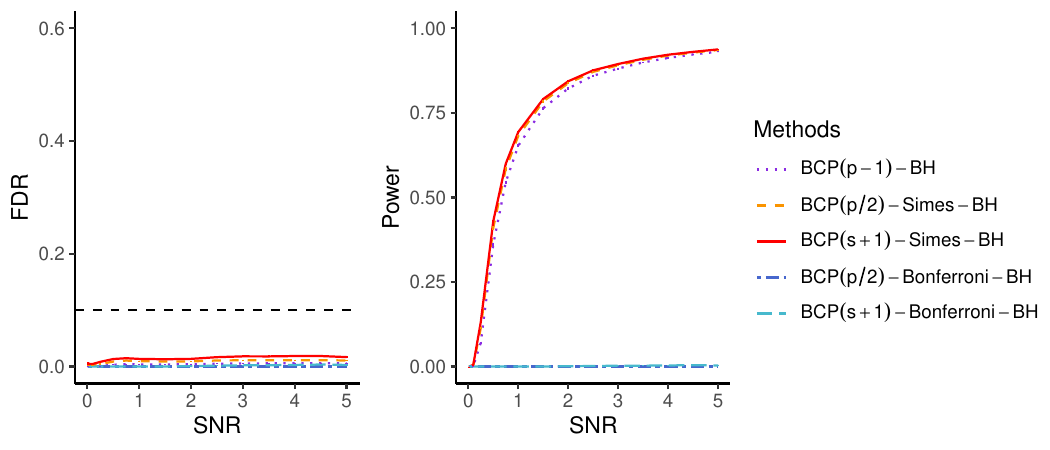}
    \caption{Comparison of FDR and average power for variable selection with multivariate normal covariates. The target FDR level is 10\% and error bars correspond to $\pm2$ Monte Carlo standard errors.}
    \label{fig: fdr bonferroni normal}
\end{subfigure}

\caption{Comparison of BCP$(\overline{s})$  with Simes and Bonferroni PC p-values for multivariate normal covariates.}
\label{fig: bonferroni normal}
\end{figure}

Figures \ref{fig: bonferroni dirichlet}, \ref{fig: bonferroni logistic-normal}, and \ref{fig: bonferroni normal} show a considerable decrease in power and average power of BCP$(\overline{s})$, BCP$(\overline{s})$-Holm, and BCP$(\overline{s})$-BH methods when we use Bonferroni p-values instead of Simes in the cases of Dirichlet, Logistic-normal, and multivariate normal covariates respectively. In figures \ref{fig: bonferroni logistic-normal} and \ref{fig: bonferroni normal} we see that BCP methods with Bonferroni p-values almost always have zero power or average power. One reason for this is that 1,500 dCRT resamples are insufficient to achieve non-zero power at these SNR values when using Bonferroni p-values, unlike when Simes p-values are used. In order to confirm that, we have increase the number of dCRT resamples to 25,000 in \Cref{fig: bonferroni dirichlet} and we observe that more BCP methods achieve non-zero (average) power using Bonferroni p-values. We also observe from \Cref{fig: fwer bonferroni dirichlet} that unlike the Simes p-values, the Bonferroni p-values are not monotone with respect to $\overline{s}$. BCP$(p-1)$-Bonferroni-Holm has non zero average power while BCP$(p/2)$-Bonferroni-Holm and BCP$(s+1)$-Holm both have zero average power in the SNR range used.

While a smaller $\overline{s}$ should provide more information and yield smaller p-values, Bonferroni’s p-values lack monotonicity with respect to $\overline{s}$. To illustrate, consider a simple example: for a particular $\H_{0j}^{\overline{s}}$ with $p=4$ and $s=1$, suppose the base p-values are ${P_{\{j\}^\mathsf{c},j}} = (0.2,0.25,0.3)$. Bonferroni’s PCH p-value for $\overline{s}=1$ is $P_{1 / 3}^\mathrm{B} = 3 \times 0.2 = 0.6$, while for $\overline{s}=2$, it is $P_{2 / 3}^\mathrm{B} = 2 \times 0.25 = 0.5$. Thus, the p-value for $\overline{s}=1$ is higher than for $\overline{s}=2$, contrary to the expectation that a smaller $\overline{s}$ would result in a lower p-value. Unlike Bonferroni’s p-values $P_{r / n}^\mathrm{B}$, Simes' p-values satisfy monotonicity with respect to $\overline{s}$.

\begin{figure}[ht]
\centering
\begin{subfigure}{\textwidth}
    \centering
    \includegraphics[width=0.85\textwidth]{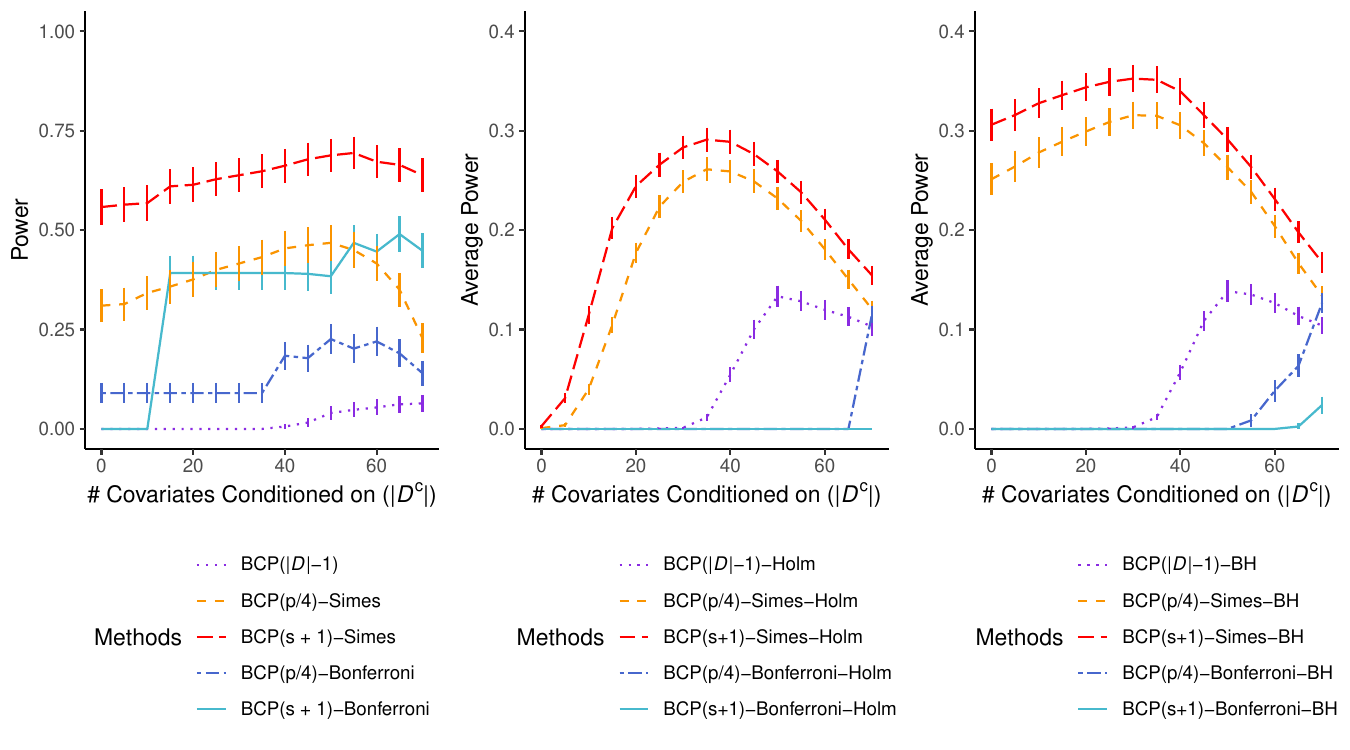}
    \caption{Effect of conditioning out the sparse covariates on power for single testing (left) and average power for FWER control (middle) and FDR control (right), respectively. Error bars correspond to ±2 Monte Carlo standard errors.}
    \label{fig: bonferroni power effect of conditioning}
\end{subfigure}

\begin{subfigure}{\textwidth}
    \centering
    \includegraphics[width=0.85\textwidth]{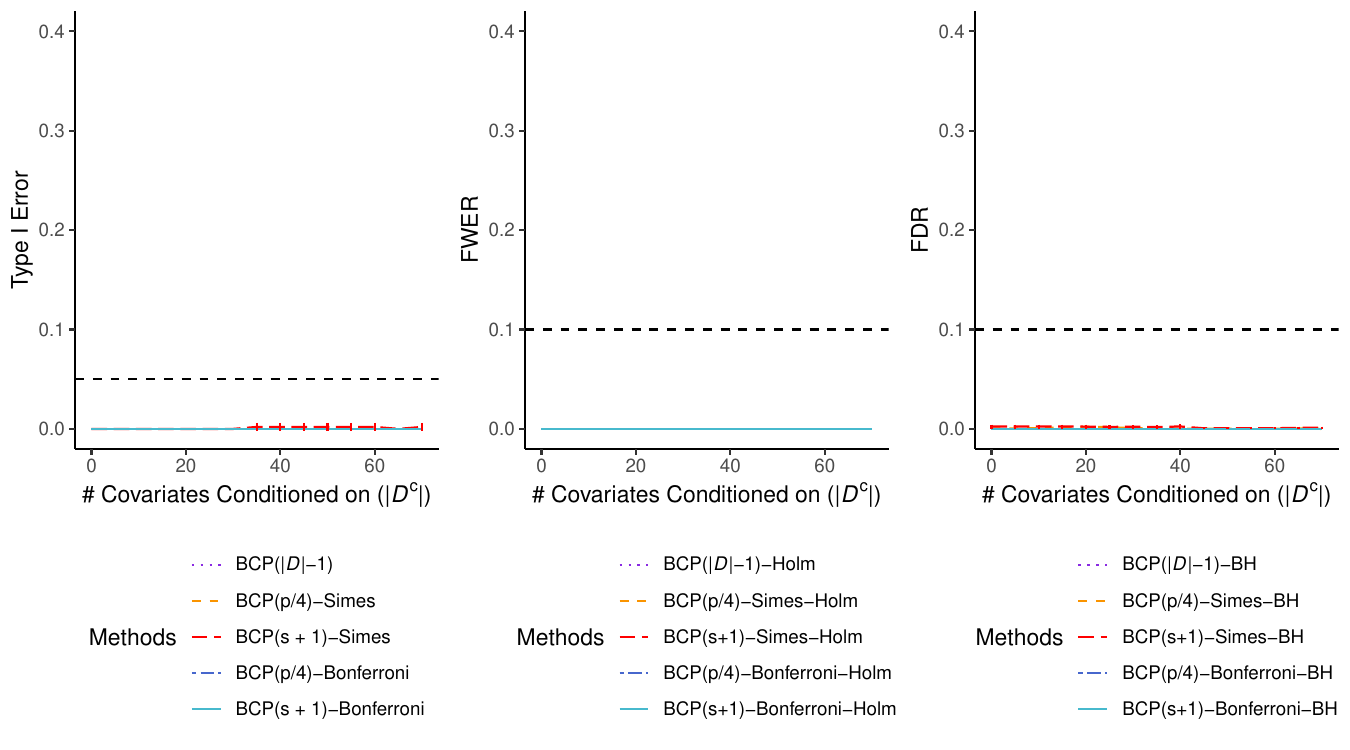}
    \caption{Effect of conditioning out the sparse covariates on type I error (left), FWER (middle), and FDR (right), respectively. The target type I error, FWER, and FDR are 5\%, 10\%, and 10\%, respectively, and error bars correspond to ±2 Monte Carlo standard errors.}
    \label{fig: bonferroni error effect of conditioning}
\end{subfigure}

\caption{Comparison of effect of conditioning out sparse covariates on BCP$(\overline{s})$  with Simes and Bonferroni PC p-values.}
\label{fig: bonferroni effect of conditioning}
\end{figure}

\Cref{fig: bonferroni effect of conditioning} shows that the power and average power of BCP methods using Bonferroni p-values increases when we condition out more and more sparse covariates. In fact, when we do not condition out any covariate, all BCP methods have zero power in both single and multiple testing problems.

\subsection{Performance of AdaFilter}\label{simulation adafilter}
AdaFilter~\citep{adafilter} is a method designed for testing multiple PCHs. Standard PCH p-values can be quite conservative which reduces power in detecting weak signals. AdaFilter involves a filtering step which reduces the potential detected signals to a smaller set of more promising hypotheses, giving it an edge over traditional methods in terms of average power. Since we also combine multiple PCH tests, AdaFilter may seem like a natural tool to use in our framework. However, AdaFilter requires independence of the base p-values for validity and as discussed in the paragraph after \Cref{thm: FWER control}, the base p-values will typically be far from independent in our framework. 

\begin{figure}[!]

\begin{subfigure}{\textwidth}
\centering
\includegraphics[width=\textwidth]{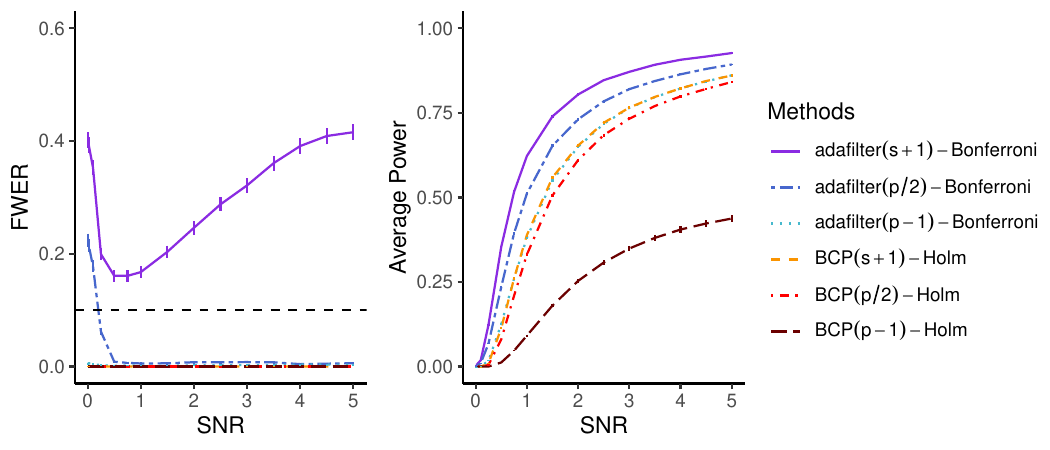} 
\caption{FWER and average power of AdaFilter for variable selection with Dirichlet covariates. The target FWER level is 10\% and error bars correspond to $\pm 2$ Monte Carlo standard errors.}
\label{fig: FWER adafilter}
\end{subfigure}

\begin{subfigure}{\textwidth}
\includegraphics[width=\textwidth]{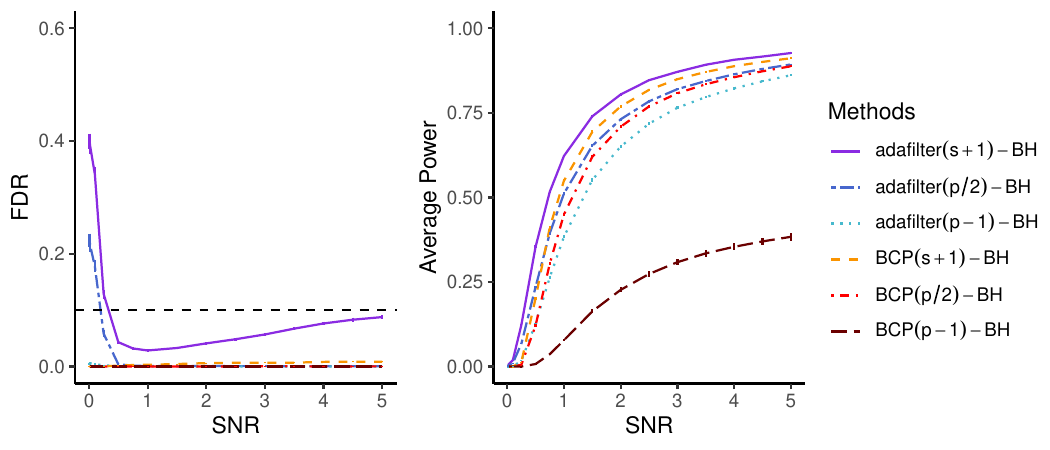}
\caption{FDR and average power of AdaFilter for variable selection with Dirichlet covariates. The target FDR level is 10\% and error bars correspond to $\pm 2$ Monte Carlo standard errors.}
\label{fig: FDR adafilter}
\end{subfigure}

\caption{Performance of AdaFilter on Dirichlet data in controlling (a) FWER and (b) FDR. 
}
\label{fig:adafilter_performance}
\end{figure}

In \Cref{fig:adafilter_performance} we denote AdaFilter--BH and AdaFilter--Bonferroni methods with $r=\overline{s}$~\cite[Section 3.1]{adafilter} as adafilter($\overline{s}$)-BH and adafilter($\overline{s}$)-Bonferroni respectively. \Cref{fig:adafilter_performance} shows that, in the Dirichlet data setting (see \Cref{subsec: comp-data-simulation}), AdaFilter fails to control the FDR and the FWER at the nominal level (10\%) particularly in the low SNR regime, while the proposed methods, BCP($\overline{s}$)-BH and BCP($\overline{s}$)-Holm, effectively control both the FDR and Familywise Error Rate (FWER) well below the nominal level of 10\% while having comparable average power. In fact, for FWER control, adafilter($s+1$)-Bonferroni fails to control the FWER anywhere in our SNR range. 
So, AdaFilter is not a suitable method for our problem. The reason, as we stated in the previous paragraph, is the strong dependence among the p-values. 

\subsection{Computational speedups}\label{appendix: simulation computational speedups}
In \Cref{sec: computational speedups}, we discussed two ways of significantly speeding up the proposed method. In this section, we look at simulation studies quantifying said speedups. 

But first, we mention one further computational improvement we considered that is specific to the dCRT (in fact it would apply to any conditional randomization test), which is the conditional independence testing approach we used in our simulations to calculate the $ P_{i,j}$. 
In the dCRT, resamples $\tilde{X}_{\{i,j\}}$ are drawn from the conditional distribution of $ X_{\{i,j\}}\mid X_{\{i,j\}^\mathsf{c}} $, and a test statistic $ T(y, \tilde X_{\{i,j\}}, X_{\{i,j\}^\mathsf{c}}) $ is computed for each CRT resample to be compared with $ T(y, X_{\{i,j\}} ,X_{\{i,j\}^\mathsf{c}})$. However, many resamples (we denote the number of resamples by $K$) are needed for a powerful test, especially for variable selection where multiple testing corrections are applied. 
For faster computation, one can initially draw a smaller number (e.g., $K/10$) of resamples and compute an initial p-value based on just these resamples. If that initial p-value is above, say, 0.1, then it is very unlikely that more resamples will lead to a very small p-value that would provide strong evidence against any null hypothesis, and hence we could just set the p-value to 1 without further resampling in this case. If the initial p-value is small, then we would continue resampling to the full number $K$ and compute the p-value as normal.

As mentioned in the main text, all three proposed computational speedups take the form of data-dependent screening: they replace individual p-values by 1 in the absence of evidence of them being non-null. This makes the resulting p-values strictly superuniform when both $i$ and $j$ are null variables. While this proves validity of FWER control method as in \Cref{Algorithm: adaptive-Simes-Holm} defined using Bonferroni's combining function, proving FDR control is difficult. This is because the screening procedures can introduce complex dependence among the p-values, potentially breaking the PRDS structure. 

To assess the statistical impact of these computational speedups, we perform a simulation study where we consider the same setup as in model 1 of \Cref{subsec: comp-data-simulation}. The setup consists of 100 observations of the form $(Y,X)$ where $X$ has 100 columns and each row of $X$ is sampled independently from a Dirichlet distribution, making $X$ compositional. 
\begin{figure}[!]

\begin{subfigure}{\textwidth}
\centering
\includegraphics[width=\textwidth]{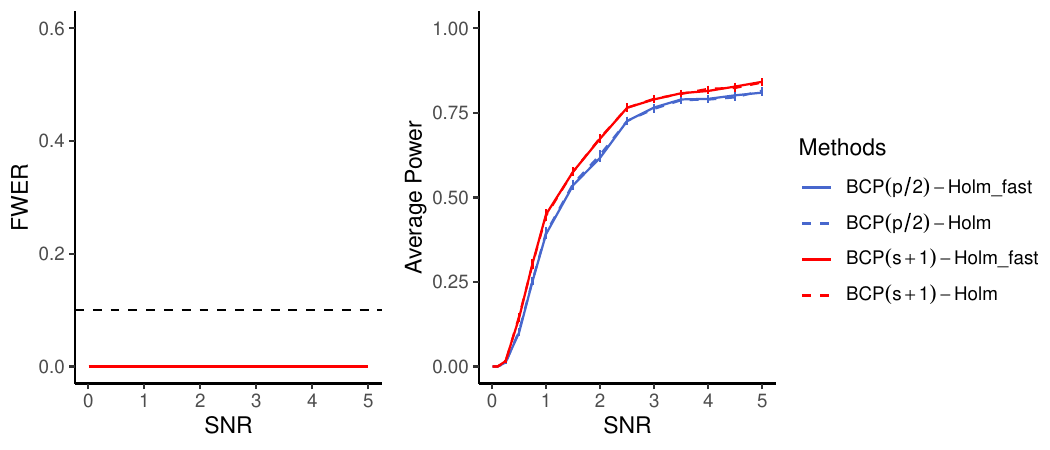} 
\caption{Comparison of FWER and average power for variable selection with computational speedups. The target FWER level is 10\% and error bars correspond to $\pm 2$ Monte Carlo standard errors.}
\end{subfigure}

\begin{subfigure}{\textwidth}
\centering
\includegraphics[width=\textwidth]{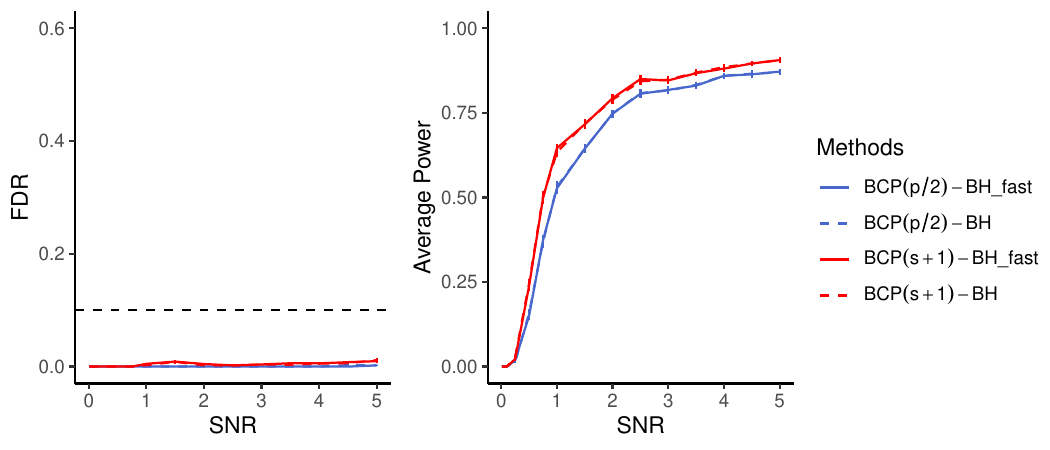} 
\caption{Comparison of FDR and average power for variable selection with computational speedups.  The target FDR level is 10\% and error bars correspond to $\pm 2$ Monte Carlo standard errors.}
\end{subfigure}
\caption{Comparison of proposed method with computational speedups discussed in \Cref{sec: computational speedups} for variable selection while controlling (a) FWER and (b) FDR control. 
}
\label{fig:computational speedups}
\end{figure}
As \Cref{fig:computational speedups} shows, the proposed methods with the speedups effectively maintain the False Discovery Rate (FDR) and Family-Wise Error Rate (FWER) well below the designated threshold of 10\%. Notably, the BCP($\overline{s}$) procedures, coupled with computational improvements, demonstrate average power levels almost identical to their counterparts without these speedups. Moreover, these optimizations reduce the average computation time for the BCP($\overline{s}$) methods with BH to 3.22 minutes, compared to 15.55 minutes without the speedups, yielding an approximate 4.8-fold decrease in computation time. The runtimes were calculated using a single core of a second-generation Intel Xeon processor. It is also observed that the average computation time increases with the SNR level. Specifically, the average runtimes for SNR values of 0.25, 0.5, 1, 2, and 4 are 1.4, 2.3, 3.3, 3.5, and 3.8 minutes, respectively for BCP($\overline{s}$)-BH with the computational speedups. In contrast, the run times for methods without the speedups decrease with SNR, with corresponding times of 19.1, 17.1, 14.9, 13.5, and 12.9 minutes.

While we get approximately 5 times speedup for a $100\times 100$ covariate matrix, the computational gains are much more pronounced for higher dimensional datasets. To illustrate this we ran a similar experiment with $1000\times 1000$ dimensional covariate matrix. In this case, we ran the proposed methods, both with and without the speedups, on a 24-core Intel Xeon CPU parallelly. The average time taken for the proposed method without any speedups was 23.99 hours while that of the method with them was 42.22 minutes. So, we achieve around 34 fold increase in speed. This shows the effectiveness of these computational improvements in real world datasets where covariates are often high dimensional. Another important point to note here is that, in this simulation, $X$ contains no sparse columns. When $X$ contains sparse columns, one can substantially speed up our method using the technique mentioned in \Cref{subsec: sparse covariates conditioning method}.

\end{document}